\documentclass[aps,pre,reprint,superscriptaddress,nofootinbib]{revtex4-2}

\usepackage{xcolor}       % colors
\usepackage{hyperref}	  % links and bookmarks
\hypersetup{colorlinks=true,linkcolor=blue,citecolor=blue,breaklinks={true}} %hyperref configuration
\usepackage[normalem]{ulem}

\usepackage{graphicx}
\graphicspath{{./FIG/}}
\DeclareGraphicsExtensions{.eps,.png,.jpg,.jpeg,.bmp,.gif,.pdf}

\usepackage{tabularx,booktabs}	% use more thick toplines and bottom lines
\usepackage{multirow}	        % multiple columns in one cell
\usepackage{array}
\newcolumntype{?}{!{\vrule width 2pt}}	% thick vertical line
\usepackage{amsmath}      		% equations
\usepackage{amssymb}	  		% symbols
\usepackage{cancel}
\usepackage{bm}		      		% bold symbols
\usepackage{physics}	  		% derivatives
\usepackage{mathtools}			% :=
\usepackage{dsfont}		  		% use \mathds{1} for indicator function
\usepackage{soul}				% strikethrough command (\st)
\usepackage{relsize}			% for large summation, integral, etc symbols
\usepackage{arydshln}           % for dashed lines in matrices

\newcommand{\N}{\mathbb{N}}		% natural numbers
\newcommand{\R}{\mathbb{R}}		% real numbers
\newcommand{\transp}{\mathsf{T}}					% transpose
\usepackage{amsthm}
\newtheorem{theor}{Theorem}

\theoremstyle{definition}
\newtheorem{defin}{Definition}

\newcommand*{\xdash}[1][3em]{\rule[0.5ex]{#1}{0.55pt}}

\begin{document}

\title{Functional observability and subspace reconstruction in nonlinear systems}

\author{Arthur~N.~Montanari}
\email{arthur.montanari@uni.lu}
\affiliation{Luxembourg Centre for Systems Biomedicine, University of Luxembourg, 7 avenue des Hauts-Fourneaux, L-4362 Esch-sur-Alzette, Luxembourg}

\author{Leandro~Freitas}
\affiliation{Department of Industrial Automation and Information Technology, Instituto Federal de Educa\c{c}\~ao, Ci\^encia e Tecnologia de Minas Gerais, Campus Betim, Rua Itagua\c cu 595, 32677-562 Betim, MG, Brazil}

\author{Daniele Proverbio}
\affiliation{Luxembourg Centre for Systems Biomedicine, University of Luxembourg, 7 avenue des Hauts-Fourneaux, L-4362 Esch-sur-Alzette, Luxembourg}
\affiliation{College of Engineering, Mathematics and Physical Sciences, University of Exeter, EX4 4QL, Exeter, UK}

\author{Jorge Gonçalves}
\affiliation{Luxembourg Centre for Systems Biomedicine, University of Luxembourg, 7 avenue des Hauts-Fourneaux, L-4362 Esch-sur-Alzette, Luxembourg}
\affiliation{Department of Plant Sciences, University of Cambridge, CB2 3EA, Cambridge, UK}

\date{\today}
	
\begin{abstract}
Time-series analysis is fundamental for modeling and predicting dynamical behaviors from time-ordered data, with applications in many disciplines such as physics, biology, finance, and engineering. Measured time-series data, however, are often low dimensional or even univariate, thus requiring embedding methods to reconstruct the original system's state space. The observability of a system establishes fundamental conditions under which such reconstruction is possible. However, complete observability is too restrictive in applications where reconstructing the entire state space is not necessary and only a specific subspace is relevant. Here, we establish the theoretic condition to reconstruct a nonlinear functional of state variables from measurement processes, generalizing the concept of functional observability to nonlinear systems. When the functional observability condition holds, we show how to construct a map from the embedding space to the desired functional of state variables, characterizing the quality of such reconstruction. The theoretical results are then illustrated numerically using chaotic systems with contrasting observability properties. By exploring the presence of functionally unobservable regions in embedded attractors, we also apply our theory for the early warning of seizure-like events in simulated and empirical data. The studies demonstrate that the proposed functional observability condition can be assessed \textit{a priori} to guide time-series analysis and experimental design for the dynamical characterization of complex systems.
\end{abstract}
% Our theory thus generalizes


\maketitle

%===========================================================================================================
\section{Introduction}
\label{sec.intro}

% Reconstructing the state space of complex dynamical systems is a key step for their quantitative understanding and forecasting. To determine if such reconstruction is possible, to quantify its expected performance and to apply it to empirical data, we present the conditions for functional observability of nonlinear systems. 
%
% Our study is motivated by current time series analysis methods, which characterize and model dynamical behaviors from recorded time-ordered data, being constrained by available measurement processes and data quality: data are often irregularly sampled, noisy, relatively short, and univariate. When the measured time series are expected to be lower dimensional compared to the system dynamics, the original state space can be reconstructed with embedding methods \cite{Marwan2009,Lekscha2018,Bhat2022},

Reconstructing the state space of complex dynamical systems is a key step for their quantitative understanding and forecasting. To do so, time-series analysis methods have been developed to characterize and model dynamical behaviors from recorded time-ordered data. In practice, such methods are constrained by the available measurement processes and data quality: data are often irregularly sampled, noisy, relatively short, and univariate. In cases that measured time series are expected to be lower dimensional compared to the system dynamics, the original state space can be reconstructed with embedding methods \cite{sauer1991embedology,Marwan2009,Lekscha2018,Bhat2022}. 
Embedding has been successfully applied across many fields, including for the characterization of chaotic dynamics \cite{kennel1992method,Marwan2007}, ecological and economic modeling \cite{Sugihara2012,Groth2017}, financial forecasting \cite{Araujo2019,Gidea2020},
medical diagnosis \cite{Carvalho2018,PerezToro2020}, and detection of dynamical transitions in palaeoclimate data \cite{Lekscha2020} and oil-water flows \cite{Gao2013}. 
These applications assume that an embedding is possible, that is, that a map from the time-series data to the original system's state space exists. However, such assumption may not hold in general: an established link between embedding and observability theories shows that reconstructing the \textit{entire} original system is not always possible or suitable depending on the available time series~\cite{Letellier2005,Aguirre2005}.

The \textit{a priori} conditions under which the entire system state can be reconstructed from available measurements are determined by the observability property \cite{Kalman1959}. In the linear case, observable systems satisfy the necessary and sufficient conditions for the inference of the full-state of the system, for example, via embedding methods and state estimators like Luenberger observers~\cite{Luenberger1966} or Kalman filters~\cite{Kalman1960}. For nonlinear systems, a generalized notion of observability~\cite{Hermann1977} sets a sufficient condition for the existence of an invertible mapping (diffeomorphism) between the system's original state space and the differential embedding space constructed from a given measurement function \cite{Letellier2005,Aguirre2005}. 

\begin{figure*}[t]
    \centering
    \includegraphics[width=0.80\linewidth]{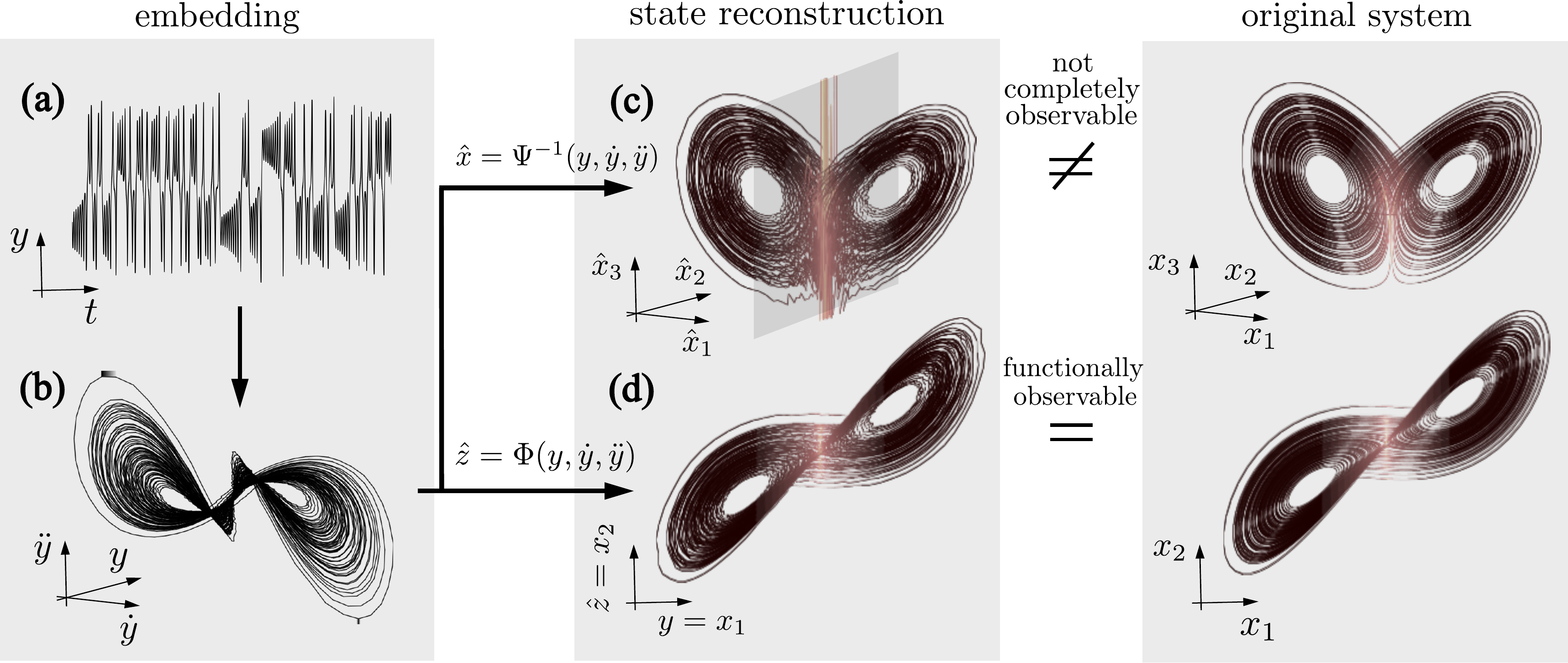}
    \caption{Functional observability versus complete observability for the reconstruction of the Lorenz attractor from time-series data.
    (a)~Time-series data $y = h(\bm x) = x_1$ measured from one of the system's state variables. 
    (b)~Differential embedding coordinates of the time-series data. 
    (c)~Reconstruction of the original coordinates of the Lorenz system from the differential embedding. The system is unobservable at the shaded plane ($x_1 = 0$) and, therefore, the reconstructed attractor (left side) shows large errors closer to this region compared to the ground-truth (unmeasured) attractor (right side). 
    (d)~Reconstruction of the low-dimensional subspace $(x_1,x_2)$ of the Lorenz attractor. The system is functionally observable everywhere with respect to the functional $z = g(\bm x) = x_2$ and, therefore, the reconstructed subspace (left side) is accurate compared to the ground-truth subspace (right side).
    Brighter colors correspond to states of the attractor closer to the system's unobservable plane.
    % The attractors are color coded according to the proximity of each state to the system's unobservable plane.
    }
    \label{fig.summary}
\end{figure*}

Not only the relation between embedding and observability determines the possibility of reconstructing the original dynamical system from the embedding of time-series data, but it also characterizes \textit{how good} such reconstruction is. For example, the embedding space of the R\"ossler system was empirically shown to have singularity points depending on the measured variable \cite{Pecora1990}, which was later theoretically proven to be a consequence of unobservable regions in the original space \cite{Letellier2002}. 
%Another example in the well-known Lorenz system, is shown in Fig.
Another case is the well-known Lorenz system: Fig.~\ref{fig.summary}a-c shows that unobservable regions hamper the quality of the attractor reconstruction from an embedded time series. Consequently, the applicability and performance of methods based on embedding \cite{Portes2016,Carroll2018,Portes2019} and state estimation \cite{Haber2017,Guan2018,Montanari2019} are highly dependent on the observability properties of the system.
% *omitted Portes2014*
Building from chaotic systems, often used as benchmark cases due to their short horizon of predictability, observability studies showed the potential to foster further discoveries in complex systems, including neuronal models \cite{Su2017,Aguirre2017}, metabolic reactions \cite{Liu2013c}, ecological systems \cite{Aparicio2021}, and networks \cite{Pasqualetti2013,Sun2013,Whalen2015,Letellier2018,Angulo2020,Montanari2020}.
% Hence, observability studies hold great potential to foster discoveries in complex dynamical systems, from chaotic benchmarks to neuronal models \cite{Aguirre2017,Su2017}, metabolic reactions \cite{Liu2013c}, and network systems \cite{Sun2013,Whalen2015,Letellier2018,Montanari2020}.

For many systems and applications, though, complete observability is a condition that may be too restrictive. In practice, even if the original state space is not entirely observable (reconstructible), one may focus on particular subspaces (e.g., state variables) that are relevant to the considered applications. 
Examples include the estimation of the phase variable of nonlinear oscillators for synchronization analysis of chaotic systems \cite{Pecora1990,Rosenblum1997,Freitas2018}, modeling of climate dynamics \cite{Oh2014}, and forecasting of financial crashes \cite{Smug2018}; the positioning and tracking of a particular spatial coordinate (e.g., altitude) in autonomous aerial vehicles from indirect measurements \cite{Youn2020}; or the inference of control variables (which dictate how close a system is to a bifurcation) 
for the early warning of transitions from healthy to disease states in atrial fibrillation \cite{quail2015predicting} and epileptic seizures \cite{Jirsa2014}.

These practical problems motivate the concept of \textit{functional observability}~\cite{Fernando2010,Jennings2011}, which establishes conditions under which a desired functional of the system variables can be inferred from the available measurements (e.g., via the design of functional observers \cite{Darouach2000,Fernando2010,Trinh2012}). However, in spite of several applications designed for feedback control \cite{Alhelou2019,Trinh2012}, fault detection \cite{Emami2015}, and, more recently, large-scale networks \cite{Montanari2021}, the functional observability property is still restricted to linear dynamical systems.
Therefore, a theory to establish conditions for the reconstruction of a \textit{nonlinear} functional of a \textit{nonlinear} system, and that provides guidance on how to perform and leverage this reconstruction, is still missing.

% Paper contributions
In this paper, we provide a generalization of the functional observability property to nonlinear dynamical systems. This establishes a sufficient condition for the reconstruction of a nonlinear functional of state variables from a (possibly nonlinear) measurement function. If the system is functionally observable, we show how to determine the mapping between the differential embedding space and the original functional sought to be reconstructed, also proposing a coefficient of functional observability to locally characterize the quality of the functional reconstruction. Fig.~\ref{fig.summary} shows that, even though a system may not be completely observable (thus hampering the state-space reconstruction), it might still be functionally observable with respect to some low-dimensional subspace of interest, in which accurate reconstruction is still feasible. To illustrate the theoretical advantages of such framework in interpreting the effects of singularities and symmetries in the embedded state space, we present numerical simulations for chaotic benchmark systems with contrasting observability properties. Finally, we apply our theory for the analysis of a phenomenological model of seizure-like events, known as Epileptor~\cite{Jirsa2014}. We demonstrate that the presence of unobservable regions in the Epileptor's attractor can be used to provide early-warning signals of transitions from normal to seizure states in both simulated and empirical data.

% Outline.
The paper is organized as follows. Section~\ref{sec.back} provides a background on the complete observability of dynamical systems. Section~\ref{sec.method} presents the theoretical results for the generalization of functional observability to nonlinear systems. Section~\ref{sec.numresult} presents and discusses the numerical results in chaotic benchmark systems. Section \ref{sec.Epileptor} applies the proposed framework for the analysis of the Epileptor model and early warning of seizures. Finally, Section~\ref{sec.conc} concludes the work.

%--------------------------------------------------------------------
%====================================================================
\section{Background on observability of nonlinear systems}
\label{sec.back}

%--------------------------------------------------------------------
% \subsection{Observability of nonlinear systems}

This section provides background on observability theory for nonlinear systems \cite{Hermann1977,VidyasagarBook,Montanari2020}. 
%The reader is referred to references for further details \cite{Hermann1977,VidyasagarBook,Montanari2020}. 
Consider the following nonlinear dynamical system
\begin{equation}
\begin{cases}
\dot{\bm x} = \bm f(\bm x), \\
\bm y = \bm h(\bm x),
\end{cases}
\label{eq.nonlinearsys}
\end{equation}

\noindent
where $\bm x\in\mathcal X\subseteq\R^n$ is the state vector, $\bm y\in\R^q$ is the output vector (measurements), and $\bm f:\mathcal X\mapsto \mathcal V(\mathcal X)\subseteq \R^n$ and $\bm h:\mathcal X\mapsto\mathcal H(\mathcal X)\subseteq\R^q$ are smooth nonlinear functions. The explicit dependence of time in $\bm x(t)$ and $\bm y(t)$ is often omitted throughout this work. Let the flow map $\bm\Phi_T(\bm x(t_0)):\mathcal X\mapsto\mathcal X$ be the solution of \eqref{eq.nonlinearsys}:
\begin{equation}
\bm\Phi_T(\bm x(t_0)) := \bm x(t_0+T) = \bm x(t_0) + \int_{t_0}^{t_0+T} \bm f(\bm x(t)){\rm d}t .
\end{equation}

\noindent
The notion of local observability is formalized as follows.

\begin{defin} \cite{Hermann1977,VidyasagarBook}
	The nonlinear system \eqref{eq.nonlinearsys}, or the pair $\{\bm f,\bm h\}$, is \textit{locally observable at} $\bm x_0$ if there exists a neighborhood $\mathcal U\subseteq\mathcal X$ of $\bm x_0$ such that, for every state $\bm x_0\neq\bm x_1\in\mathcal U$, $\bm h\circ \bm\Phi_T(\bm x_0)\neq\bm h\circ\bm\Phi_T(\bm x_1)$ for some finite time interval $t\in [t_0, t_0+T]$. 
	Otherwise, $\{\bm f,\bm h\}$ is \textit{locally unobservable at} $\bm x_0$.
	The system is said to be \textit{locally observable} if it is locally observable at every $\bm x_0\in\mathcal X$.
\label{def.obsvnonlinear}
\end{defin}

Definition \ref{def.obsvnonlinear} states that a system is locally observable around an initial state $\bm x(t_0)$ if $\bm x(t_0)$ can be uniquely reconstructed from the measurements $\bm y$ over a finite trajectory, as defined by the composition map $\bm h\circ\bm\Phi_T(\bm x(t_0))$. The local observability of a nonlinear system can be verified through the following algebraic condition.

\begin{defin}\cite{VidyasagarBook}
    The observable space $\mathcal O(\bm x)$ of system \eqref{eq.nonlinearsys} is the linear space of functions over the field $\R$ spanned by all functions of the form
	\begin{equation}
	    \mathcal L^\nu_{\bm f} h_j(\bm x), \quad  0\leq\nu\leq s, \quad 1\leq j\leq q,
	\label{eq.obsv_space}
	\end{equation}
	
	\noindent
	where $\mathcal L^\nu_{\bm f} h_j(\bm x)$ denotes the $\nu$-th Lie derivative of the $j$-th component of $\bm h(\bm x)$ along the vector field $\bm f(\bm x)$, and $s$ is the smallest integer such that $\gradient\mathcal L^{k}_{\bm f} h_j(\bm x)$ belongs to the span formed by functions~\eqref{eq.obsv_space} for all $k> s$. By definition, $\mathcal L^0_{\bm f} h_j(\bm x)\coloneqq h_j(\bm x)$ and 
	\begin{equation}
	    \mathcal L^\nu_{\bm f} h_j(\bm x)\coloneqq\gradient \mathcal L^{\nu-1}_{\bm f} h_j(\bm x)\cdot \bm f(\bm x),
	    \label{eq.liederivatives}
	\end{equation}
	where $\gradient$ is the gradient operator with respect to $\bm x$.
\label{def.obs_space}
\end{defin}

\begin{theor}
	The system \eqref{eq.nonlinearsys}, or the pair $\{\bm f, \bm h\}$, is locally observable at $\bm x_0$ if there exists a neighborhood $\mathcal{U} \subseteq \mathcal{X}$ of $\bm x_0$ such that
    \begin{equation}
    	    \mathrm{dim} \left\{\gradient\mathcal{O}(\bm x) \right\} = n
    	    \label{eq.nonlinearobsvcondition}
    	\end{equation}
	holds for every $\bm x\in\mathcal U\subseteq\mathcal X$, where $\mathcal O(\bm x)$ is the observable space.
    \label{theor.nonlinearobsv}
\end{theor}

\begin{proof}
    See \cite{Hermann1977,VidyasagarBook,Montanari2020}.
\end{proof}

% Remark
Note that the minimum order $\nu$ of Lie derivatives \eqref{eq.obsv_space} such that condition \eqref{eq.nonlinearobsvcondition} is satisfied depends on $\{\bm f, \bm h\}$. As a special case, if $\{\bm f, \bm h\}$ are rational functions, it suffices to check condition \eqref{eq.nonlinearobsvcondition} for the observable space $\mathcal O(\bm x)$ spanned by functions \eqref{eq.obsv_space} up to the $(n-1)$-th Lie derivative (i.e., $0\leq\nu\leq n -1$) \cite{AnguelovaThesis}.

% Remark
Theorem \ref{theor.nonlinearobsv} is a generalization of the observability property of linear systems to nonlinear systems. If the pair $\{\bm f,\bm h\}$ is given by the linear functions $A\bm x$ and $C\bm x$, where $A\in\R^{n\times n}$ and $C\in\R^{q\times n}$, then it follows that the observable space $\mathcal O(\bm x)$ is defined by the row space of Kalman's observability matrix \cite{Kalman1959,Chi-TsongChen1999}
\begin{equation}
%\mathcal O = 
%\begin{bmatrix}
%C \\ CA \\ \vdots \\ CA^{n-1}
%\end{bmatrix}.
\gradient \mathcal O(\bm x) = [C^\transp \,\,\, (CA)^\transp \,\,\, (CA^2)^\transp \,\,\, \ldots \,\,\, (CA^{n-1})^\transp]^\transp
\label{eq.obsvmatrix}
\end{equation}
and therefore the linear system is observable if and only if $\operatorname{rank}(\gradient \mathcal O)=n$.

%--------------------------------------------------------------------
\subsection{Observability and embedding}
\label{sec.obsvembedding}

Differential embedding relates the original coordinates of a dynamical system to the derivatives of the measured time series. Formally, given some measurement $\bm y = \bm h(\bm x)$ over time $t\in[0, T]$, a differential embedding space $\mathcal E$ can be constructed via an appropriate choice of higher-order derivatives of $\bm y$ as coordinates:
\begin{equation}
    \mathcal E = {\rm span} \{\gradient y_j^{(\nu)} : j\in\{1,\ldots,q\} \,\, \text{and} \,\, \nu\in\{0,\ldots,s\}\},
\label{eq.embeddingspace}
\end{equation}

\noindent
where $y^{(\nu)}_j$ denotes the $\nu$-th time derivative of the $j$-th component of the measured signal $\bm y$.
If the map $\bm\Psi : \mathcal X\mapsto \mathcal E$ is a diffeomorphism, then the original state  space $\mathcal X$ and the embedding space $\mathcal E$ are related by a smooth and invertible change of coordinates. Therefore, in applications where the original state space is not available (due to unmeasured state variables), the underlying dynamical system can be assessed via an embedding of the measured time series. 
In practice, measured time-series data are available in discrete time and hence the embedding can be constructed by either computing their differential coordinates $y^{(\nu)}_j$ (e.g., via regularization methods to denoise the derivatives \cite{Rudin1992,Chartrand2011}) or using time-delay coordinates $y_j(t-k\tau)$, for some time delay $\tau$ and $k\in\N$, following Takens' theorem \cite{Takens1981,Stark1997}.

%LF: for the citation, I consulted the Anishchenko et al. book (https://link.springer.com/book/10.1007/978-3-540-38168-6), page 266.

From Definition~\ref{def.obsvnonlinear}, $\{\bm f,\bm h\}$ is observable if the map
\begin{equation}
    \bm\Psi(\bm x) = 
    \begin{bmatrix}
    \bm y \\ \dot{\bm y} \\ \vdots \\ \bm y^{(\nu)}
    \end{bmatrix}
    =
    \begin{bmatrix}
    \bm h(\bm x) \\ \dv{\bm h(\bm x)}{t} \\ \vdots \\ \dv[\nu]{\bm h(\bm x)}{t}
    \end{bmatrix}
    =
    \begin{bmatrix}
    \bm{\mathcal L}_{\bm f}^0 \bm h(\bm x) \\ \bm{\mathcal L}_{\bm f}^1 \bm h(\bm x) \\ \vdots \\ \bm{\mathcal L}_{\bm f}^{\nu} \bm h(\bm x) 
    \end{bmatrix}
\label{eq.embeddingmap}
\end{equation}
is locally left invertible (injective) for some $0\leq\nu\leq s$, where $\bm{\mathcal L}_{\bm f}^\nu \bm h(\bm x) \coloneqq [\mathcal L_{\bm f}^\nu h_1(\bm x) \,\, \ldots \,\, \mathcal L_{\bm f}^\nu h_q(\bm x)]^\transp$. That is, if $\bm x$ is uniquely determinable from $\bm y$ and its successive derivatives. The map \eqref{eq.embeddingmap} is equivalent to the set of functions \eqref{eq.obsv_space} spanning the observable space $\mathcal O(\bm x)$ \cite{Letellier2005}, hence $\bm\Psi : \mathcal X \mapsto \mathcal E \equiv \mathcal O(\bm x)$.
This equivalence establishes a direct relation between the theories of observability and embedding \cite{Letellier2005,Aguirre2005}:
following the Inverse Function Theorem \cite[Section 7.1]{VidyasagarBook}, $\bm\Psi(\bm x)$ is locally invertible if its Jacobian matrix has full (column) rank (i.e., if condition \eqref{eq.nonlinearobsvcondition} holds).
Therefore, a state $\bm x_0$ is only reconstructable from the differential embedding of its measurements $\bm y$ over a finite time interval if $\bm\Psi$ is locally invertible or, in other words, the system is locally observable at $\bm x_0$.

% Remark
Note that the choice of differential coordinates $y_j^{(\nu)}$ for the embedding space $\mathcal E$ is not unique and that an inappropriate selection of linearly dependent derivatives of $\bm y$ (i.e., functions of form \eqref{eq.obsv_space}) may lead to $\operatorname{dim}(\gradient\mathcal E)<\operatorname{dim}(\gradient\mathcal O(\bm x))$ \cite{Aguirre2005}. Accordingly, we assume throughout this work that the embedding space \eqref{eq.embeddingspace} is defined by a minimum selection of linearly independent functions $\eqref{eq.obsv_space}$ such that $\operatorname{dim}(\gradient\mathcal E)=\operatorname{dim}(\gradient\mathcal O(\bm x))$ around some neighborhood of $\bm x_0$. For example, we consider the embedding space $
\mathcal E = \{y,\dot y,\ddot y\}$ for the reconstruction of the Lorenz attractor (Fig.~\ref{fig.summary}a--c).

%====================================================================
\section{Functional observability of nonlinear systems}
\label{sec.method}

Complete observability characterizes the sufficient condition for the (local) reconstruction of the \textit{full-state} vector $\bm x$ of a dynamical system \eqref{eq.nonlinearsys} from measurements $\bm y$ over finite time. 
%Moreover, the observability property directly determines the existence of a mapping between the original state space $\mathcal X$ and the embedding space $\mathcal E$ defined by $\bm y$ and its differential coordinates.
%
Nonetheless, reconstructing the entire state vector $\bm x$ is often unfeasible or unnecessary in practice, and only a lower-dimensional function or subspace might be of interest, defined as
\begin{equation}
\bm z = \bm g(\bm x),
\label{eq.nonlinearfunc}
\end{equation}

\noindent
where $\bm z\in\R^{r}$ is the vector sought to be reconstructed, often with dimension $r \ll n$, and $\bm g:\mathcal{X} \mapsto\mathcal{G}(\mathcal{X})\subseteq\R^{r}$ is a nonlinear smooth functional.

In what follows, we provide a generalization of the observability property \cite{Hermann1977}, termed \textit{functional observability}.
Given a nonlinear dynamical system~\eqref{eq.nonlinearsys}, functional observability establishes the conditions under which the functional~\eqref{eq.nonlinearfunc} is reconstructible from the measured signal $\bm y$, without necessarily requiring the full state $\bm x$ to be reconstructible. Therefore, a system may be functionally observable with respect to some functional \eqref{eq.nonlinearfunc} even though it is (completely) unobservable.
Our results also generalize the functional observability property, originally established for linear systems \cite{Jennings2011}, to dynamical systems described by a nonlinear vector field $\bm f$, a nonlinear measurement function $\bm h$ and a nonlinear functional $\bm g$. Fig.~\ref{fig.generalizations} summarizes the relation between our theory and previous works. \\

\begin{figure}[t]
    \centering
    \includegraphics[width=\columnwidth]{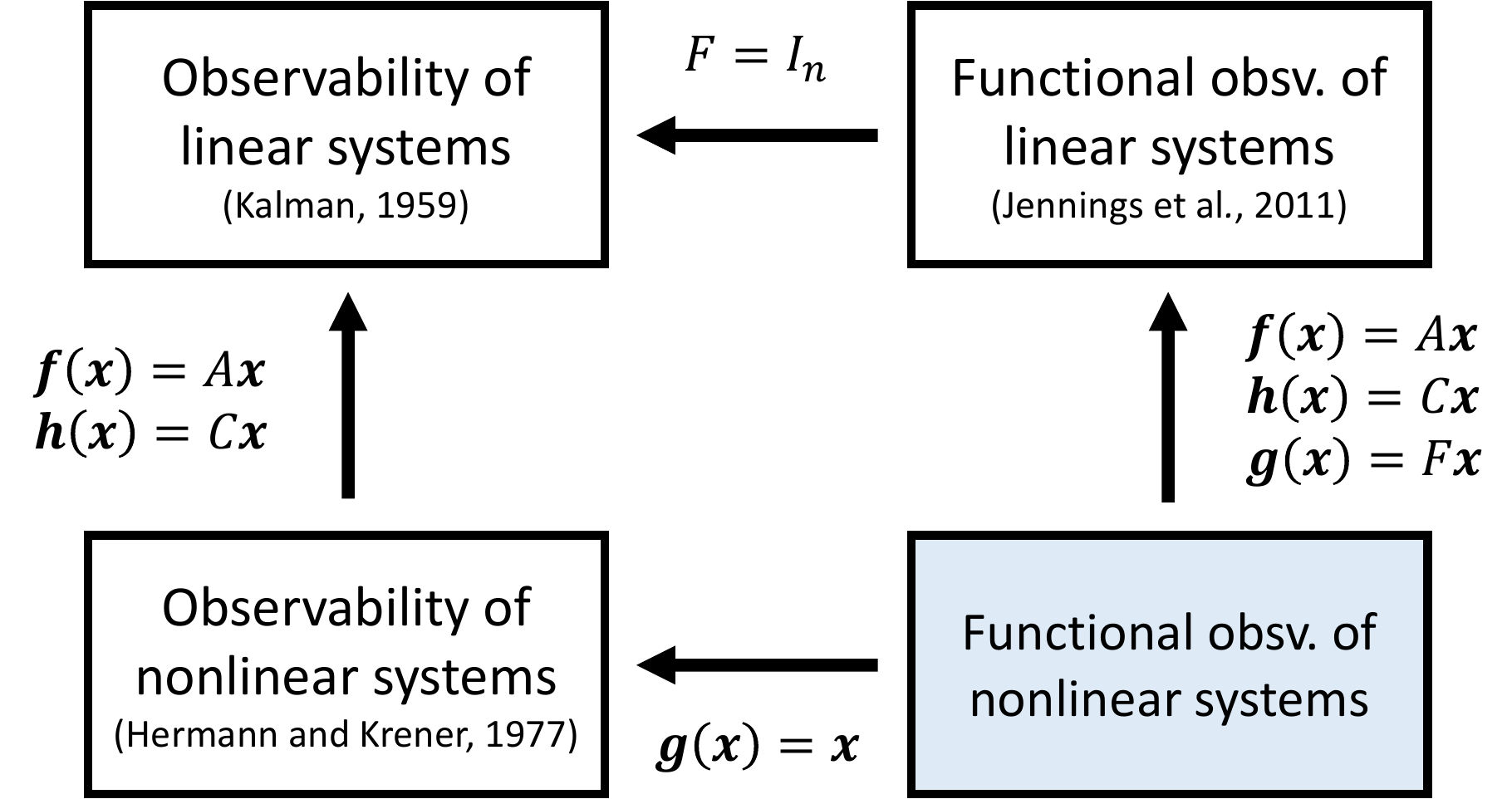}
    \caption{Functional observability of nonlinear systems and its special cases.}
    \label{fig.generalizations}
\end{figure}

Before presenting the main result, we first formally define functional observability as a generalization of complete observability (Definition~\ref{def.obsvnonlinear}):

\begin{defin}
	The nonlinear system \eqref{eq.nonlinearsys} and \eqref{eq.nonlinearfunc}, or the triple $\{\bm f,\bm h, \bm g\}$, is \textit{locally functionally observable at} $\bm x_0$ if there exists a neighborhood $\mathcal U\subseteq\mathcal X$ of $\bm x_0$ such that, for every state $\bm g(\bm x_0)\neq\bm g(\bm x_1)$, $\bm h\circ \bm\Phi_T(\bm x_0)\neq\bm h\circ\bm\Phi_T(\bm x_1)$ for some finite time interval $t\in [t_0,t_0+T]$. 
	Otherwise, $\{\bm f,\bm h,\bm g\}$ is \textit{locally functionally unobservable at} $\bm x_0$.
	The system is said to be \textit{locally functionally observable} if it is locally functionally observable at every $\bm x_0\in\mathcal X$.
\label{def.functobsvnonlinear}
\end{defin}

Analogous to Definition \ref{def.obsvnonlinear}, Definition \ref{def.functobsvnonlinear} states that a system is locally functionally observable around  an initial state $\bm x(t_0)$ if the functional $\bm g(\bm x(t_0))$ can be uniquely reconstructed from the measurements $\bm y$ over a finite trajectory. Analogous to Definition~\ref{def.obs_space}, we define the functional space related to \eqref{eq.nonlinearfunc} as follows:

% To address this problem, some definitions based on \cite[p.419]{VidyasagarBook} will be of help.

\begin{defin}
The functional space $\mathcal{F}(\bm x)$ of system~\eqref{eq.nonlinearsys} and \eqref{eq.nonlinearfunc} is the linear space of functions over the field $\mathbb{R}$ spanned by all functions of the form
\begin{equation}
    \mathcal L^\nu_{\bm f} g_j(\bm x),\quad 0 \le \nu \leq \mu,\quad 1 \le j \le r ,
\end{equation} 
where $\mu$ is the smallest integer such that $\gradient \mathcal L^{k}_{\bm f} g_j(\bm x)$ belongs to the span formed by~\eqref{def.func_space} for all $k>\mu$.
\label{def.func_space}
\end{defin}
Based on Definitions~\ref{def.obs_space} and~\ref{def.func_space}, we now establish the condition for the functional observability analysis of nonlinear systems. Consider a locally unobservable system at $\bm x_0$, i.e.,
    \begin{equation}
        \mathrm{dim} \left\{ \gradient\mathcal{O}(\bm x_0) \right\} = k \le n, \,\,\, \forall \bm x_0 \in \mathcal U\subseteq\mathcal X.
    \label{eq.obsvdimension}
    \end{equation}
    Recall the theorem for the decomposition of unobservable systems \cite[Theorem 97]{VidyasagarBook}: %(Decomposition of Unobservable Systems), 
    there exists a diffeomorphism $\bm T$ on $\mathcal U$ such that choosing an appropriate state transformation $\tilde{\bm{x}} = \bm T(\bm x)$ yields the partitioned vector
    \begin{equation}
        \tilde{\bm{x}} = \begin{bmatrix} \tilde{\bm{x}}_a \\ \tilde{\bm{x}}_b \end{bmatrix}, %\quad \tilde{\bm{x}}_a \in \mathbb{R}^k, \,\,\, \tilde{\bm{x}}_b \in \mathbb{R}^{n-k},
    \label{eq.statevectorpartitioned}
    \end{equation}
    where $\tilde{\bm x}_a\in\R^k$ and $\tilde{\bm x}_b\in\R^{n-k}$ correspond, respectively, to the observable and unobservable parts of the system in some neighborhood of $\bm T(\bm x_0)$. The transformed vector field is now given by
    \begin{equation} \label{eq:decompos}
        \tilde{\bm{f}}(\tilde{\bm x}) \equiv \tilde{\bm{f}}(\tilde{\bm{x}}_a, \tilde{\bm{x}}_b) = \begin{bmatrix} \tilde{\bm{f}}_a (\tilde{\bm{x}}_a) \\ \tilde{\bm{f}}_b (\tilde{\bm{x}}_a, \tilde{\bm{x}}_b) \end{bmatrix},
    \end{equation}
    where $\tilde{\bm f}_a : \R^k\mapsto\R^k$ and $\tilde{\bm f}_b : \R^n \mapsto\R^{n-k}$,
    and the transformed measurement function is given by
    \begin{equation}
        \tilde{\bm{h}}(\tilde{\bm x}) \equiv \tilde{\bm{h}}(\tilde{\bm{x}}_a) = \bm{h} \left( \bm T^{-1} (\tilde{\bm{x}}) \right),
    \label{eq.transformedmeasurement}
    \end{equation}
    which will only depend on $\tilde{\bm{x}}_a$.
    
    After this change of coordinates, the system \eqref{eq.nonlinearsys} and \eqref{eq.nonlinearfunc} is given by
    \begin{equation} \label{eq.transformedsys}
    \begin{cases} 
        \begin{bmatrix} \dot{\tilde{\bm{x}}}_a \\ \dot{\tilde{\bm{x}}}_b \end{bmatrix} = \begin{bmatrix} \tilde{\bm{f}}_a (\tilde{\bm{x}}_a) \\ \tilde{\bm{f}}_b (\tilde{\bm{x}}_a, \tilde{\bm{x}}_b) \end{bmatrix},  \\
        \bm{y} = \tilde{\bm{h}}(\tilde{\bm{x}}_a), \\
        \bm{z} = \tilde{\bm{g}}(\tilde{\bm{x}}_a, \tilde{\bm{x}}_b).
    \end{cases}
    \end{equation}
    
    \noindent
    We now note that if the following condition 
    \begin{equation}
    	    \mathrm{dim} \left\{\gradient\mathcal{O}(\bm x_0) \right\} = \mathrm{dim} \left\{\gradient\mathcal{O}(\bm x_0), \gradient\mathcal{F}(\bm x_0) \right\}
    	    \label{eq.functobsvcondition2}
    \end{equation}
	
	\noindent
    holds locally for every $\bm x_0 \in \mathcal U \subseteq \mathcal{X}$, then, for some $\tilde{\bm{x}}= \bm{T}(\bm{x})$, we have that $\tilde{\bm{g}}(\tilde{\bm{x}})\equiv \tilde{\bm{g}}(\tilde{\bm{x}}_a, \tilde{\bm{x}}_b)\equiv \tilde{\bm{g}}(\tilde{\bm{x}}_a)$ for all $\bm{x}\in\mathcal U$. Consequently, the system is functionally observable given that
    \begin{equation}
        \tilde{\bm{g}}(\tilde{\bm x}) \equiv \tilde{\bm{g}}(\tilde{\bm x}_a) = \bm{g} \left( \bm T^{-1} (\tilde{\bm{x}}) \right)
    \label{eq.transformedfunctional}
    \end{equation}
    depends only on $\tilde{\bm{x}}_a$, which is the state vector corresponding to the observable subsystem.

Condition \eqref{eq.functobsvcondition2} provides a sufficient condition for the local functional observability of the nonlinear system \eqref{eq.nonlinearsys} and \eqref{eq.nonlinearfunc}, or the triple $\{\bm f, \bm h, \bm g\}$, at $\bm x_0 \in \mathcal U$. Note that this condition is locally equivalent to $\mathcal F(\bm x_0)\subseteq \mathcal O(\bm x_0)$, that is, that the functional space (the subspace to be reconstructed) should be contained inside the observable space (the subspace reconstructible from the measurement function).
The following theorem provides a condition for the functional observability of a nonlinear system that is equivalent to condition \eqref{eq.functobsvcondition2}, but easier to implement since it does not require calculation of the functional space $\mathcal F(\bm x_0)$ and the involved higher-order Lie derivatives of $\bm g$. 

% --------------------------------------------------------
\begin{theor}
	The nonlinear system \eqref{eq.nonlinearsys} and \eqref{eq.nonlinearfunc}, or the triple $\{\bm f, \bm h, \bm g\}$, is locally functionally observable at $\bm x_0$ if there exists a neighborhood $\mathcal U \subseteq \mathcal{X}$ of $\bm x_0$ such that
    	\begin{equation}
    	    \mathrm{dim} \left\{\gradient\mathcal{O}(\bm x_0) \right\} = \mathrm{dim} \left\{\gradient\mathcal{O}(\bm x_0), \gradient\bm g(\bm x_0) \right\}
    	    \label{eq.functobsvcondition}
    	\end{equation}
	    holds for every $\bm x_0 \in \mathcal U \subseteq \mathcal{X}$, where $\mathcal O(\bm x)$ is the observable space.
	\label{theor.nonlinearfuncobsv}
\end{theor}

\begin{proof}
See Appendix \ref{app.appendproof}.
\end{proof}

% Remark

The functional observability condition \eqref{eq.functobsvcondition} establishes an easy-to-implement test that can be directly verified on the system's equations, that is, functions $\{\bm f,\bm h,\bm g\}$. We note that,  even though its derivation is based on the system decomposition \eqref{eq.transformedsys} presented in \cite[Theorem 97]{VidyasagarBook}, this condition does not require any system transformation or prior knowledge of its equivalence transformation map $\bm T(\bm x)$. This provides a significant advantage to analyze the limitations in subspace reconstruction of dynamical systems where deriving the transformation map $\bm T(\bm x)$ and its inverse can be computationally intensive, such as the Epileptor model studied in Section~\ref{sec.Epileptor}.

% Remark
If the triple $\{\bm f,\bm h,\bm g \}$ is given by linear functions $A\bm x$, $C\bm x$ and $F\bm x$, where $A\in\R^{n\times n}$, $C\in\R^{q\times n}$ and $F\in\R^{r\times n}$, then the observable space $\mathcal O(\bm x)$ is defined by the row space of the observability matrix \eqref{eq.obsvmatrix} and the functional observability of a linear system can be verified using the following rank condition derived in \cite{Jennings2011,Rotella2016a}:
\begin{equation}
\rank(\gradient\mathcal O)
=
\rank\left(
\begin{bmatrix}
\gradient\mathcal O \\ F
\end{bmatrix}\right)
.
\label{eq.linearfunctobsv}
\end{equation}

% Remark
If the full-state vector is sought to be reconstructed (i.e., $\bm g(\bm x) = \bm x$), then $\operatorname{dim}\{\gradient \bm g(\bm x)\} = n$ and, therefore, the functional observability condition \eqref{eq.functobsvcondition} reduces to the complete observability condition \eqref{eq.nonlinearobsvcondition}. Likewise, in the linear case,  complete observability is a special case of functional observability by considering $F=I_n$, where $I_n$ is the identity matrix of size $n$, which reduces condition \eqref{eq.linearfunctobsv} to the classical condition $\rank(\gradient\mathcal O)=n$. 
\\

% --------------------------------------------------------
% \begin{example}
%
As an illustrative example, consider a dynamical system \eqref{eq.nonlinearsys} and \eqref{eq.nonlinearfunc}, or equivalently the triple $\{\bm f, h, g\}$, defined by
\begin{equation}
    \bm f (\bm x) = \begin{bmatrix} 2\,x_{1}\\ x_{2}+\frac{x_{3}}{\sqrt{x_{1}}}\\ 2x_{3} \end{bmatrix}, \,\,\,
    h(\bm x) = x_2, \,\,\,
    g (\bm x) = x_{2}+\frac{x_{3}}{\sqrt{x_{1}}},
\label{eq.ex_001}
\end{equation}
where $\bm x = [x_1 \,\, x_2 \,\, x_3]^\transp\in\mathcal X$ and $\mathcal{X} = \left\{\bm x\in\mathbb{R}^3 \,\, | \,\, x_1 \neq 0 \right\}$ is the region of analysis. Following Theorem~\ref{theor.nonlinearfuncobsv}, we first need to determine a basis for the observable and functional spaces according to Definitions~\ref{def.obs_space} and \ref{def.func_space}:
\begin{align}
\gradient\mathcal{O}(\bm x) 
&=
\pdv{}{\bm x} \begin{bmatrix} h({\bm x}) \\ \mathcal L_{\bm f} h({\bm x}) \\ \mathcal L^2_{\bm f} h({\bm x}) \\ \mathcal L^3_{\bm f} h({\bm x}) \\ \mathcal L^4_{\bm f} h({\bm x}) \\ \vdots \end{bmatrix} 
= 
\begin{bmatrix}
0 & 1 & 0 \\
-\frac{1}{2}\frac{x_3}{x_1^{3/2}} & 1 & \frac{1}{\sqrt{x_1}} \\
-\frac{x_3}{x_1^{3/2}} & 1 & \frac{2}{\sqrt{x_1}} \\
-\frac{3}{2}\frac{x_3}{x_1^{3/2}} & 1 & \frac{3}{\sqrt{x_1}} \\
-{2}\frac{x_3}{x_1^{3/2}} & 1 & \frac{4}{\sqrt{x_1}} \\
\vdots & \vdots & \vdots
\end{bmatrix}
\\
\gradient {g}(\bm x) &=  
\begin{bmatrix} 
-\frac{1}{2}\frac{x_{3}}{{x_{1}}^{3/2}} & 1 & \frac{1}{\sqrt{x_{1}}} 
\end{bmatrix}.
\end{align}
Rows $\{3, 4, 5, ...\}$ of $\gradient\mathcal O(\bm x)$ are linear combinations of the first and second rows, leading to $\operatorname{dim}\{\gradient\mathcal O\} = 2, \, \forall \bm x \in \mathcal{X}$. This shows that, according to Theorem~\ref{theor.nonlinearobsv}, the pair $\{\bm f, h\}$ is not (completely) observable, i.e., it is not possible to reconstruct the entire state vector $\bm x$ from $y=x_2$. However, the triple $\{\bm f, h, g\}$ may still be functionally observable. By inspection, it is easy to see that $\gradient g(\bm x)$ is a linear combination of the rows of $\gradient\mathcal{\mathcal O}(\bm x)$, therefore satisfying condition \eqref{eq.functobsvcondition} of Theorem~\ref{theor.nonlinearfuncobsv} for all $\bm x\in\mathcal X$. This example illustrates that, even though a system is locally unobservable, it can still be locally functionally observable. %\QEDT
%\label{examp.obsv-vs-funct}
% \end{example}

% --------------------------------------------------------
\subsection{Functional observability and embedding}
\label{sec.functobsvembedding}

Complete observability establishes a sufficient condition for the existence of the (local) left-inverse map $\bm\Psi^{-1}:\mathcal E\mapsto\mathcal X$ from an embedding space to the original state space. Here, we generalize this relation by showing that functional observability establishes a sufficient condition for the existence of a map $\bm\Phi : \mathcal E\mapsto\mathcal G(\mathcal X)$ from the embedding space to the subspace sought to be reconstructed. Furthermore, we demonstrate how to construct such a map if the system is functionally observable.

Let $\{\bm f,\bm h, \bm g\}$ be a functionally observable system with an observable space $\mathcal O(\bm x)$ of local dimension \eqref{eq.obsvdimension}.
%and consider a choice of embedding space such that $\mathcal E^* \equiv \mathcal O(\bm x)$, where $\operatorname{dim}\{\gradient\mathcal O(\bm x)\} = k\leq n$. 
Following \cite[Theorem 97]{VidyasagarBook}, there exists a diffeomorphism $\bm T$ on $\mathcal U$ such that the transformation $\tilde{\bm x} = \bm T(\bm x)$ partitions the state vector as in \eqref{eq.statevectorpartitioned}. Consequently, the triple $\{\bm f,\bm h,\bm g\}$ can now be represented by $\{\tilde{\bm f}, \tilde{\bm h}, \tilde{\bm g}\}$ as in \eqref{eq.transformedsys}.
The diffeomorphism $\bm T$ is not unique and can be designed by partitioning it as
\begin{equation}
    \bm T(\bm x) = \begin{bmatrix}
    \tilde{\bm x}_a \\ \tilde{\bm x}_b
    \end{bmatrix} = \begin{bmatrix}
    \bm\Psi_a(\bm x) \\ \bm\Psi_b(\bm x)
    \end{bmatrix},
\label{eq.partitioneddiff}
\end{equation}

\noindent
for some neighborhood $\mathcal U\subseteq\mathcal X$ of $\bm x_0$. The functions $\bm\Psi_a(\bm x)$ and $\bm\Psi_b(\bm x)$ can be designed as follows:
\begin{enumerate}
    \item Construct a map $\bm\Psi_a(\bm x)$ by selecting a minimum set of linearly independent functions of form \eqref{eq.obsv_space} such that $\operatorname{dim}\{\gradient\bm\Psi_a(\bm x)\}=\operatorname{dim}\{\gradient\bm\Psi_a(\bm x), \gradient\mathcal O(\bm x)\}=k$, $\forall \bm x\in\mathcal U\subseteq \mathcal X$.
    
    \item Construct an arbitrary map $\bm\Psi_b(\bm x)$
    such that $\operatorname{dim}\{\gradient \bm T(\bm x)\}=n$, $\forall \bm x\in\mathcal U\subseteq\mathcal X$, i.e., $\gradient\bm T(\bm x)$ has full rank.
\end{enumerate}

Remind that $\bm\Psi_a (\bm x)$ and $\bm\Psi_b (\bm x)$ are only valid around a local neighborhood $\mathcal U\subseteq\mathcal X$ of some state $\bm x_0\in\mathcal X$. 
The transformation $\bm\Psi_a (\bm x): \mathcal U \mapsto \mathcal E$ defines a basis for the embedding space $\mathcal E$, which depends only on $\bm y$ and its successive derivatives (see Eq.~\eqref{eq.embeddingspace}). 
Since $\bm y=\bm h(\bm x) = \tilde{\bm h}(\tilde{\bm x}_a)$ is a function only of $\tilde{\bm x}_a$ (following Eq.~\eqref{eq.transformedmeasurement}), then step~1 guarantees that $\tilde{\bm x}_a\in\mathcal E$.
Step 2 constructs an arbitrary function $\bm\Psi_b (\bm x) : \mathcal U \mapsto \mathcal E^{\rm c}$ that defines a basis to the complement of the embedding space $\mathcal E^{\rm c}$ in order to accomplish the Inverse Function Theorem.
Therefore, the designed diffeomorphism $\bm T : \mathcal U \mapsto \mathcal E \cup \mathcal E^{\rm c}$ has a local inverse map $\bm T^{-1}$. 

Since the system is functionally observable, then relation \eqref{eq.transformedfunctional} holds and $\bm z = \bm g(\bm x)=\tilde{\bm g}(\tilde{\bm x}_a)$, which depends only on $\tilde{\bm x}_a\in\mathcal E$. Therefore, there exists a map $\tilde{\bm g} :\mathcal E\mapsto\tilde{\mathcal G}(\mathcal E)$ from the embedding space $\mathcal E$ to the functional sought to be reconstructed $\bm z = \tilde{\bm g}(\tilde{\bm x}_a)$. Equivalently, $\bm z = \bm g(\bm x)$ can be reconstructed from the composition map $\bm \Phi = \bm g \circ \bm T^{-1}$, which can be locally constructed from the known function $\bm g$ and the designed transformation $\bm T$ according steps 1 and 2. Finally, if the system is functionally observable, then
\begin{align}
    \big[\bm\Phi : \mathcal E \cup \mathcal E^{\rm c} \xmapsto{\bm T^{-1}} \mathcal U \xmapsto{\bm g} {\mathcal G}(\mathcal U) \big]
    \equiv
    \big[\tilde{\bm g} : \mathcal E \mapsto \tilde{\mathcal G}(\mathcal E) \big]
\label{eq.compositionmapPhi}
\end{align}

\noindent
Fig.~\ref{fig.compositionmap} illustrates a commutative diagram of the composition map \eqref{eq.compositionmapPhi}.
Clearly, if $k=n$, then the relation between complete observability and embedding follows as a special case:
\begin{equation}
    \bm T(\bm x)=\bm\Psi_a(\bm x) \,\,
    \text{and} \,\,
    \bm\Phi : \mathcal E \xmapsto{\bm\Psi_a^{-1}} \mathcal U \xmapsto{\bm g} \mathcal G(\mathcal U).
\end{equation}

\begin{figure}
    \centering
    \includegraphics[width=0.85\columnwidth]{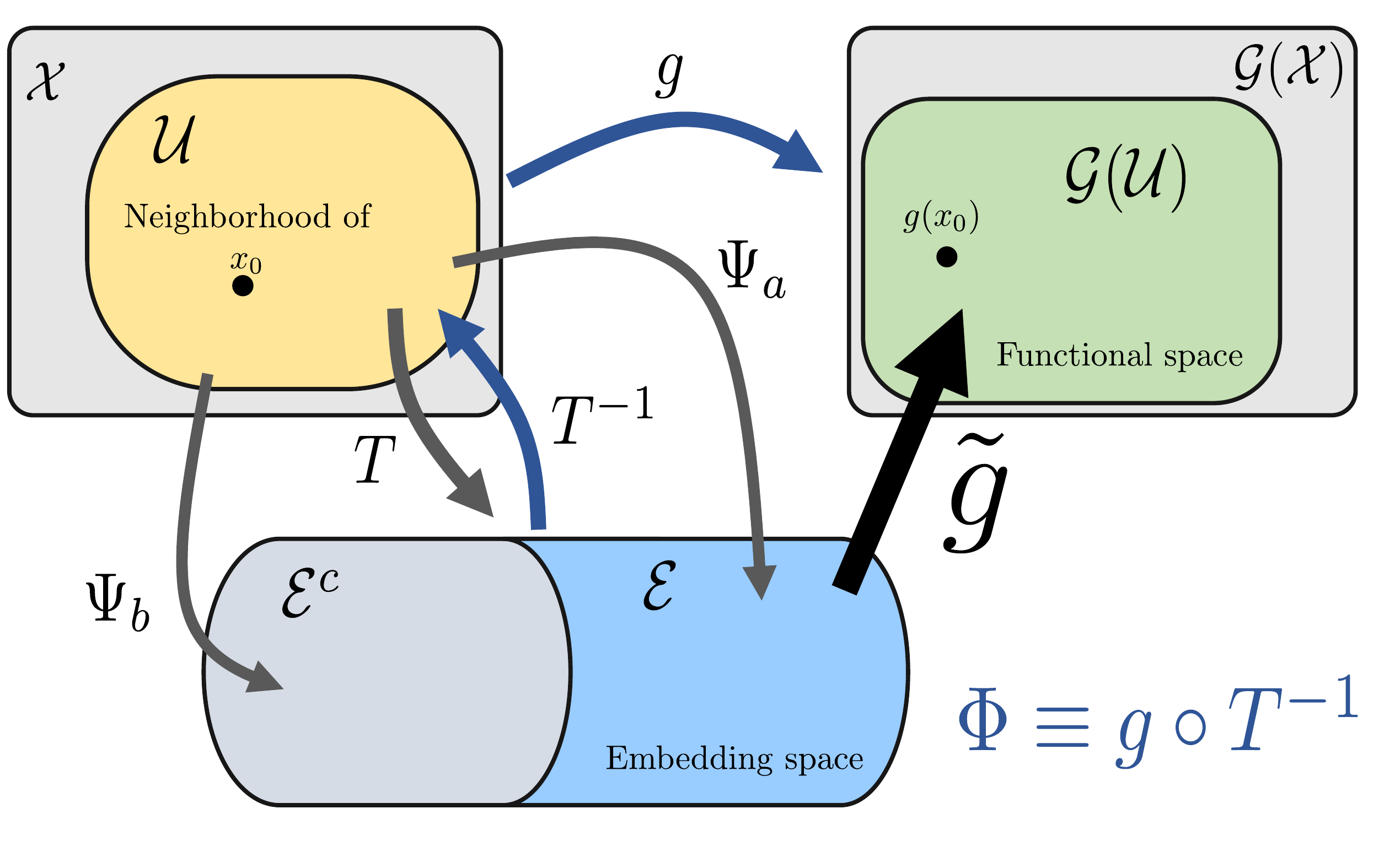}
    \caption{Commutative diagram of the composition map $\bm\Phi$ (blue path) which maps an embedding space $\mathcal E$ to the subspace sought to be reconstructed $\mathcal G(\mathcal U)$.}
    \label{fig.compositionmap}
\end{figure}

% Functional observability guarantees the existence of an injective map $\bm\Psi_{\bm g} : \bm g(\bm x) \mapsto  \{\bm y,\dot{\bm y}, \ddot{\bm y},\ldots \}$ or, equivalently, $\bm\Psi_{\bm g} : \mathcal G(\mathcal X) \mapsto \mathcal E^*$. If $\bm\Psi_{\bm g}$ is injective (functionally observable), then there exists a subspace $\mathcal X_a$ such that $T(\mathcal G(\mathcal X)) \equiv \tilde{\mathcal G}(\mathcal X_a)$ and that $\bm\Psi : \mathcal X_a \mapsto \mathcal E^*$ is invertible for some choice of embedding space $\mathcal E^*$. This subspace $\mathcal X_a$ is defined by the observable part $\tilde{\bm x}_a$ in the transformation $T(\bm x)$.

% --------------------------------------------------------
%\begin{example}
%
%Following Example \ref{examp.obsv-vs-funct}, 

Building on the previous example, given that the system \eqref{eq.ex_001} is functionally observable, we demonstrate how to compute (reconstruct) the sought vector $\bm z$ from the measurement signal $\bm y$. To this end, we first decompose the system as in \eqref{eq.statevectorpartitioned} by constructing a diffeomorphism $\bm T(\bm x)$ partitioned as \eqref{eq.partitioneddiff}:
\begin{equation}
    \tilde{\bm x} 
    \coloneqq 
    \begin{bmatrix} \tilde{\bm x}_a \\ \xdash[1.5em] \\ \tilde{\bm x}_b \end{bmatrix} 
    = 
    \bm T(\bm x) 
    =
    \begin{bmatrix} h(\bm x) \\ \mathcal L_{\bm f} h({\bm x}) \\ \xdash[3em] \\ \bm\Psi_b({\bm x}) \end{bmatrix}
    \equiv
    \begin{bmatrix} y \\ \dot{y} \\ \xdash[2em] \\ \bm\Psi_b({\bm x}) \end{bmatrix},
\end{equation}
where $\tilde{\bm x}_a \in \mathbb{R}^2$, $\tilde{\bm x}_b \in \mathbb{R}^1$, and $\operatorname{dim}\{\gradient\mathcal O\} = k = 2$. Note that $\bm\Psi_a(\bm x) = [y \,\,\, \dot{y}]^\transp$ defines a map between the observable vector and the differential embedding coordinates, while $\bm\Psi_b({\bm x})$ is chosen arbitrarily to accomplish the Inverse Function Theorem. Therefore, the diffeomorphism
\begin{equation}
    \tilde{\bm x} = \bm T(\bm x) = \begin{bmatrix} x_2 \\ x_{2}+\frac{x_{3}}{\sqrt{x_{1}}} \\ \xdash[4em] \\ x_1 \end{bmatrix}
\end{equation}
has the inverse function
\begin{equation}
    \bm x = \bm T^{-1}(\tilde{\bm x}) = \begin{bmatrix} \tilde{x}_3 \\ \tilde{x}_1 \\ \xdash[6em] \\ (\tilde{x}_2-\tilde{x}_1)\sqrt{\tilde{x}_3} \end{bmatrix}, \quad \forall \bm x\in\mathcal X.
\end{equation}
Consequently, $\bm z$ can be computed as
\begin{equation}
    \tilde{g} (\tilde{\bm x}) = g \left( \bm T^{-1}(\tilde{\bm x}) \right) = \tilde{x}_2 = x_2 + \frac{x_3}{\sqrt{x_1}} = \dot{x}_2 = \dot{y}.
\end{equation}
As expected, we have that $\tilde{{g}}({\tilde{\bm x}_a}, {\tilde{\bm x}_b})\equiv \tilde{{g}}({\tilde{\bm x}_a})$ and, therefore, $g \left( \bm T^{-1}(\tilde{\bm x}) \right)$ depends only on $\tilde{\bm x}_a$, which is a function of $y$ and its subsequent derivatives.
%\QEDT
%\end{example}

%====================================================================
\section{Observability of chaotic systems}
\label{sec.numresult} 

We explore the functional observability property in different types of chaotic dynamical systems with contrasting observability properties, considering different measurement functions as well as functionals sought to be reconstructed.
Following the theoretical conditions established in Section~
\ref{sec.method}, functional (or full-state) reconstruction is possible when the system is functionally (or completely) observable. Beyond this binary characterization of the system observability (i.e., either the system is or is not observable),
we show that the performance of the reconstructed functional (full-state) vector is dependent on the proximity of the system state to functionally (completely) unobservable regions in the state space. 
The reconstruction errors are related to the sensitivity of the maps $\bm\Psi^{-1}$ and $\bm\Phi$ to small perturbations (e.g., noise in the measured signals $\bm y(t)$), and can be quantified by the absolute condition number of the inverse maps between the embedding coordinates and the reconstructed state:
\begin{equation}
    \kappa(\bm\Psi^{-1}) = \norm{(\gradient\bm\Psi)^{-1}}
    \text{,} \,\,\,
    \kappa(\bm\Phi) = \norm{\gradient\bm g\cdot(\gradient\bm T)^{-1}}.
    \label{eq.coeffuncobsv}
\end{equation}

\noindent
We address the condition numbers $\kappa(\bm\Psi^{-1})$ and $\kappa(\bm\Phi)$ as the coefficients of complete and functional observability, respectively; a nomenclature that was previously adopted for $\kappa(\bm\Psi^{-1})$ in studies restricted to complete observability \cite{Letellier2002,Aguirre2005,Whalen2015,Montanari2020}. Results show that these coefficients can be employed to assess the quality of the (functional) state reconstruction as $\bm x(t)$ approaches unobservable regions: the larger $\kappa$, the higher the reconstruction error in the corresponding states. 

% As an alternative way to present the discussed results, which will be conducted in the following sections, we consider the corresponding absolute condition number to the inverse maps between the embedding coordinates and the reconstructed state:
% %
% \begin{equation}
%     \kappa(\bm\Psi^{-1}) = \norm{\gradient\bm\Psi}
%     \,\,\, \text{and} \,\,\,
%     \kappa(\bm\Phi) = \norm{\gradient\bm g\gradient\bm T^{-1}}.
% \end{equation}

% We address the condition numbers $\kappa(\bm\Psi)$ and $\kappa(\bm\Phi)$ as the coefficients of complete and functional observability, respectively\textemdash as previously done for the complete observability case \cite{Letellier2002,Montanari2020}. These coefficients quantify the sensitivity of the maps $\bm\Psi^{-1}$ and $\bm\Phi$ to small perturbations (e.g., noise in the measurements $\bm y$) and, therefore, can be used to assess the quality of the (functional) state reconstruction as $\bm x(t)$ approaches unobservable regions: the larger the $\kappa$, the higher the reconstruction error in the corresponding states. 

In what follows, chaotic systems were numerically integrated using a fourth-order Runge-Kutta integrator with time step ${\rm d}t = 0.01$s for a total simulation time $T = 1100$s, where the initial transient $T_{\rm trans} = 1000$s was discarded and initial conditions were randomly drawn from a normal distribution (i.e., $x_i(0)\sim\mathcal N(0,1)$, $i = 1,\ldots, n$). Codes are publicly available at \url{https://github.com/montanariarthur/NonlinearObservability}.
The symbolic construction of Lie derivatives \eqref{eq.obsv_space} spanning the observable space $\mathcal O(\bm x)$, as well as maps $\bm\Psi$ and $\bm\Phi$, is illustrated in these codes. Note that $\bm\Psi$ is composed by the minimum set of linearly independent functions \eqref{eq.obsv_space} spanning $\mathcal O(\bm x)$. Therefore, as a special case for $q = 1$ and $\nu = n-1$, $\bm\Psi(\bm x) = \mathcal O(\bm x)$.

%--------------------------------------------------------------------
\subsection{Lorenz system: observability and symmetry}
\label{sec.numresult.lorenz}

The well-known Lorenz'63 system is given by
\begin{equation}
\begin{cases}
\dot x_1 = \sigma(x_2-x_1), \\
\dot x_2 = Rx_1 - x_2 - x_1x_3, \\
\dot x_3 = x_1x_2 - bx_3,
\end{cases}
\label{eq.lorenz}
\end{equation}

\noindent
where $(R,\sigma,b) = (28,10,8/3)$  is a set of parameters that leads to a chaotic attractor (Fig.~\ref{fig.summary}c, right). Here, we consider that only the state variable $x_1$ is available for measurement (i.e., $y = h(\bm x) = x_1$). The observability of the pair $\{\bm f,h\}$ can thus be verified through the observability matrix
\begin{equation}
    \gradient\mathcal O(\bm x) =
    \begin{bmatrix}
        1 & 0 & 0 \\
        -\sigma & \sigma & 0 \\
        \sigma^2 + \sigma(\rho - x_3) & -\sigma(\sigma + 1) & -\sigma x_1
    \end{bmatrix} .
\end{equation}

\noindent
Since $\det(\gradient\mathcal O(\bm x)) = -\sigma^2 x_1$, then, following condition \eqref{eq.nonlinearobsvcondition}, the system is locally observable at every state $\bm x\in\R^3$ except at $\bm x_0 = [0 \,\,\, x_2 \,\,\, x_3]^\transp$, where $\operatorname{dim}\{\gradient\mathcal O(\bm x_0)\}<3$. 

\begin{figure}
    \centering
    \includegraphics[width=\columnwidth]{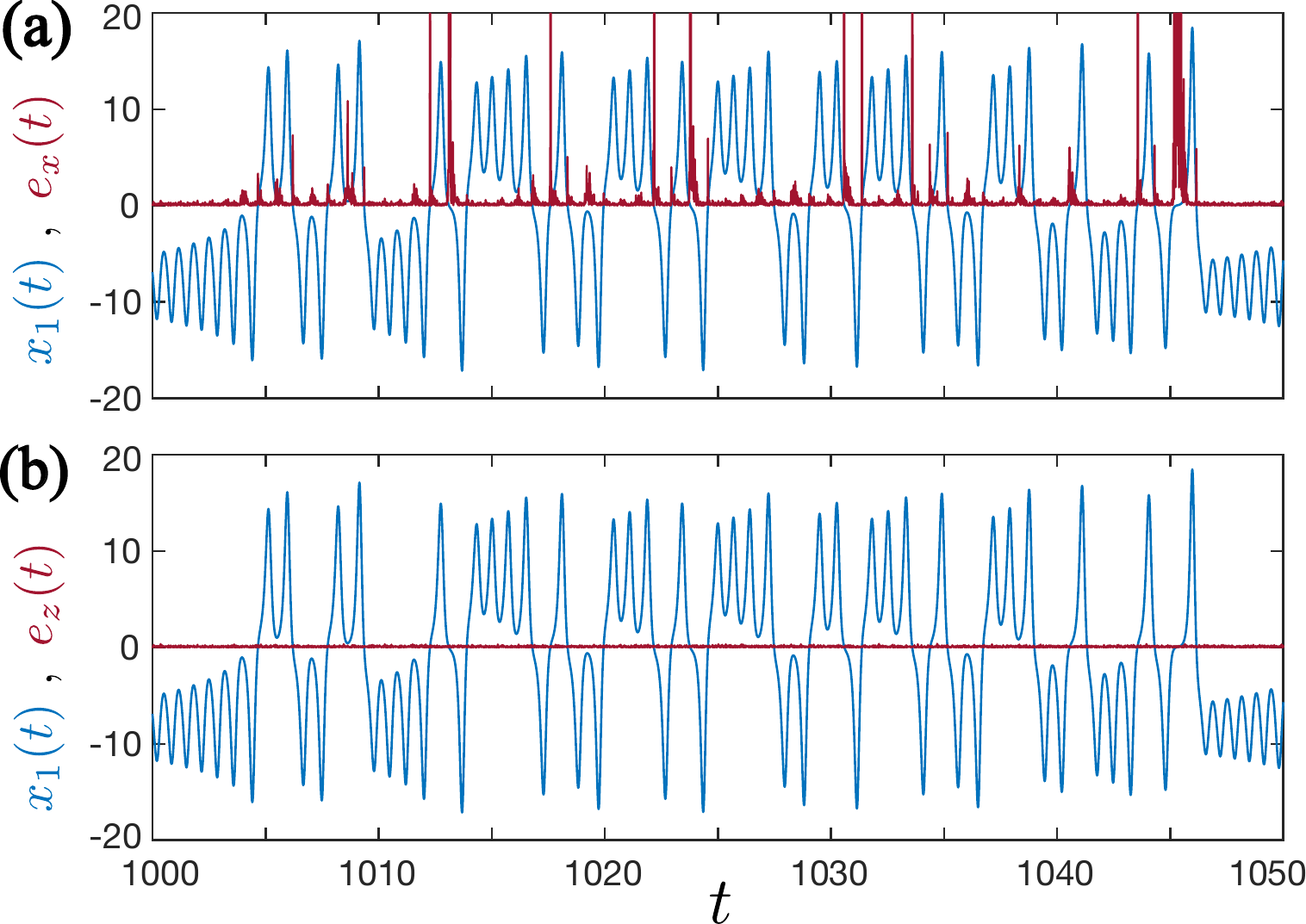}
    \caption{Reconstruction error (red line) of the Lorenz system as a function of time for the (a)~entire state vector $\bm x(t)$ and (b)~functional vector $z(t) = x_1(t)$ sought to be reconstructed. Time series of $x_1(t)$ (blue line) is shown for reference.}
    \label{fig.lorenzreconst}
\end{figure}

Given the differential embedding coordinates $\mathcal E = \{y, \dot y, \ddot y\}$ (Fig.~\ref{fig.summary}b), the entire state vector $\bm x$ can be reconstructed by computing the inverse map $\mathcal O^{-1}(y) : \mathcal E \mapsto \mathcal X$ (where, as a special case, $\bm\Psi(\bm x) = \mathcal O(\bm x)$). Theoretically, existence of this map is not guaranteed only when the system is locally unobservable. In this example, the unobservable subspace corresponds to the exact region in the state space where $x_1 = 0$, which has null Lesbegue dimension. However, in practice, this map degenerates as the system trajectory approaches the neighborhood of $x_1 = 0$, corresponding to a gradual loss of system observability \cite{Letellier2002}. Fig.~\ref{fig.lorenzreconst}a illustrates the reconstruction error $e_x(t) = \norm{\bm x(t) - \hat{\bm x}(t)}$, where $\hat{\bm x} = \mathcal O^{-1}(y,\dot y,\ddot y)$ is the reconstructed (estimated) state vector, considering a small additive noise to the measured time series: $y(t) = h(\bm x,t) + v(t)$, $v(t)\sim\mathcal N(0,10^{-2})$. Note that noise is largely amplified and the reconstruction error increases significantly as the system state approaches the unobservable region ($x_1(t)\rightarrow 0$).

Despite the unobservability at $x_1=0$, the system $\{\bm f,h,g_1\}$ is always functionally observable with respect to the functional $g_1(\bm x) = x_2$, where $\operatorname{row}(\gradient g(\bm x))\subseteq \operatorname{row}(\gradient \mathcal O(\bm x))$, $\forall \bm x\in\R^3$. In practice, $z = g_1(\bm x)$ can be reconstructed from the composition map \eqref{eq.compositionmapPhi} given by
\begin{equation}
    \tilde{\bm x} = 
    \begin{bmatrix}
        y \\ \dot y \\ \xdash[1em] \\ x_3
    \end{bmatrix}
    % =
    % \begin{bmatrix}
    %     x_1 \\ \sigma(x_2 - x_1) \\ \xdash[2em] \\ x_3
    % \end{bmatrix}
    , \,
    \bm T^{-1}(\tilde{\bm x}) = 
    \begin{bmatrix}
        \tilde{x}_1 \\ \tilde x_1 + \frac{\tilde x_2}{\sigma} \\ \xdash[2em] \\ \tilde x_3
    \end{bmatrix}
    , \,
    \tilde g(\tilde{\bm x}) = y + \frac{\dot y}{\sigma}.
\end{equation}

\noindent
Fig.~\ref{fig.lorenzreconst}b shows the reconstruction error $e_z(t) = \norm{z(t) - \hat z(t)}$, where $\hat{z} = \Phi(y,\dot y)$. As expected, since the system is functionally observable for all $\bm x\in\R^3$, the reconstruction error $e_z(t)$ of the functional is not affected by the unobservable region $x_1=0$ and remains bounded (see also Fig.~\ref{fig.summary}d). %The limits of the bounds depend on the characteristics of the measurement noise $v(t)$. 

Fig.~\ref{fig.lorenzcond} presents the coefficients of observability for the Lorenz system. The coefficient of complete observability increases as $x_1(t)\rightarrow 0$, indicating a substantial increase of sensitivity of the local reconstruction map $\bm\Psi^{-1}(\bm x)$ to small perturbations, as observed in the large reconstruction errors $e_x(t)$ for $x_1(t)\rightarrow 0$ in Fig.~\ref{fig.lorenzreconst}a. On the other hand, the coefficient of functional observability remains well-conditioned and constant throughout the entire attractor, which is supported by the insensitivity to noise in the reconstruction error $e_z(t)$ in Fig.~\ref{fig.lorenzreconst}b. These results demonstrate that these coefficients can provide proxy indicators of the ``practical'' consequences of the lack of (functional) observability of these systems as the state approaches (functionally) unobservable regions, where local reconstruction maps become highly sensitivity to noise and, therefore, fail to provide an accurate reconstruction of the original state space.

\begin{figure}
    \centering
    \includegraphics[width=\columnwidth]{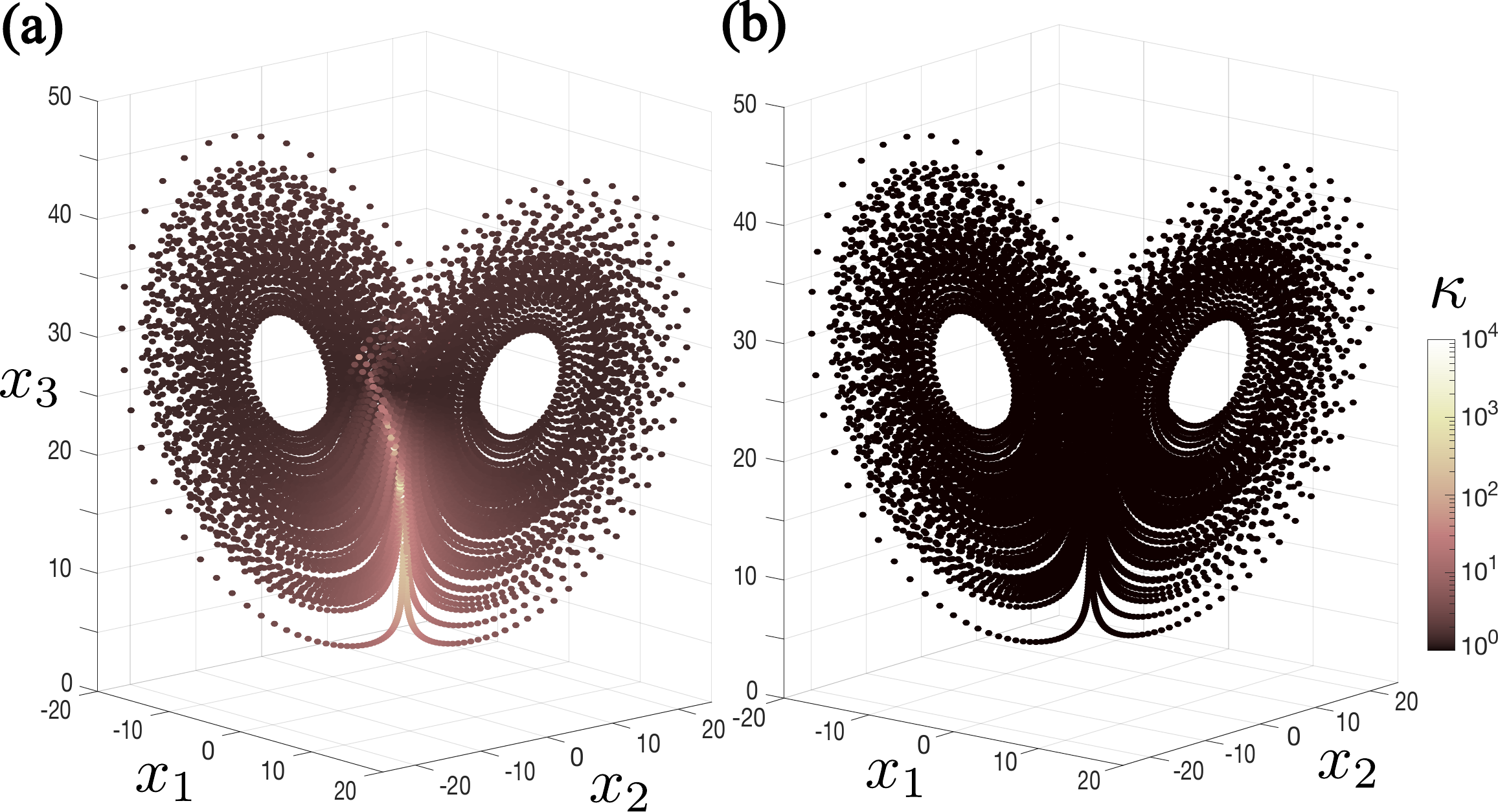}
    \caption{Coefficient of (a)~complete observability and (b)~functional observability ($g_1(\bm x)=x_2$) computed over the state space of the Lorenz attractor.}
    \label{fig.lorenzcond}
\end{figure}

\begin{figure}
    \centering
    \includegraphics[width=\columnwidth]{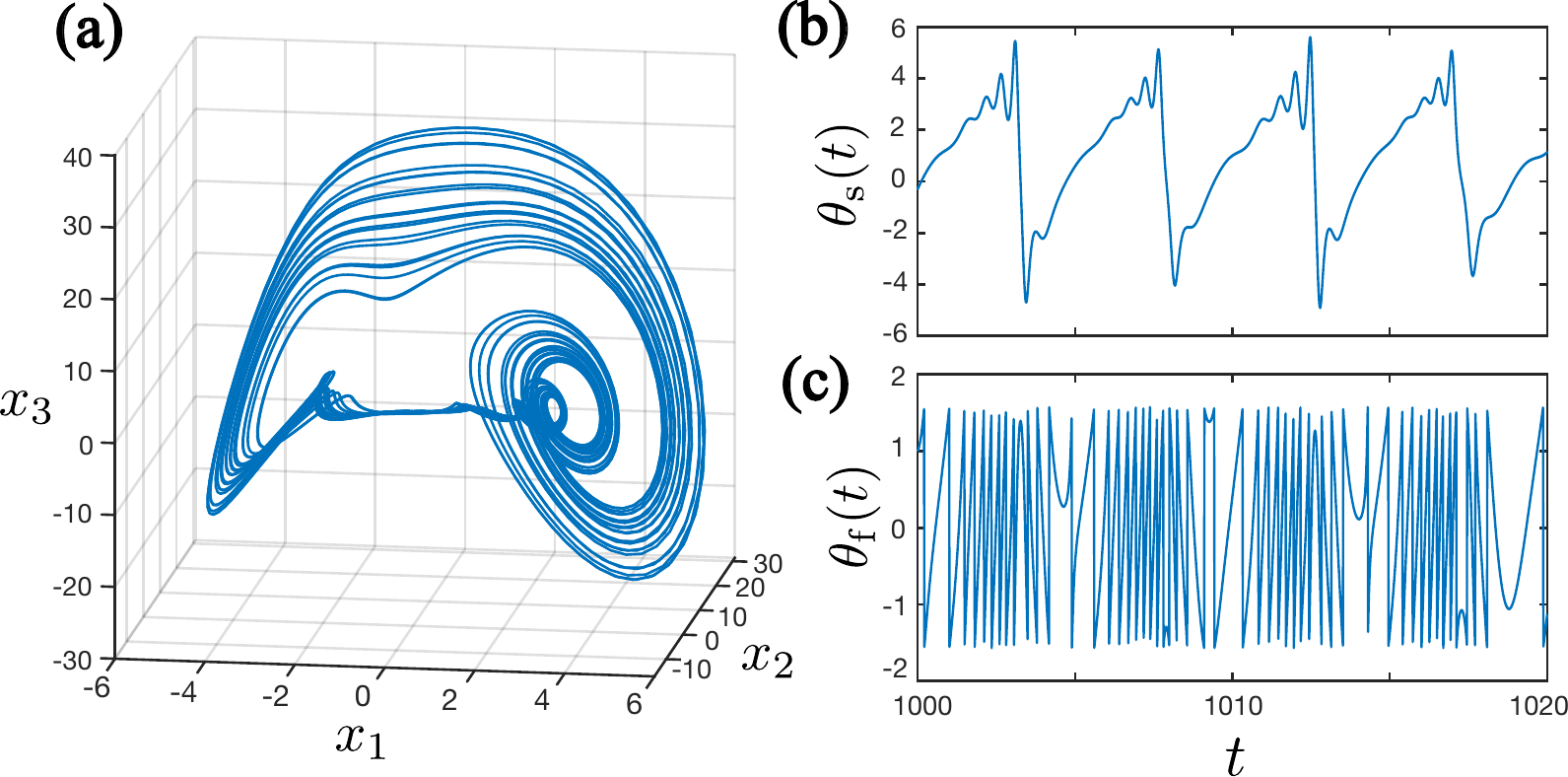}
    \caption{Cord system. (a) State space of the Cord attractor. (b) Time series of the slow phase $\theta_{\rm s}(t)$. (c) Time series of the fast phase $\theta_{\rm f}(t)$.}
    \label{fig.cord}
\end{figure}

The Lorenz system is marked by a clear relation between observability and symmetry. Since $h(\bm x)=x_1$ is directly measured and $g(\bm x) = x_2$ is functionally observable, one can observe (reconstruct) the dynamics in the $(x_1,x_2)$ plane of the Lorenz attractor. The functional observability of this plane is directly related to the global invariance of the Lorenz attractor under the map $[x_1 \,\, x_2 \,\, x_3]^\transp \mapsto [-x_1 \,\, -x_2 \,\,\, x_3]^\transp$ \cite{Letellier2002}. Given that $x_3$ is invariant under this symmetry, one can only distinguish which ``wing'' of the chaotic attractor the system state belongs to at a given time instant $t$ by accurately observing variables $x_1$ and $x_2$ (Fig~\ref{fig.summary}d). Therefore, the functional observability of the triple $\{\bm f,h,g_1\}$ provides the necessary and sufficient information for this characterization based on the measured time series $y(t)$. Moreover, it is evident that the lack of complete observability in the system $\{\bm f,h,g_1\}$ is due to variable $x_3$, which can be rigorously verified by noting that the functional observability condition \eqref{eq.functobsvcondition} is only satisfied for a triple $\{\bm f,h,g_2\}$, $g_2(\bm x)=x_3$, if $x_1\neq 0$.

%--------------------------------------------------------------------
\subsection{Cord system: fast and slow dynamics}
\label{sec.numresult.cord}

The Cord system, a variation of the Lorenz'84 system, is given by \cite{Letellier2012}
\begin{equation}
\begin{cases}
\dot x_1 = -x_2 - x_3 - ax_1 + aF, \\
\dot x_2 = x_1x_2 - bx_1x_3 - x_2 + G, \\
\dot x_3 = bx_1x_2 + x_1x_3 - x_3,
\end{cases}
\label{eq.cord}
\end{equation}

\noindent
where $(a,b,F,G) = (0.258,4.033,8,1)$. The chaotic attractor is illustrated in Fig.~\ref{fig.cord}a. The system dynamics is marked by two oscillation modes with a clear timescale separation \cite{Freitas2020}. Oscillations in the slow timescale can be approximately monitored by the ``slow phase'' variable $\theta_{\rm s}=x_1$ (Fig.~\ref{fig.cord}b), in which a full revolution of the system is completed every time the trajectory approximates the cord filament close to the origin (defining the Poincaré section $\mathcal P = \{\bm x: x_1 = 0, \dot x_1>0 \}$) \cite{Freitas2018}. Oscillations in the fast timescale, on the other hand, can be monitored by the ``fast phase'' variable $\theta_{\rm f} = \tan^{-1}(x_2/x_3)$ (Fig.~\ref{fig.cord}c).

Here, we consider the measured time series $y = h(\bm x) = x_2$ and that the slow and fast phase variables are the functionals sought to be reconstructed, i.e., $g_1(\bm x) = \theta_{\rm s}$ and $g_2(\bm x) = \theta_{\rm f}$. Fig.~\ref{fig.cordcond} shows the coefficients of observability for the Cord system. Full-state reconstruction of the Cord system from the measured time series is not possible for a considerable range of states in the system trajectory (Fig.~\ref{fig.cord}a), defined by the plane
\begin{equation}
\begin{aligned}
    \det({\mathcal O(\bm x)}) =& \,\,
    b^2 x_3 (-a F + 2 x_1^2 + x_2 + 2 x_3) + 2 b^3 x_1^2 x_2
    \\ 
    & - b (G x_1 - a F x_2 + x_2^2 + 3 x_1 x_3) + x_2^2 
    \\
    =& \,\, 0
    ,
    \label{eq.corddet}
\end{aligned}
\end{equation}

\noindent
and is expected to be ill-conditioned when the system state is close to the vicinity of this plane. The unobservable plane can be visualized in the $(x_2,x_3)$ section of the attractor in Fig.~\ref{fig.cordcond}a.

\begin{figure}
    \centering
    \includegraphics[width=\columnwidth]{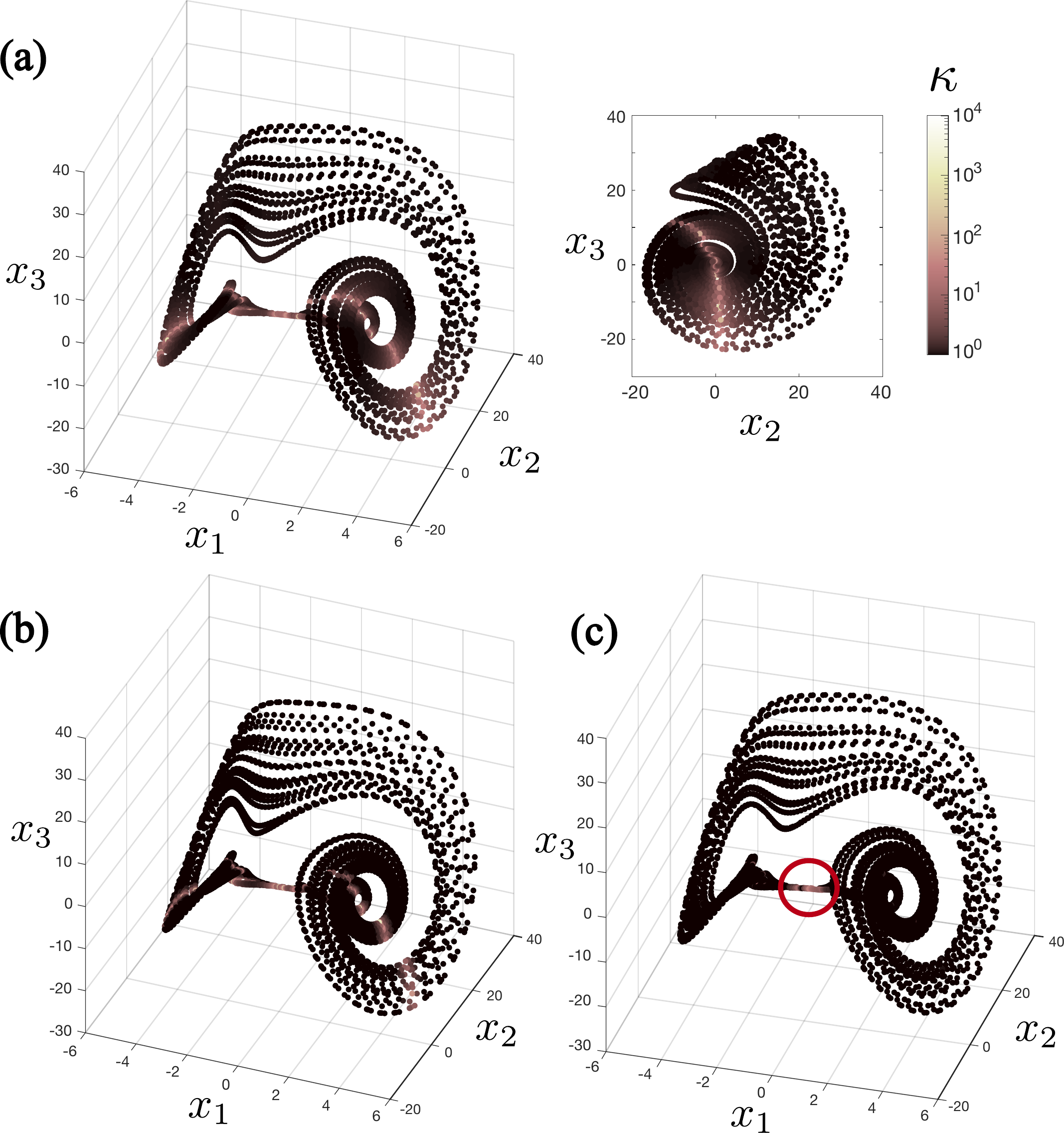}
    \caption{Coefficient of (a)~complete observability, (b)~functional observability with respect to $g_1(\bm x)$, and (c)~functional observability with respect to $g_2(\bm x)$ computed over the state space of the Cord attractor.}
    \label{fig.cordcond}
\end{figure}

Similarly to the full-state reconstruction problem, the reconstruction of the system's slow timescale (i.e., $g_1(\bm x)$) from time-series data of a state variable dominated by a fast timescale (e.g., $x_2$) is hampered by the lack of observability in a large subspace of the state space, as indicated by the regions with very large coefficients of functional observability in Fig.~\ref{fig.cordcond}b. Contrariwise, reconstruction of the fast timescale (i.e., $g_2(\bm x)$) from $x_2$ is well-conditioned throughout the entire system trajectory, except as $(x_2,x_3)\rightarrow (0,0)$ (Fig.~\ref{fig.cordcond}c, red circle). The lack of functional observability at this singularity region in the attractor is not surprising: it corresponds exactly to the region in which the fast phase variable $\theta_{\rm f}$ is a (locally) ill-defined function,
\begin{equation}
    \gradient \theta_{\rm f}(\bm x) =
    \begin{bmatrix}
        0 & -\frac{x_3}{\sqrt{x_2^2 + x_3^2}} & \frac{x_2}{\sqrt{x_2^2 + x_3^2}}
    \end{bmatrix},
\end{equation} 

\noindent
and the fast phase ``collapses'', undergoing an inversion of its rotational direction \cite{Letellier2012}. 

This example illustrates that, even though the system may have a (relatively) large unobservable region $\mathcal X_{\rm u}\subset\mathcal X$, one may find that, even in this unobservable region, the system can still be functionally observable with respect to some functional $\bm g(\bm x)$ aside from a significantly smaller subregion $\mathcal X_{\rm fu}'\subset\mathcal X_{\rm u}$. In this example, the region of interest $\mathcal X$ is the Cord attractor $\mathcal A$, the ``completely'' unobservable region $\mathcal X_{\rm u}$ is the 2-dimensional plane defined by \eqref{eq.corddet}, and the functionally unobservable region is the 1-dimensional line $\mathcal X_{\rm fu} = \{(x_1,0,0) \, | \, \bm x\in\mathcal A\}$.

%--------------------------------------------------------------------
\subsection{Hindmarsh-Rose system: a neuron model}
\label{sec.numresult.neuron}

Building up from the chaotic benchmarks, we now consider a phenomenological model of neuron dynamics given by the Hindmarsh-Rose (HR) model \cite{Hindmarsh1984}:
\begin{equation}
\begin{cases}
\dot x_1 = x_2 - ax_1^3 + bx_1^2 - x_3 + I, \\
\dot x_2 = c - dx_1^2 - x_2, \\
\dot x_3 = r(sx_1 - x_{\rm R}) - x_3,
\end{cases}
\label{eq.hrneuron}
\end{equation}

\noindent
where $x_1$ is the membrane potential, $x_2$ is the fast recovery current, and $x_3$ is the slow adaptation current. Providing both a simplification of the biophysical Hodgkin-Huxley neuronal model and a generalization of the FitzHugh-Nagumo model, the HR model can reproduce a wide range of dynamical behaviors, including quiescence and (irregular) spiking and bursting \cite{Storace2008}. Moreover, depending on the bifurcation parameters, this system can also shift to chaotic regimes, as investigated both computationally \cite{Storace2008} and experimentally \cite{Gu2013}. Here, we consider the set of parameters lying in the chaotic regime: $(a,b,c,d,I,r,s,x_{\rm R}) = (1,3,1,5,3.25,0.001,4,-8/5)$. 

The measurement functions $h_1(\bm x) = x_1$, $h_2(\bm x) = x_2$, and $h_3(\bm x) = x_3$ yield observability matrices with determinants given by, respectively, $\det(\mathcal O_1(\bm x)) = r-1$, $\det(\mathcal O_2(\bm x)) = 4d^2 x_1^2$, and $\det(\mathcal O_3(\bm x)) = r^2s^2$. Thus, for the considered set of parameters, $\{\bm f,h_1\}$ and $\{\bm f,h_3\}$ are locally observable everywhere, while $\{\bm f,h_2\}$ becomes locally unobservable only at $x_1 = 0$ \cite{Aguirre2017}. Accordingly, the coefficients of complete observability show a considerable increase as $x_1\rightarrow 0$ (Fig.~\ref{fig.hrneuroncond}a).

\begin{figure*}
    \centering
    \includegraphics[width=\linewidth]{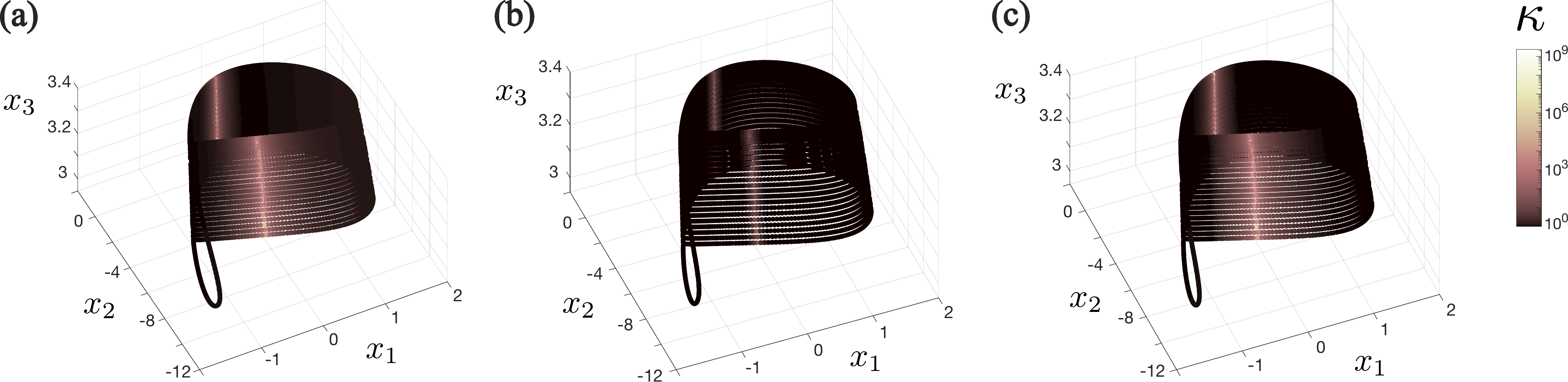}
    \caption{Coefficients of (a) complete observability, (b) functional observability with respect to $g_1(\bm x)$, and (c) functional observability with respect to $g_2(\bm x)$ computed over the state space of the HR neuron model, considering the measured time series $y = h_2(\bm x)$. Simulations are presented for $(T,T_{\rm trans}) = (2500, 1500)$.}
    \label{fig.hrneuroncond}
\end{figure*}

\begin{figure}
    \centering
    \includegraphics[width=0.95\columnwidth]{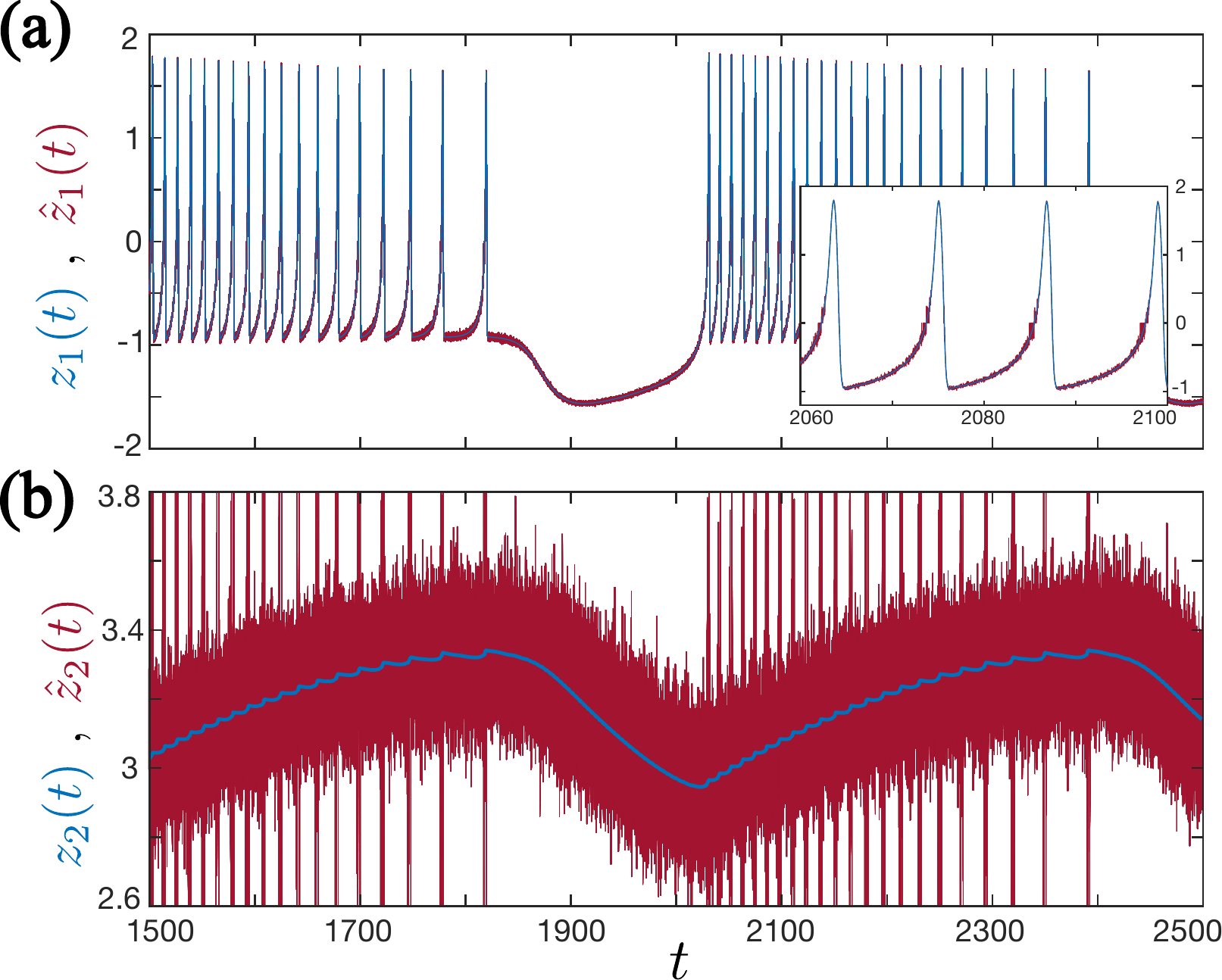}
    \caption{Reconstructed functional states (a) $\hat z_1$ and (b) $\hat z_2$ as a function of time (red lines) for the HR neuron model. Time series of $z_1=g_1(\bm x)$ and $z_2=g_2(\bm x)$ (blue lines) are shown for reference. Simulations are presented for $(T,T_{\rm trans}) = (2500, 1500)$.}
    \label{fig.hrneuronerror}
\end{figure}

One might wonder if, despite the lack of complete observability at $x_1=0$, the system $\{\bm f,h_2,g_i\}$ is still functionally observable with respect to, for example, $g_1(\bm x) = x_1$ or $g_2(\bm x) = x_3$. However, unlike the previous examples, the HR model remains locally unobservable at $x_1 = 0$ with respect to both functionals, as observed in the coefficients of functional observability shown in  Fig.~\ref{fig.hrneuroncond}b,c. Nevertheless, note that the neighborhood of $x_1 = 0$ where the reconstruction map is (locally) ill-conditioned is substantially smaller for $\{\bm f,h_2,g_1\}$ compared compared to $\{\bm f,h_2,g_2\}$. These results suggest that reconstruction of $g_1(\bm x)$ is more reliable than $g_2(\bm x)$ in the presence of small perturbations as $x_1\rightarrow 0$. 

Examining the local maps yield
\begin{align}
    \gradient g_1 \gradient\mathcal O_2^{-1} &= 
    \begin{bmatrix}
        -\frac{1}{2dx_1} & -\frac{1}{2dx_1} & 0
    \end{bmatrix},
    \label{eq.hrneuron.localmap1}
    \\
    \gradient g_2 \gradient\mathcal O_2^{-1} &= 
    \begin{bmatrix}
        1-\frac{\xi}{2dx_1^2} & \frac{x_1 + \xi}{2dx_1^2} & \frac{1}{2dx_1} 
    \end{bmatrix},
     \label{eq.hrneuron.localmap2}
\end{align}

\noindent
where $\xi = I + x_2 - x_3 - 4ax_1^3 + 3bx_1^2$. The presence of the terms $x_1$ and $x_1^2$ in the denominator of Eqs.~\eqref{eq.hrneuron.localmap1} and \eqref{eq.hrneuron.localmap2} elucidate the results shown in Fig.~\ref{fig.hrneuroncond}b,c. The sensitivity to small perturbations in the reconstruction of functional $g_1(\bm x)$ is only inversely proportional to the distance between $x_1$ and the unobservable region, whereas the sensitivity of the reconstruction of $g_2(\bm x)$ is inversely proportional to the \textit{quadratic} of this distance\textemdash leading to a highly ill-conditioned map for $|x_1|\ll 1$. This theoretical (local) analysis is also supported by computing the reconstruction maps $\bm\Phi : \mathcal E \mapsto\mathcal G(\mathcal X)$ and evaluating the corresponding reconstruction performance for each functional. Fig.~\ref{fig.hrneuronerror} shows that, in the presence of small measurement noise $v(t)\sim\mathcal N(0,0.01)$, reconstruction of $\hat z_2 = g_2(\bm x)$ yields very poor results, with a high root-mean-square error (RMSE) of 0.4006, compared to the RMSE of 0.0242 for the reconstructed vector $\hat z_1 = g_1(\bm x)$.

As in the Cord example, reconstruction of the slow timescale dynamics ($g_2(\bm x) = x_3$ in the HR model) from time-series data corresponding to a variable dominated by the fast timescale ($h_2(\bm x) = x_2$) is marked by the presence of unobservable regions which significantly hamper the quality of the reconstruction in the vicinity of these regions. On the other hand, 
%despite this lack of functional observability with respect to the variables in the slow timescale, 
measuring a variable dominated by the fast timescale can still provide accurate reconstruction of other variables dominated by the same timescale ($g_1(\bm x) = x_1$). This relation between the timescale separation and functional observability of a system, with respect to variables belonging to the same or different timescales than the measured variable, can be observed both in the Cord and HR models.

%====================================================================
\section{Early warning of seizures}
\label{sec.Epileptor}

\begin{figure*}
    \centering
    \includegraphics[width=0.7\linewidth]{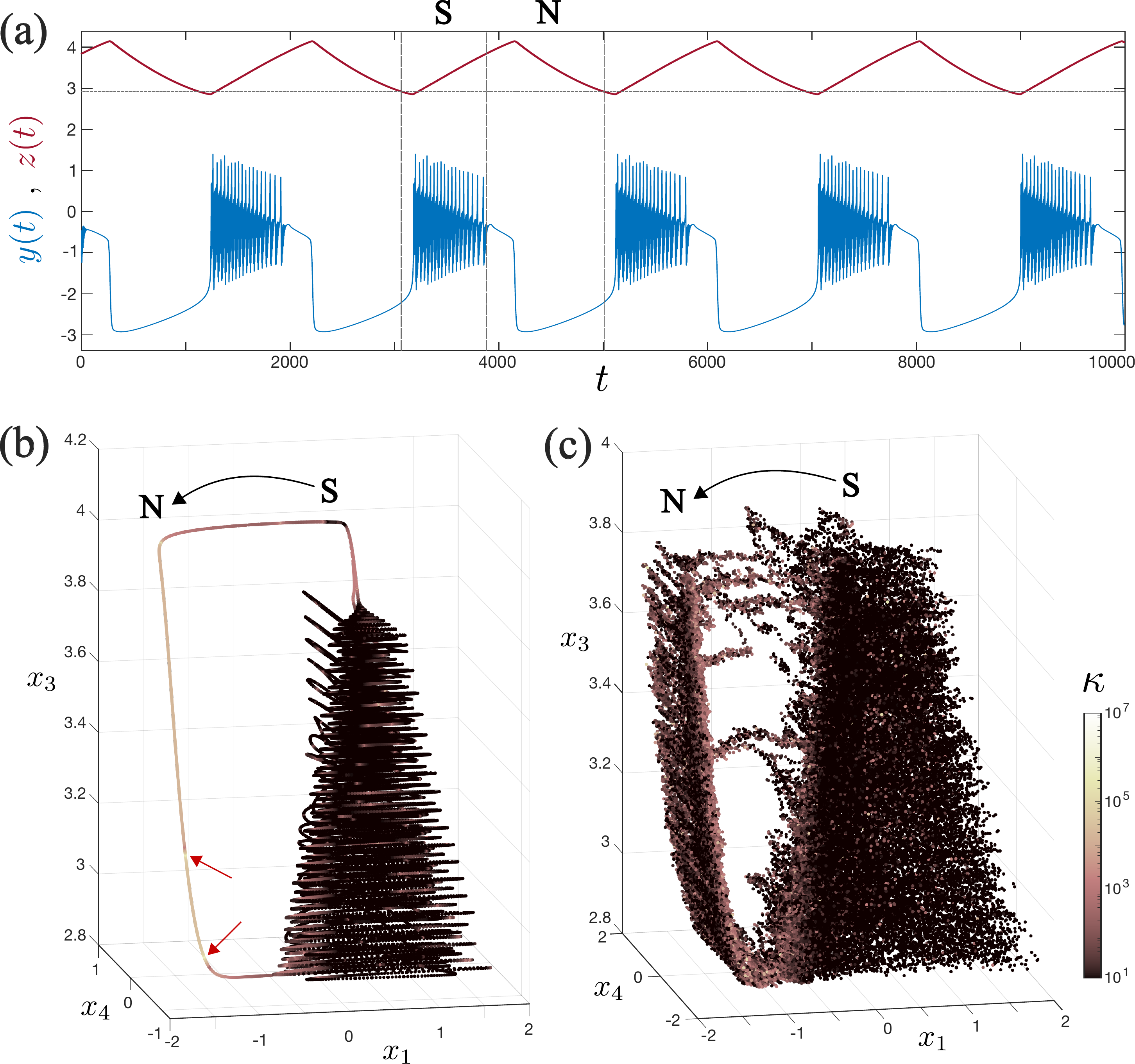}
    \caption{Functional observability of the Epileptor model. 
    (a)~Time series of the measured signal $y(t)$ (modeling EEG data) and the functional $z(t)$ (slow permittivity variable) for the deterministic model.
    (b,c)~Coefficients of functional observability computed over the state space of the Epileptor's attractor, considering (b) deterministic and (c) stochastic representations of the Epileptor.
    Red arrows point highly functionally unobservable states in the attractor. Transitions from seizure (S) to normal (N) regime in the Epileptor are indicated in the plots.
    Simulations are presented for $(T,T_{\rm trans},d{\rm t}) = (10^4,0,0.01)$ and $\bm x(0) = [0 \,\, {-5} \,\, 3 \,\, 0 \,\, 0]^\transp$. For the stochastic model, the model \eqref{eq.Epileptor} was numerically integrated using Euler-Maruyama method where additive process noises $\mathcal N(0,0.01)$ and $\mathcal N(0,0.1)$ were introduced to subsystems $(x_1,x_2)$ and $(x_4,x_5)$, respectively.}
    \label{fig.epileptorcond}
\end{figure*}

Characterizing and predicting epileptic seizures are long-standing challenges in clinical neuroscience \cite{mormann2007seizure,cook2013prediction}. Accurate and interpretable methods for the prediction of seizure events will drastically improve epilepsy management, providing early warnings to alert patients or trigger interventions \cite{kuhlmann2018epilepsyecosystem}. On top of black-box and data-greedy deep learning algorithms, dynamical-based topological analysis can concur in discovering universal routes to epilepsy and foster new methods for early warning, many of which can be based on the embedding of time-series data \cite{yuan2008comparison,Jirsa2014}.
We investigate the observability and embedding properties of such application by considering a dynamical model describing seizure dynamics in the brain. The model, termed Epileptor \cite{Jirsa2014}, involves bifurcation dynamics to reproduce resting, spiking, and bursting behaviors observed in electroencephalogram (EEG) signals, modeling the multiple timescale oscillations recorded during epileptic seizures.
The Epileptor is defined by \cite{Jirsa2014}
\begin{equation}
\begin{cases}
\dot x_1 = x_2 - f(x_1,x_4) - x_3 + I_{1},\\
\dot x_2 = r_2 - 5x_1^2 - x_2,\\
\dot x_3 = \frac{1}{\tau_0} (4(x_1 - r_1) - x_3),\\
\dot x_4 = -x_5 + x_4 - x_4^3 + I_{2} + 0.002g(x_1) - 0.3(x_3 - 3.5),\\
\dot x_5 = \frac{1}{\tau_2} (-x_5 + f_2(x_4)),
\end{cases}
\label{eq.Epileptor}
\end{equation}

\noindent
where $(r_{1}, r_{2}, I_{1}, I_{2},\gamma) = (-1.6, 1, 3.1, 0.42, 0.01)$ are the system parameters, $(\tau_0,\tau_2) = (2857, 10)$ are the timescale constants, and the coupling functions are given by
\begin{equation}
    \begin{aligned}
        g(x_1) &= \int_{t_0}^t \exp(-\gamma (t-\tau))x_1(\tau){\rm d}{\tau}, \\
        f_1(x_1,x_4) &= 
        \begin{cases}
        x_1^3 - 3x_1^2, \quad & x_1 < 0, \\
        (x_4 - 0.6(x_3 - 4)^2)x_1, \quad & x_1 \geq 0, \\
        \end{cases} \\
        f_2(x_4) &= 
        \begin{cases}
        0, \quad & x_4 < -0.25, \\
        6(x_4 + 0.25), \quad & x_4 \geq -0.25. \\
        \end{cases}
    \end{aligned}
\end{equation}

\noindent
This model consists of three subsystems with different timescales: $(x_1,x_2)$ governs the system's oscillatory behavior, $(x_4,x_5)$ introduces the spikes and wave components typical in seizure-like events, and $x_3$ represents a slow permittivity variable that determines how close the system is to the seizure threshold. Due to the slow-fast timescale separation induced by $\tau_0$, $x_3$ is usually interpreted as a quasi-steady state parameter \cite{Jirsa2014}, enabling bifurcation analysis. %Note that $g(x_1)$ can be represented as a differential equation by introducing an auxiliary state variable $x_6$ \cite{Jirsa2014}. 

%--------------------------------------------------------------------
\subsection{Functional observability analysis}
\label{sec.epileptor.obsv}

Monitoring the permittivity variable $x_3$ provides an early-warning signal of a dynamical transition from normal to seizure states in the Epileptor model. Despite the phenomenological nature of the model, this permittivity variable is most likely related to slowly changing biophysical parameters (e.g., extracellular processes or ionic concentrations) \cite{Jirsa2014}. Given that such parameters are hardly measurable in biomedical setups, we investigate whether it is possible to infer the permittivity variable (i.e., the functional $z = g(\bm x) = x_3$) from more easily accessible measurements, such as EEG recordings of seizure-like events (modeled as the measurement signal $y = h(\bm x) = x_1 + x_4$ due to its close resemblance to actual EEG data \cite{Jirsa2014}). Fig.~\ref{fig.epileptorcond}a illustrates the dynamics of the functional $z(t)$ and output $y(t)$. Under the assumption that the Epileptor is a proper representation of the underlying process, a functional observability analysis of model \eqref{eq.Epileptor} can establish if it is feasible to reconstruct this functional and, therefore, provide an early-warning signal of seizure events from EEG data.%\textemdash a long-pursued topic in the data-driven community \cite{}. 

Due to the model complexity, an analytical derivation of the functionally unobservable regions of the Epileptor model is hardly tractable. Instead, Fig.~\ref{fig.epileptorcond}b,c presents the coefficients of functional observability of the triple $\{\bm f,h,g\}$ computed over the system's attractor. The system alternates between two regions of the attractor, the normal state and the seizure state, as $x_3$ crosses predetermined thresholds marking bifurcation points (i.e., points in the parameter space where qualitative changes in system dynamics occur). While the coefficients $\kappa$ are fairly well-conditioned in the seizure region of the attractor ($\kappa<10^2$), the normal region has relatively larger coefficients ($\kappa \approx 10^{4}$) with two remarkable  ill-conditioned singularities ($\kappa > 10^{7}$) highlighted in Fig.~\ref{fig.epileptorcond}b. This indicates the presence of two functionally unobservable states in the normal region of the Epileptor's state space, one of them located exactly at the saddle-node bifurcation point from normal to seizure regime ($x_3 \approx 2.9$, $\dot x_3<0$ \cite{el2020epileptor}). Introducing linear additive process noise to the model \eqref{eq.Epileptor} promotes a larger exploration of system's state space, uncovering other functionally unobservable singularities in the normal region (Fig.~\ref{fig.epileptorcond}c), including a few in the seizure region. Nonetheless, the analysis remains qualitatively similar between the deterministic and stochastic systems: both show considerably larger values of $\kappa$ in the normal region compared to the seizure region (see also Fig.~\ref{fig.ewsepileptor}b). Consequently, high errors are expected in the reconstruction of the permittivity variable from the measured signal $y(t)$ during the normal regime of the Epileptor.

%--------------------------------------------------------------------
\subsection{Early-warning signals and observability}

At first, large coefficients of functional observability in the Epileptor's normal region indicate that reconstructing the slow permittivity variable from EEG data is particularly challenging. However, our analysis established an interesting relation between the Epileptor's observability and topological features: normal (seizure) regions of the attractor correspond to regions with large (small) coefficients of functional observability. This relation can be explored to develop early-warning indicators of seizure-like events in simulated and empirical data.
% Despite this outcome, we now show an interesting link between observability and embedding can be explored to develop early-warning indicators of seizure-like events in simulated and experimental data.

\begin{figure}
    \centering
    \includegraphics[width=0.9\columnwidth]{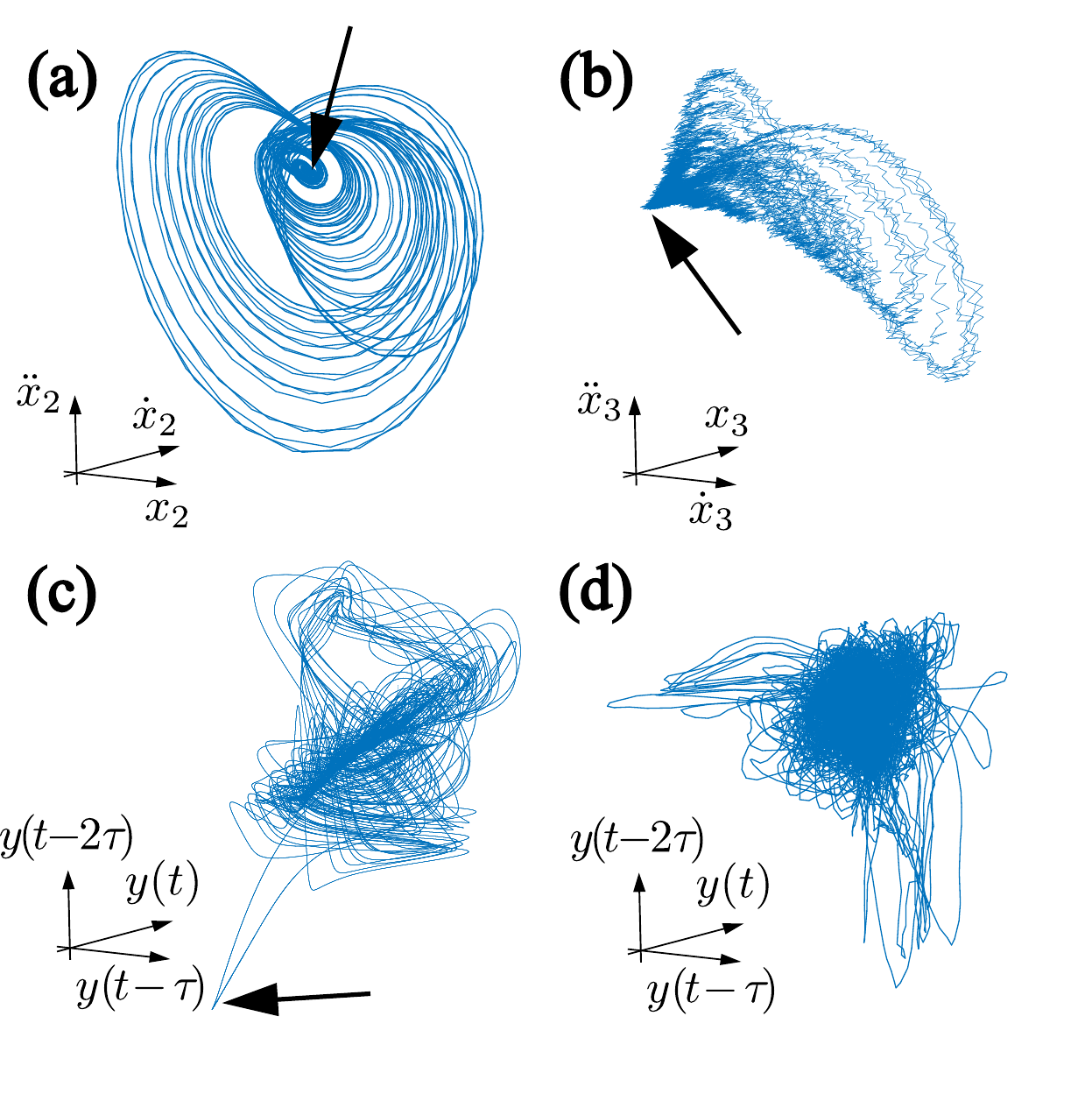}
    \caption{Embedded attractors of different dynamical systems. Differential embedding coordinates of the (a) Cord system and (b) R\"ossler system. Time-delay embedding coordinates of the (c) Epileptor model and (d) human EEG data.
    Unobservable regions in subplots (a,b,c) are pointed out by arrows.}
    \label{fig.embeddingattractors}
\end{figure}

Typical early-warning signals of critical transitions studied in the literature, such as variance and autocorrelation, are computed from time-series data \cite{Scheffer2009}. Evaluating the system's observability, on the other hand, requires prior knowledge of the system's equations \eqref{eq.nonlinearsys}, often absent for real-world systems. Theory states that unobservable regions in the state space are associated to the loss of dimension of the observable space (i.e., condition \eqref{eq.functobsvcondition} does not hold). As a consequence, closer to unobservable regions, embedded trajectories squeeze into a small low-dimensional neighborhood due to the loss of diffeomorphism between the embedding space and the original state space \cite{Letellier2005}. This phenomenon is illustrated in Fig.~\ref{fig.embeddingattractors} for the embedded attractors of different dynamical systems with poorly observable regions as well as real-world data. Such local topological feature can be assessed by monitoring the smallest singular value $\sigma_{d_e}$ computed from an embedded time series with embedding dimension $d_e$ (see Appendix~\ref{app.timeseriesobsv} for details). As $\sigma_{d_e}\rightarrow 0$, the effective dimension of the embedded time series drops, implying that the diffeomorphism between the embedded and original attractors is not locally preserved (and, therefore, the system is locally unobservable). In what follows, we apply the coefficient $\sigma_{d_e}$, hereby referred to as ``time series-based singular value decomposition'' (tSVD), as a proxy measure of the system's observability computed from time-series data, which, as we show next, has a high correlation with the coefficients of observability (Fig.~\ref{fig.ewsepileptor}d).

\begin{figure*}[tb]
    \centering
    \includegraphics[width=0.98\linewidth]{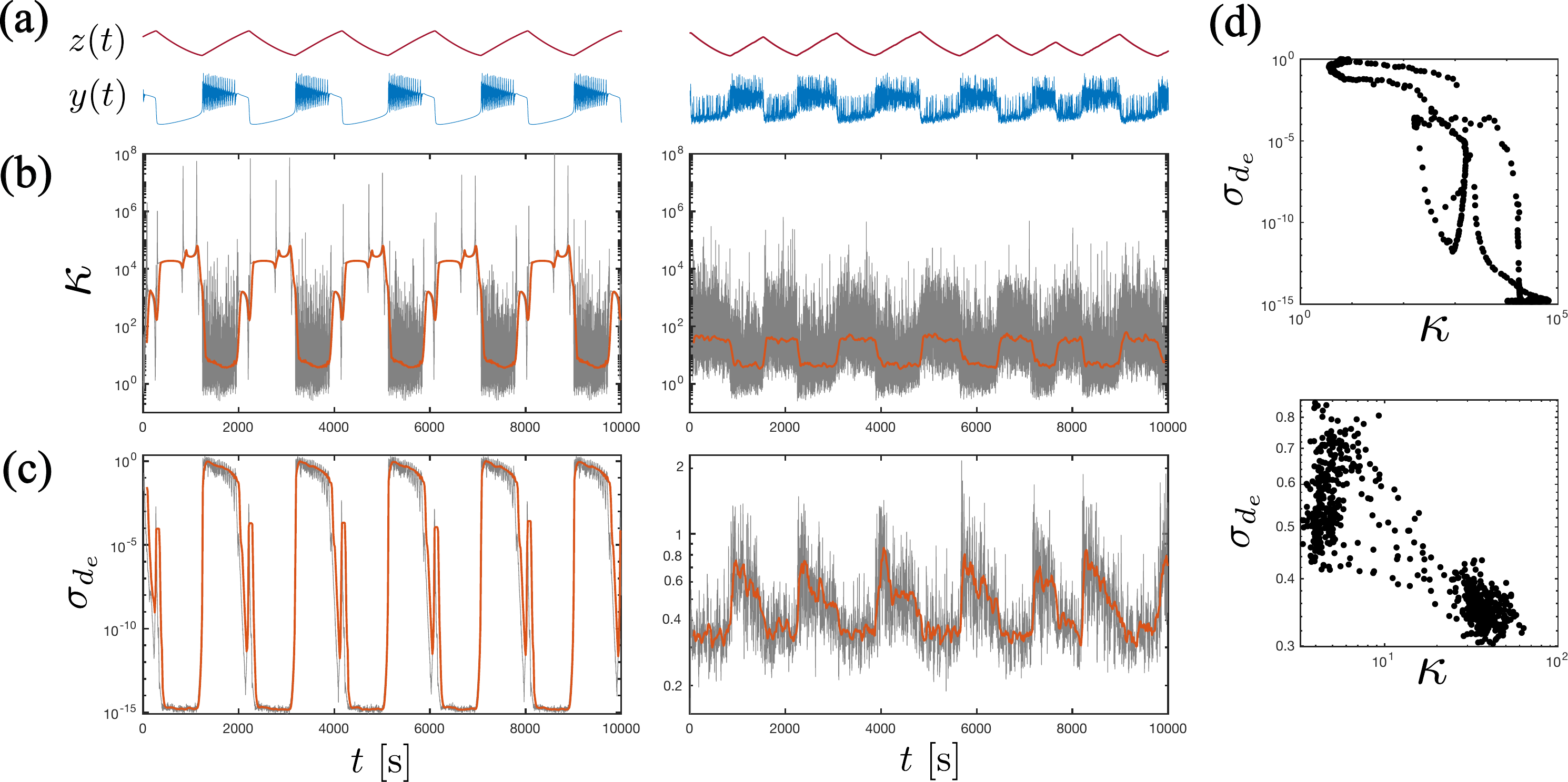}
    \caption{Observability and early-warning signal of the deterministic (left) and stochastic (right) Epileptor models. 
    (a) Time series of the functional $z(t)$ and the measured signal $y(t)$. In the stochastic case, noise can trigger more transitions for the same time interval.
    (b) Coefficient of functional observability $\kappa$ as a function of time. 
    (c) tSVD $\sigma_{d_e}$ as a function of time, computed over a moving time-series window with length $N = 5$s and embedding parameters $(d_e,\tau) = (5, 0.1{\rm s})$. 
    Orange curves show the smoothed coefficients computed over a moving average window with length $N_{\rm avg}=150$s, sampled every $30$s.
    (d) Correlation between $\kappa$ and $\sigma_{d_e}$ for the deterministic (top) and stochastic (bottom) system.}
    % See Fig.~\ref{fig.epileptorcond} for other simulation details.}
    \label{fig.ewsepileptor}
\end{figure*}

Fig.~\ref{fig.ewsepileptor}a--c shows the coefficient of functional observability and the tSVD for a given time series $y(t)$. For the deterministic model, $\kappa$ and $\sigma_{d_e}$ are anti-correlated: as $\kappa$ increases (decreases) during normal (seizure) regimes of the Epileptor, the tSVD decreases (increases). The anti-correlation is confirmed by a Pearson's correlation index of $\rho = -0.96$ between both coefficients in logarithmic scale (Fig.~\ref{fig.ewsepileptor}d, top). 
As expected, when the Epileptor switches to the normal region, which is poorly functionally observable ($\kappa>10^{4}$), the smallest singular value $\sigma_{d_e}$ tends to zero ($\sigma_{d_e}\approx 10^{-16}$), implying an effective loss of dimension of the embedded time series. For the stochastic model, the broader state-space exploration of this system yields a higher variation of $\kappa$ and $\sigma_{d_e}$. Nevertheless, the same anti-correlation between $\kappa$ and $\sigma_{d_e}$ can be observed by smoothing the coefficients over a moving average window. On average, $\kappa$ ($\sigma_{d_e}$) increases (decreases) during the normal regime of the Epileptor, yielding a Pearson's correlation index of $\rho = -0.82$ (Fig.~\ref{fig.ewsepileptor}d, bottom).

\begin{figure}
    \centering
    \includegraphics[width=0.95\columnwidth]{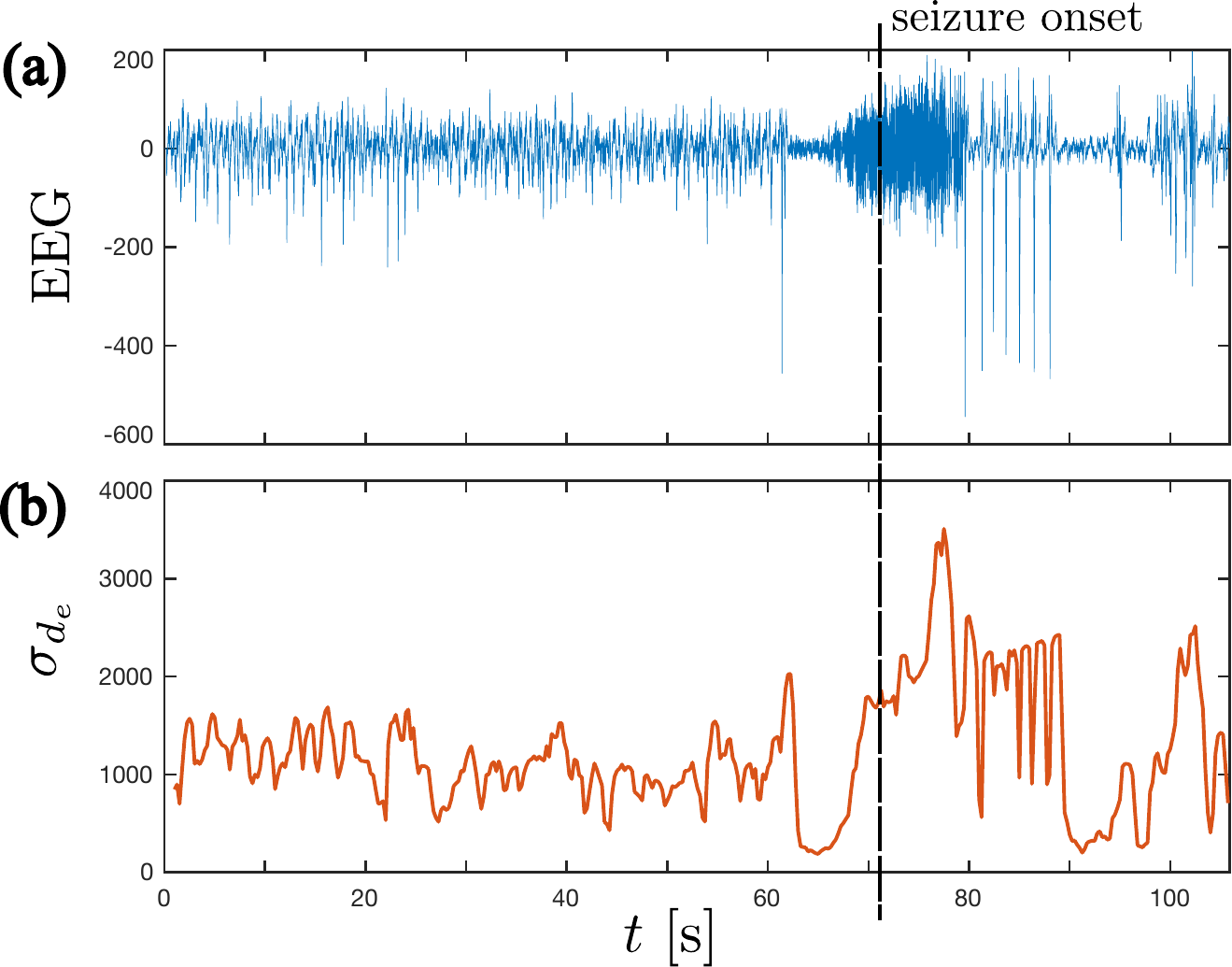}
    \caption{Early warning of seizure events in human intracranial EEG data. (a) Intracranial EEG data (channel 2) of patient 10 in database \cite{KAROLY2018}. (b) tSVD $\sigma_{d_e}$ computed over a moving time-series window with parameters $(d_e,\tau,N) = (5,0.05{\rm s},1{\rm s})$. Seizure onset (informed by expert opinion) is marked by the dashed line.}
    \label{fig.ewshuman}
\end{figure}

The behavior of the tSVD measure is remarkably aligned with that of typical early-warning signals used to detect critical transitions from time-series data \cite{kuehn2011mathematical}, including those for seizure warning from EEG data \cite{Maturana2020}. Indeed, there is a sharp increase of the tSVD close to the dynamical transition from normal to seizure state (Fig.~\ref{fig.ewsepileptor}c). This may be attributed to the Epileptor's unobservability at the saddle-node bifurcation point, as pinpointed in Section~\ref{sec.epileptor.obsv}. Fig.~\ref{fig.embeddingattractors}c shows that the embedded trajectories squeeze to singularity point at the unobservable (bifurcation) point\textemdash a feature that can be further explored for early-warning detection of seizure events. This topological characteristic of the embedded attractor is present not only in the Epileptor model but also in real data, as shown in Fig.~\ref{fig.embeddingattractors}d for the embedding space constructed from human intracranial EEG data (public data available at \cite{KAROLY2018}; sampling protocols and preliminary analysis are described in \cite{cook2013prediction}). Further computing the tSVD in human EEG data provides interesting results as illustrated for a representative patient in Fig.~\ref{fig.ewshuman}: the seizure onset is often preceded by a decrease of $\sigma_{d_e}$ followed by a sharp increase close to the critical transition, a characteristic that may be explored for real-time monitoring and detection of seizure events. The same pattern was observed for all patients in the considered dataset, although a thorough statistical investigation of the tSVD as a early-warning signal of seizure events (and other critical transitions in complex systems) is left for future work.

%====================================================================
\section{Discussion}
\label{sec.conc}

% Contributions and main relevance to physics/engineering community
The established relation between observability and embedding theories opens a new research direction of special interest to (nonlinear) time-series analysis. Our theory formally determines the conditions for reconstructing the system state from time-series data, often low-dimensional or univariate. In fact, measuring every relevant state variable is in practice constrained by physical limitations or operational costs. Hence, indirect estimation of unmeasured variables is required for the observation of physical, biological, ecological, and other complex dynamical systems.

For applications which require reconstructing only a few key state variables or lower-dimensional subspaces, we formalize the notion of functional observability for nonlinear systems. Our results can provide \textit{a priori} knowledge of the reconstruction limitations and embedding features, depending on the available time-series data. We show that, even if a system is not completely observable (reconstructible), it may still be functionally observable with respect to the variables or subspace of interest. This provides useful insights about the dynamical system's properties and can be used to guide experimental design and data-processing methods according to the investigated hypotheses and available measurement processes. \\

%As a key example, in the context of systems biology, technical limitations prevent to simultaneously measure different types of biophysical variables in closed systems like single cells (e.g., in multi-omics data sequencing \cite{seydel2021single}) or multiple ion channels \cite{}). Although this is often circumvented by using biological replicates in experiments, this procedure introduces errors ... \cite{}. 

% One application example is systems biology, where reconstructing complete state spaces is often hampered by lack of simultaneous measurements of multiple biophysical variables, performed upon single closed systems like cells or brains. Here, our functional observability analysis can guide experimental design and analysis. In fact, it allows to reveal whether, in underlying system's models, there are structural limitations that prevent accurate inference of certain variables from data collected on others variables.  

In the context of systems biology, observing system dynamics is often hampered by technical limitations that prevent the simultaneous measurement of multiple biophysical variables (e.g., multiple ion channels in single neurons). By identifying conditions for accurate inference of variables from time-series data, the presented functional observability analysis can thus guide experimental design.
Consider, for example, the HR neuron model investigated in Section~\ref{sec.numresult.neuron}. The fast recovery current and the slow adaptation current represented by the system's state variables are related to transport rates of fast (e.g., sodium and potassium) and slow (e.g., calcium) ion channels, respectively \cite{Gu2013}. Our analysis reveals structural limitations in the HR neuron model that prevent an accurate inference of calcium flux from measures of sodium/potassium fluxes. The opposite, instead, seems feasible, given that the system is completely observable everywhere when inferring sodium/potassium flux from measures of calcium flux. 

Likewise, in a biomedical context, evaluating the functional observability of the Epileptor model shows that reconstructing the slow permittivity variable from EEG data ($y = x_1 + x_4$) is complicated by the system's poor observability in the attractor's normal state  (Section~\ref{sec.Epileptor}). Contrariwise, independently measuring the state variables $x_1$ and $x_4$ (i.e., $y = [x_1 \,\,\, x_4]^\transp$) yields well-conditioned coefficients of observability throughout the entire attractor. This suggests that applying data pre-processing methods in EEG time series to uncouple the oscillatory behavior (modeled by $x_1$) from the spikes and wave components (modeled by $x_4$) may lead to better performance in reconstructing the permittivity variable for early-warning of seizures.\\

%Observability is a ubiquitous property for the design of state estimators, feedback controllers, and fault detectors in control processes and engineering. Contrarily to its broad application in linear systems, generalizations of observability (and its dual, controllability) to nonlinear systems still remain elusively linked to the design of nonlinear state estimators (and controllers), particularly due to the local validity of these properties in the state space. For physics, experimental biology, and natural sciences, however,

% Future research
In addition to applications, the proposed theory opens new theoretical research directions for many disciplines.
First, the Cord and HR neuron examples show interesting links between the functional observability of a system and its intrinsic timescales. In both cases, high reconstruction errors stem from estimating slow variables from measures of fast variables. Future works can formally explore this interesting relation, complementing the analysis for linear systems \cite{Berger2017}, by extending the notion of functional observability to (nonlinear) differential-algebraic systems of form
\begin{equation}
    \begin{cases}
        \dot{\bm x}_1 = \bm f_1(\bm x_1, \bm x_2), \\
        0 = \bm f_2(\bm x_1, \bm x_2),
    \end{cases}
\end{equation}

\noindent
where a strong timescale separation arises from a quasi-steady-state assumption ($\dot {\bm x}_2 \approx 0$).

Second, our analysis of the Epileptor model shows a potential relation between the system's observability and its bifurcation points. Whether the loss of observability close to critical transitions is a universal behavior or a particularity of the Epileptor remains to be investigated. In ecological networks, time-series data of variables that make the system completely observable often lead to earlier warning of critical transitions \cite{Aparicio2021}. Our time-series-based coefficient tSVD, aside from indirectly quantifying observability, may also capture features related to the Central Limit Theorem \cite{Haragus2010} (that, close to bifurcation points, dynamical systems can be locally reduced to low-dimensional normal forms). However, it is still to be investigated whether our framework only applies to transitions induced by local bifurcations, or it can be extended to other types like boundary crisis involving chaotic attractors \cite{tantet2018resonances,tantet2018crisis}. The potential use of tSVD for early-warning detection of critical transitions in complex systems, similar to other signals like increasing variance and autocorrelation \cite{kuehn2011mathematical,proverbio2022buffering}, may lead to promising theoretical developments and applications.

Third, our theory fosters data-driven methods for the automated construction of the embedding space, and its map $\bm \Phi$ to the original system's attractor, in applications where analytical analysis of the model is untractable (e.g., due to unknown parameters or high dimensionality). This would thus extend previous works on automated embedding construction \cite{Sadeghzadeh2022} and full system identification from embedding coordinates \cite{Brunton2016}. Although our application examples focus on univariate measurements and functionals, the theory is formalized for multivariate cases ($q, r \geq 1$) and can be directly applied to determine the existence and conditioning of such map, assessing how good the reconstruction is expected to be (locally).
% for applications where constructing the map $\bm\Phi$ between the embedding and functional spaces is analytically untractable (e.g., due to unknown parameters or high dimensionality), our theory 
% High-dimensional measurements, however, require an appropriate selection of coordinates for the construction of the embedding space $\mathcal E$. In applications where constructing the map $\bm\Phi$ between the embedding and functional spaces is analytically untractable (e.g., due to unknown parameters or high dimensionality), the presented theory can still be used to determine the existence and conditioning of such map, assessing how good the reconstruction is expected to be (locally). 
% For these cases, data-driven construction of embedding spaces \cite{Sadeghzadeh2022} and its map $\bm \Psi$ to the original system's attractor \cite{Brunton2016} have been proposed in the literature for full-state reconstruction. 
%The generalized problem of automated construction of a map $\bm\Phi$ from the embedded time-series data to the relevant subspaces (functionals), however, still remains an open challenge\textemdash which may be further pursued based on the theory developed in this work. 

Finally, for the study of high-dimensional problems, our results call for extensions based on graph-theoretical conditions \cite{Liu2013c,Angulo2020,Montanari2021} or network motifs \cite{Whalen2015}. In fact, as the computation of Lie derivatives is particularly demanding for high-dimensional systems, scalable strategies have yet to be developed to investigate the functional observability of large-scale nonlinear networks.

%===================================================================
\appendix

\section{Proof of Theorem~\ref{theor.nonlinearfuncobsv}}
\label{app.appendproof}

\begin{proof}
Given sufficiently smooth functions $f(\bm x)$ and $\bm h(\bm x)$, we show that condition \eqref{eq.functobsvcondition2} holds if and only if condition \eqref{eq.functobsvcondition} holds.

\textit{Sufficiency.} If condition \eqref{eq.functobsvcondition2} holds, then there exists some matrices $L_i\in\R^{r\times q}$, $i=1,\ldots,s$, such that
\begin{equation}
    \gradient\mathcal L_{\bm f}^0\bm g(\bm x) = \sum_{i=0}^s L_i \gradient\mathcal L_{\bm f}^i\bm h(\bm x).
    \label{eq.prooftheorem2suff}
\end{equation}

\noindent
Thus, condition \eqref{eq.functobsvcondition} holds given that $\mathcal L_{\bm f}^0\bm g(\bm x) = \bm g(\bm x)$.

\textit{Necessity.} If condition \eqref{eq.functobsvcondition} is satisfied, then relation \eqref{eq.prooftheorem2suff} holds. Right-multiplying \eqref{eq.prooftheorem2suff} by $\bm f(\bm x)$ yields
\begin{equation}
    \mathcal L_{\bm f}^1 \bm g(\bm x) = \sum_{i=0}^s L_i\mathcal L_{\bm f}^{i+1} \bm h(\bm x).
\end{equation}

\noindent
By induction, successively taking the gradient on both sides yields:
\begin{equation}
    \begin{aligned}
        \gradient\mathcal L_{\bm f}^1 \bm g(\bm x) &= \sum_{i=1}^{s+1} L_{i1} \gradient \mathcal L_{\bm f}^{i}\bm h(\bm x), \\
        %\gradient\mathcal L_{\bm f}^2 \bm g(\bm x) &= \sum_{i=2}^{s+2} L_{i2} \gradient \mathcal L_{\bm f}^{i}\bm h(\bm x), \\
        &\,\,\,\vdots \\
        \gradient\mathcal L_{\bm f}^\mu \bm g(\bm x) &= \sum_{i=\mu}^{s+\mu} L_{i\mu} \gradient \mathcal L_{\bm f}^{i}\bm h(\bm x),
\label{eq.prooftheorem2nec}
\end{aligned}
\end{equation}

\noindent
for some matrices $L_{ij}\in\R^{r\times q}$, $i=1,\ldots,s$ and $j = 1,\ldots,\mu$. From Definition \ref{def.obs_space}, $\gradient\mathcal L_{\bm f}^{s+j}\bm h(\bm x)$ is a linear combination of $\{\gradient\mathcal L_{\bm f}^{0}\bm h(\bm x), \ldots, \gradient\mathcal L_{\bm f}^{s}\bm h(\bm x)\}$. Therefore, equations \eqref{eq.prooftheorem2nec} can be expressed as
\begin{equation}
    \gradient\mathcal L_{\bm f}^j \bm g(\bm x) = \sum_{i=0}^{s} L_{ij} \gradient \mathcal L_{\bm f}^{i}\bm h(\bm x), 
\end{equation}

\noindent
which implies that condition \eqref{eq.functobsvcondition2} is satisfied.
\end{proof}

% For completeness, we restate below the theorem presented in \cite[Theorem 97]{VidyasagarBook} for the decomposition of observable and unobservable parts of nonlinear systems, which was used in the derivation of our results.

% \begin{theor}
%     Consider the nonlinear system \eqref{eq.nonlinearsys} and let $\bm x(t_0)\in\mathcal X$ be given. Let $\mathcal O$ be as in Definition~\ref{def.obs_space}. For each $\bm x\in\mathcal X$, let $\gradient\mathcal O(\bm x)$ denote the subspace of $\R^n$ consisting of all row vectors $\gradient L_{\bm f}^\nu h_j(\bm x)$, $L_{\bm f}^\nu h_j(\bm x)\in\mathcal O$. Suppose there exists a neighborhood $\mathcal U\subseteq\mathcal X$ of $\bm x_0$ such that Eq.~\eqref{eq.obsvdimension} holds. Then, there exists a diffeomorphism $T$ on $\mathcal U$ such that, for the transformation $\tilde{\bm x} = T(\bm x)$ and 

% \label{thm.vidyasagar}
% \end{theor}

%-------------------------------------------------------------------
\section{Coefficients of observability from time-series data}
\label{app.timeseriesobsv}

The coefficients of observability \eqref{eq.coeffuncobsv} can be indirectly inferred from time-series data by exploring the topological features associated with unobservable regions in the embedded state space. Let $Y(t)\in\R^{N}$ be the recorded time-series data, for time instants $t\in [0, N]$, and $X = [Y(t) \,\, Y(t-\tau) \,\, \ldots \,\, Y(t-(d_e-1)\tau)] \in\R^{N\times d_e}$ be the corresponding time-delay embedding for some embedding dimension $d_e$ and delay $\tau$. Methods based on singular value decomposition of embedded time-series data were shown to indirectly quantify the system's local observability, by measuring the geometrical complexity around some neighborhood of the embedded attractor to identify singularities in the embedded trajectories \cite{Aguirre2011,Portes2016}.

Here, we indirectly measure the system's observability by monitoring the smallest singular value $\sigma_{d_e}$ corresponding to the singular value decomposition $X = U\Sigma V^\transp$. Note that the subindex $d_e$ corresponds to the embedding dimension. The coefficient $\sigma_{d_e}$ is addressed as tSVD throughout the paper. To compare the local coefficient of functional observability $\kappa(t)$ at some time instant $t$ to the tSVD $\sigma_{d_e}(t)$ (Fig.~\ref{fig.ewsepileptor}), $\sigma_{d_e}(t)$ must be locally computed using a short time-series window close to the time instant $t$. In this work, we consider that $\sigma_{d_e}(t)$ is computed using the embedding of a moving time-series window of length $N$: $\{X(t-N), \ldots , X(t)\}$. Since numerical results may show high variability, we can also use a second moving average window of length $N_{\rm avg}$ to smooth the computed tSVD.

% \begin{figure}
%     \centering
%     % \includegraphics{}
%     \caption{Coefficient of observability computed from (a) the system equations, as in \eqref{eq.coeffuncobsv}, and (b) the embedded time-series data $X$ for the Lorenz system ($h(\bm x) = x_1$).}
%     \label{fig.tSVDbsv}
% \end{figure}

%====================================================================
\begin{acknowledgments}
	The authors thank Dr. Cristina Donato for useful insights on seizure onset. L.F. acknowledges support from Brazil's  Fundação de Amparo à Pesquisa do Estado de Minas Gerais (APQ-00781-21) and Conselho Nacional de Desenvolvimento Científico e Tecnológico (409487/2021-0). D.P. is supported by the Luxembourg National Research Fund (PRIDE DTU CriTiCS 10907093).
\end{acknowledgments}

%====================================================================

% \bibliography{library}

\begin{thebibliography}{82}%
\makeatletter
\providecommand \@ifxundefined [1]{%
 \@ifx{#1\undefined}
}%
\providecommand \@ifnum [1]{%
 \ifnum #1\expandafter \@firstoftwo
 \else \expandafter \@secondoftwo
 \fi
}%
\providecommand \@ifx [1]{%
 \ifx #1\expandafter \@firstoftwo
 \else \expandafter \@secondoftwo
 \fi
}%
\providecommand \natexlab [1]{#1}%
\providecommand \enquote  [1]{``#1''}%
\providecommand \bibnamefont  [1]{#1}%
\providecommand \bibfnamefont [1]{#1}%
\providecommand \citenamefont [1]{#1}%
\providecommand \href@noop [0]{\@secondoftwo}%
\providecommand \href [0]{\begingroup \@sanitize@url \@href}%
\providecommand \@href[1]{\@@startlink{#1}\@@href}%
\providecommand \@@href[1]{\endgroup#1\@@endlink}%
\providecommand \@sanitize@url [0]{\catcode `\\12\catcode `\$12\catcode
  `\&12\catcode `\#12\catcode `\^12\catcode `\_12\catcode `\%12\relax}%
\providecommand \@@startlink[1]{}%
\providecommand \@@endlink[0]{}%
\providecommand \url  [0]{\begingroup\@sanitize@url \@url }%
\providecommand \@url [1]{\endgroup\@href {#1}{\urlprefix }}%
\providecommand \urlprefix  [0]{URL }%
\providecommand \Eprint [0]{\href }%
\providecommand \doibase [0]{https://doi.org/}%
\providecommand \selectlanguage [0]{\@gobble}%
\providecommand \bibinfo  [0]{\@secondoftwo}%
\providecommand \bibfield  [0]{\@secondoftwo}%
\providecommand \translation [1]{[#1]}%
\providecommand \BibitemOpen [0]{}%
\providecommand \bibitemStop [0]{}%
\providecommand \bibitemNoStop [0]{.\EOS\space}%
\providecommand \EOS [0]{\spacefactor3000\relax}%
\providecommand \BibitemShut  [1]{\csname bibitem#1\endcsname}%
\let\auto@bib@innerbib\@empty
%</preamble>
\bibitem [{\citenamefont {Sauer}\ \emph {et~al.}(1991)\citenamefont {Sauer},
  \citenamefont {Yorke},\ and\ \citenamefont {Casdagli}}]{sauer1991embedology}%
  \BibitemOpen
  \bibfield  {author} {\bibinfo {author} {\bibfnamefont {T.}~\bibnamefont
  {Sauer}}, \bibinfo {author} {\bibfnamefont {J.~A.}\ \bibnamefont {Yorke}},\
  and\ \bibinfo {author} {\bibfnamefont {M.}~\bibnamefont {Casdagli}},\
  }\bibfield  {title} {\bibinfo {title} {Embedology},\ }\href@noop {}
  {\bibfield  {journal} {\bibinfo  {journal} {Journal of Statistical Physics}\
  }\textbf {\bibinfo {volume} {65}},\ \bibinfo {pages} {579} (\bibinfo {year}
  {1991})}\BibitemShut {NoStop}%
\bibitem [{\citenamefont {Marwan}\ \emph {et~al.}(2009)\citenamefont {Marwan},
  \citenamefont {Donges}, \citenamefont {Zou}, \citenamefont {Donner},\ and\
  \citenamefont {Kurths}}]{Marwan2009}%
  \BibitemOpen
  \bibfield  {author} {\bibinfo {author} {\bibfnamefont {N.}~\bibnamefont
  {Marwan}}, \bibinfo {author} {\bibfnamefont {J.~F.}\ \bibnamefont {Donges}},
  \bibinfo {author} {\bibfnamefont {Y.}~\bibnamefont {Zou}}, \bibinfo {author}
  {\bibfnamefont {R.~V.}\ \bibnamefont {Donner}},\ and\ \bibinfo {author}
  {\bibfnamefont {J.}~\bibnamefont {Kurths}},\ }\bibfield  {title} {\bibinfo
  {title} {{Complex network approach for recurrence analysis of time series}},\
  }\href@noop {} {\bibfield  {journal} {\bibinfo  {journal} {Physics Letters
  A}\ }\textbf {\bibinfo {volume} {373}},\ \bibinfo {pages} {4246} (\bibinfo
  {year} {2009})}\BibitemShut {NoStop}%
\bibitem [{\citenamefont {Lekscha}\ and\ \citenamefont
  {Donner}(2018)}]{Lekscha2018}%
  \BibitemOpen
  \bibfield  {author} {\bibinfo {author} {\bibfnamefont {J.}~\bibnamefont
  {Lekscha}}\ and\ \bibinfo {author} {\bibfnamefont {R.~V.}\ \bibnamefont
  {Donner}},\ }\bibfield  {title} {\bibinfo {title} {{Phase space
  reconstruction for non-uniformly sampled noisy time series}},\ }\href@noop {}
  {\bibfield  {journal} {\bibinfo  {journal} {Chaos}\ }\textbf {\bibinfo
  {volume} {28}},\ \bibinfo {pages} {085702} (\bibinfo {year}
  {2018})}\BibitemShut {NoStop}%
\bibitem [{\citenamefont {Bhat}\ and\ \citenamefont {Munch}(2022)}]{Bhat2022}%
  \BibitemOpen
  \bibfield  {author} {\bibinfo {author} {\bibfnamefont {U.}~\bibnamefont
  {Bhat}}\ and\ \bibinfo {author} {\bibfnamefont {S.~B.}\ \bibnamefont
  {Munch}},\ }\bibfield  {title} {\bibinfo {title} {{Recurrent neural networks
  for partially observed dynamical systems}},\ }\href@noop {} {\bibfield
  {journal} {\bibinfo  {journal} {Physical Review E}\ }\textbf {\bibinfo
  {volume} {105}},\ \bibinfo {pages} {044205} (\bibinfo {year}
  {2022})}\BibitemShut {NoStop}%
\bibitem [{\citenamefont {Kennel}\ and\ \citenamefont
  {Isabelle}(1992)}]{kennel1992method}%
  \BibitemOpen
  \bibfield  {author} {\bibinfo {author} {\bibfnamefont {M.~B.}\ \bibnamefont
  {Kennel}}\ and\ \bibinfo {author} {\bibfnamefont {S.}~\bibnamefont
  {Isabelle}},\ }\bibfield  {title} {\bibinfo {title} {Method to distinguish
  possible chaos from colored noise and to determine embedding parameters},\
  }\href@noop {} {\bibfield  {journal} {\bibinfo  {journal} {Physical Review
  A}\ }\textbf {\bibinfo {volume} {46}},\ \bibinfo {pages} {3111} (\bibinfo
  {year} {1992})}\BibitemShut {NoStop}%
\bibitem [{\citenamefont {Marwan}\ \emph {et~al.}(2007)\citenamefont {Marwan},
  \citenamefont {{Carmen Romano}}, \citenamefont {Thiel},\ and\ \citenamefont
  {Kurths}}]{Marwan2007}%
  \BibitemOpen
  \bibfield  {author} {\bibinfo {author} {\bibfnamefont {N.}~\bibnamefont
  {Marwan}}, \bibinfo {author} {\bibfnamefont {M.}~\bibnamefont {{Carmen
  Romano}}}, \bibinfo {author} {\bibfnamefont {M.}~\bibnamefont {Thiel}},\ and\
  \bibinfo {author} {\bibfnamefont {J.}~\bibnamefont {Kurths}},\ }\bibfield
  {title} {\bibinfo {title} {{Recurrence plots for the analysis of complex
  systems}},\ }\href@noop {} {\bibfield  {journal} {\bibinfo  {journal}
  {Physics Reports}\ }\textbf {\bibinfo {volume} {438}},\ \bibinfo {pages}
  {237} (\bibinfo {year} {2007})}\BibitemShut {NoStop}%
\bibitem [{\citenamefont {Sugihara}\ \emph {et~al.}(2012)\citenamefont
  {Sugihara}, \citenamefont {May}, \citenamefont {Ye}, \citenamefont {Hsieh},
  \citenamefont {Deyle}, \citenamefont {Fogarty},\ and\ \citenamefont
  {Munch}}]{Sugihara2012}%
  \BibitemOpen
  \bibfield  {author} {\bibinfo {author} {\bibfnamefont {G.}~\bibnamefont
  {Sugihara}}, \bibinfo {author} {\bibfnamefont {R.}~\bibnamefont {May}},
  \bibinfo {author} {\bibfnamefont {H.}~\bibnamefont {Ye}}, \bibinfo {author}
  {\bibfnamefont {C.~H.}\ \bibnamefont {Hsieh}}, \bibinfo {author}
  {\bibfnamefont {E.}~\bibnamefont {Deyle}}, \bibinfo {author} {\bibfnamefont
  {M.}~\bibnamefont {Fogarty}},\ and\ \bibinfo {author} {\bibfnamefont
  {S.}~\bibnamefont {Munch}},\ }\bibfield  {title} {\bibinfo {title}
  {{Detecting causality in complex ecosystems}},\ }\href@noop {} {\bibfield
  {journal} {\bibinfo  {journal} {Science}\ }\textbf {\bibinfo {volume}
  {338}},\ \bibinfo {pages} {496} (\bibinfo {year} {2012})}\BibitemShut
  {NoStop}%
\bibitem [{\citenamefont {Groth}\ and\ \citenamefont {Ghil}(2017)}]{Groth2017}%
  \BibitemOpen
  \bibfield  {author} {\bibinfo {author} {\bibfnamefont {A.}~\bibnamefont
  {Groth}}\ and\ \bibinfo {author} {\bibfnamefont {M.}~\bibnamefont {Ghil}},\
  }\bibfield  {title} {\bibinfo {title} {{Synchronization of world economic
  activity}},\ }\href@noop {} {\bibfield  {journal} {\bibinfo  {journal}
  {Chaos}\ }\textbf {\bibinfo {volume} {27}},\ \bibinfo {pages} {127002}
  (\bibinfo {year} {2017})}\BibitemShut {NoStop}%
\bibitem [{\citenamefont {Ara{\'{u}}jo}\ \emph {et~al.}(2019)\citenamefont
  {Ara{\'{u}}jo}, \citenamefont {Nedjah}, \citenamefont {Oliveira},\ and\
  \citenamefont {Meira}}]{Araujo2019}%
  \BibitemOpen
  \bibfield  {author} {\bibinfo {author} {\bibfnamefont {R.~d.~A.}\
  \bibnamefont {Ara{\'{u}}jo}}, \bibinfo {author} {\bibfnamefont
  {N.}~\bibnamefont {Nedjah}}, \bibinfo {author} {\bibfnamefont {A.~L.}\
  \bibnamefont {Oliveira}},\ and\ \bibinfo {author} {\bibfnamefont {S.~R.~L.}\
  \bibnamefont {Meira}},\ }\bibfield  {title} {\bibinfo {title} {{A deep
  increasing–decreasing-linear neural network for financial time series
  prediction}},\ }\href@noop {} {\bibfield  {journal} {\bibinfo  {journal}
  {Neurocomputing}\ }\textbf {\bibinfo {volume} {347}},\ \bibinfo {pages} {59}
  (\bibinfo {year} {2019})}\BibitemShut {NoStop}%
\bibitem [{\citenamefont {Gidea}\ \emph {et~al.}(2020)\citenamefont {Gidea},
  \citenamefont {Goldsmith}, \citenamefont {Katz}, \citenamefont {Roldan},\
  and\ \citenamefont {Shmalo}}]{Gidea2020}%
  \BibitemOpen
  \bibfield  {author} {\bibinfo {author} {\bibfnamefont {M.}~\bibnamefont
  {Gidea}}, \bibinfo {author} {\bibfnamefont {D.}~\bibnamefont {Goldsmith}},
  \bibinfo {author} {\bibfnamefont {Y.}~\bibnamefont {Katz}}, \bibinfo {author}
  {\bibfnamefont {P.}~\bibnamefont {Roldan}},\ and\ \bibinfo {author}
  {\bibfnamefont {Y.}~\bibnamefont {Shmalo}},\ }\bibfield  {title} {\bibinfo
  {title} {{Topological recognition of critical transitions in time series of
  cryptocurrencies}},\ }\href@noop {} {\bibfield  {journal} {\bibinfo
  {journal} {Physica A: Statistical Mechanics and its Applications}\ }\textbf
  {\bibinfo {volume} {548}},\ \bibinfo {pages} {123843} (\bibinfo {year}
  {2020})}\BibitemShut {NoStop}%
\bibitem [{\citenamefont {Carvalho}\ \emph {et~al.}(2018)\citenamefont
  {Carvalho}, \citenamefont {Portes}, \citenamefont {Beda}, \citenamefont
  {Tallarico},\ and\ \citenamefont {Aguirre}}]{Carvalho2018}%
  \BibitemOpen
  \bibfield  {author} {\bibinfo {author} {\bibfnamefont {N.~C.}\ \bibnamefont
  {Carvalho}}, \bibinfo {author} {\bibfnamefont {L.~L.}\ \bibnamefont
  {Portes}}, \bibinfo {author} {\bibfnamefont {A.}~\bibnamefont {Beda}},
  \bibinfo {author} {\bibfnamefont {L.~M.}\ \bibnamefont {Tallarico}},\ and\
  \bibinfo {author} {\bibfnamefont {L.~A.}\ \bibnamefont {Aguirre}},\
  }\bibfield  {title} {\bibinfo {title} {{Recurrence plots for the assessment
  of patient-ventilator interactions quality during invasive mechanical
  ventilation}},\ }\href@noop {} {\bibfield  {journal} {\bibinfo  {journal}
  {Chaos}\ }\textbf {\bibinfo {volume} {28}},\ \bibinfo {pages} {085707}
  (\bibinfo {year} {2018})}\BibitemShut {NoStop}%
\bibitem [{\citenamefont {P{\'{e}}rez-Toro}\ \emph {et~al.}(2020)\citenamefont
  {P{\'{e}}rez-Toro}, \citenamefont {V{\'{a}}squez-Correa}, \citenamefont
  {Arias-Vergara}, \citenamefont {N{\"{o}}th},\ and\ \citenamefont
  {Orozco-Arroyave}}]{PerezToro2020}%
  \BibitemOpen
  \bibfield  {author} {\bibinfo {author} {\bibfnamefont {P.~A.}\ \bibnamefont
  {P{\'{e}}rez-Toro}}, \bibinfo {author} {\bibfnamefont {J.~C.}\ \bibnamefont
  {V{\'{a}}squez-Correa}}, \bibinfo {author} {\bibfnamefont {T.}~\bibnamefont
  {Arias-Vergara}}, \bibinfo {author} {\bibfnamefont {E.}~\bibnamefont
  {N{\"{o}}th}},\ and\ \bibinfo {author} {\bibfnamefont {J.~R.}\ \bibnamefont
  {Orozco-Arroyave}},\ }\bibfield  {title} {\bibinfo {title} {{Nonlinear
  dynamics and Poincar{\'{e}} sections to model gait impairments in different
  stages of Parkinson's disease}},\ }\href@noop {} {\bibfield  {journal}
  {\bibinfo  {journal} {Nonlinear Dynamics}\ }\textbf {\bibinfo {volume}
  {100}},\ \bibinfo {pages} {3253} (\bibinfo {year} {2020})}\BibitemShut
  {NoStop}%
\bibitem [{\citenamefont {Lekscha}\ and\ \citenamefont
  {Donner}(2020)}]{Lekscha2020}%
  \BibitemOpen
  \bibfield  {author} {\bibinfo {author} {\bibfnamefont {J.}~\bibnamefont
  {Lekscha}}\ and\ \bibinfo {author} {\bibfnamefont {R.~V.}\ \bibnamefont
  {Donner}},\ }\bibfield  {title} {\bibinfo {title} {{Detecting dynamical
  anomalies in time series from different palaeoclimate proxy archives using
  windowed recurrence network analysis}},\ }\href@noop {} {\bibfield  {journal}
  {\bibinfo  {journal} {Nonlinear Processes in Geophysics}\ }\textbf {\bibinfo
  {volume} {27}},\ \bibinfo {pages} {261} (\bibinfo {year} {2020})}\BibitemShut
  {NoStop}%
\bibitem [{\citenamefont {Gao}\ \emph {et~al.}(2013)\citenamefont {Gao},
  \citenamefont {Zhang}, \citenamefont {Jin}, \citenamefont {Donner},
  \citenamefont {Marwan},\ and\ \citenamefont {Kurths}}]{Gao2013}%
  \BibitemOpen
  \bibfield  {author} {\bibinfo {author} {\bibfnamefont {Z.~K.}\ \bibnamefont
  {Gao}}, \bibinfo {author} {\bibfnamefont {X.~W.}\ \bibnamefont {Zhang}},
  \bibinfo {author} {\bibfnamefont {N.~D.}\ \bibnamefont {Jin}}, \bibinfo
  {author} {\bibfnamefont {R.~V.}\ \bibnamefont {Donner}}, \bibinfo {author}
  {\bibfnamefont {N.}~\bibnamefont {Marwan}},\ and\ \bibinfo {author}
  {\bibfnamefont {J.}~\bibnamefont {Kurths}},\ }\bibfield  {title} {\bibinfo
  {title} {{Recurrence networks from multivariate signals for uncovering
  dynamic transitions of horizontal oil-water stratified flows}},\ }\href@noop
  {} {\bibfield  {journal} {\bibinfo  {journal} {EPL}\ }\textbf {\bibinfo
  {volume} {103}},\ \bibinfo {pages} {50004} (\bibinfo {year}
  {2013})}\BibitemShut {NoStop}%
\bibitem [{\citenamefont {Letellier}\ \emph {et~al.}(2005)\citenamefont
  {Letellier}, \citenamefont {Aguirre},\ and\ \citenamefont
  {Maquet}}]{Letellier2005}%
  \BibitemOpen
  \bibfield  {author} {\bibinfo {author} {\bibfnamefont {C.}~\bibnamefont
  {Letellier}}, \bibinfo {author} {\bibfnamefont {L.~A.}\ \bibnamefont
  {Aguirre}},\ and\ \bibinfo {author} {\bibfnamefont {J.}~\bibnamefont
  {Maquet}},\ }\bibfield  {title} {\bibinfo {title} {{Relation between
  observability and differential embeddings for nonlinear dynamics}},\
  }\href@noop {} {\bibfield  {journal} {\bibinfo  {journal} {Physical Review
  E}\ }\textbf {\bibinfo {volume} {71}},\ \bibinfo {pages} {066213} (\bibinfo
  {year} {2005})}\BibitemShut {NoStop}%
\bibitem [{\citenamefont {Aguirre}\ and\ \citenamefont
  {Letellier}(2005)}]{Aguirre2005}%
  \BibitemOpen
  \bibfield  {author} {\bibinfo {author} {\bibfnamefont {L.~A.}\ \bibnamefont
  {Aguirre}}\ and\ \bibinfo {author} {\bibfnamefont {C.}~\bibnamefont
  {Letellier}},\ }\bibfield  {title} {\bibinfo {title} {{Observability of
  multivariate differential embeddings}},\ }\href@noop {} {\bibfield  {journal}
  {\bibinfo  {journal} {Journal of Physics A: Mathematical and General}\
  }\textbf {\bibinfo {volume} {38}},\ \bibinfo {pages} {6311} (\bibinfo {year}
  {2005})}\BibitemShut {NoStop}%
\bibitem [{\citenamefont {Kalman}(1959)}]{Kalman1959}%
  \BibitemOpen
  \bibfield  {author} {\bibinfo {author} {\bibfnamefont {R.}~\bibnamefont
  {Kalman}},\ }\bibfield  {title} {\bibinfo {title} {{On the general theory of
  control systems}},\ }\href@noop {} {\bibfield  {journal} {\bibinfo  {journal}
  {IRE Transactions on Automatic Control}\ }\textbf {\bibinfo {volume} {4}},\
  \bibinfo {pages} {110} (\bibinfo {year} {1959})}\BibitemShut {NoStop}%
\bibitem [{\citenamefont {Luenberger}(1966)}]{Luenberger1966}%
  \BibitemOpen
  \bibfield  {author} {\bibinfo {author} {\bibfnamefont {G.}~\bibnamefont
  {Luenberger}},\ }\bibfield  {title} {\bibinfo {title} {{Observers for
  multivariable systems}},\ }\href@noop {} {\bibfield  {journal} {\bibinfo
  {journal} {IEEE Transactions on Automatic Control}\ }\textbf {\bibinfo
  {volume} {AC-II}},\ \bibinfo {pages} {190} (\bibinfo {year}
  {1966})}\BibitemShut {NoStop}%
\bibitem [{\citenamefont {Kalman}(1960)}]{Kalman1960}%
  \BibitemOpen
  \bibfield  {author} {\bibinfo {author} {\bibfnamefont {R.~E.}\ \bibnamefont
  {Kalman}},\ }\bibfield  {title} {\bibinfo {title} {A new approach to linear
  filtering and prediction problems},\ }\href@noop {} {\bibfield  {journal}
  {\bibinfo  {journal} {Transactions of the ASME -- Journal of Basic
  Engineering}\ }\textbf {\bibinfo {volume} {82}},\ \bibinfo {pages} {35}
  (\bibinfo {year} {1960})}\BibitemShut {NoStop}%
\bibitem [{\citenamefont {Hermann}\ and\ \citenamefont
  {Krener}(1977)}]{Hermann1977}%
  \BibitemOpen
  \bibfield  {author} {\bibinfo {author} {\bibfnamefont {R.}~\bibnamefont
  {Hermann}}\ and\ \bibinfo {author} {\bibfnamefont {A.~J.}\ \bibnamefont
  {Krener}},\ }\bibfield  {title} {\bibinfo {title} {{Nonlinear Controllability
  and Observability}},\ }\href@noop {} {\bibfield  {journal} {\bibinfo
  {journal} {IEEE Transactions on Automatic Control}\ }\textbf {\bibinfo
  {volume} {22}},\ \bibinfo {pages} {728} (\bibinfo {year} {1977})}\BibitemShut
  {NoStop}%
\bibitem [{\citenamefont {Pecora}\ and\ \citenamefont
  {Carroll}(1990)}]{Pecora1990}%
  \BibitemOpen
  \bibfield  {author} {\bibinfo {author} {\bibfnamefont {L.~M.}\ \bibnamefont
  {Pecora}}\ and\ \bibinfo {author} {\bibfnamefont {T.~L.}\ \bibnamefont
  {Carroll}},\ }\bibfield  {title} {\bibinfo {title} {{Synchronization in
  chaotic systems}},\ }\href@noop {} {\bibfield  {journal} {\bibinfo  {journal}
  {Physical Review Letters}\ }\textbf {\bibinfo {volume} {64}},\ \bibinfo
  {pages} {821} (\bibinfo {year} {1990})}\BibitemShut {NoStop}%
\bibitem [{\citenamefont {Letellier}\ and\ \citenamefont
  {Aguirre}(2002)}]{Letellier2002}%
  \BibitemOpen
  \bibfield  {author} {\bibinfo {author} {\bibfnamefont {C.}~\bibnamefont
  {Letellier}}\ and\ \bibinfo {author} {\bibfnamefont {L.~A.}\ \bibnamefont
  {Aguirre}},\ }\bibfield  {title} {\bibinfo {title} {{Investigating nonlinear
  dynamics from time series: The influence of symmetries and the choice of
  observables}},\ }\href@noop {} {\bibfield  {journal} {\bibinfo  {journal}
  {Chaos}\ }\textbf {\bibinfo {volume} {12}},\ \bibinfo {pages} {549} (\bibinfo
  {year} {2002})}\BibitemShut {NoStop}%
\bibitem [{\citenamefont {Portes}\ and\ \citenamefont
  {Aguirre}(2016)}]{Portes2016}%
  \BibitemOpen
  \bibfield  {author} {\bibinfo {author} {\bibfnamefont {L.~L.}\ \bibnamefont
  {Portes}}\ and\ \bibinfo {author} {\bibfnamefont {L.~A.}\ \bibnamefont
  {Aguirre}},\ }\bibfield  {title} {\bibinfo {title} {{Enhancing multivariate
  singular spectrum analysis for phase synchronization: The role of
  observability}},\ }\href@noop {} {\bibfield  {journal} {\bibinfo  {journal}
  {Chaos}\ }\textbf {\bibinfo {volume} {26}},\ \bibinfo {pages} {093112}
  (\bibinfo {year} {2016})}\BibitemShut {NoStop}%
\bibitem [{\citenamefont {Carroll}(2018)}]{Carroll2018}%
  \BibitemOpen
  \bibfield  {author} {\bibinfo {author} {\bibfnamefont {T.~L.}\ \bibnamefont
  {Carroll}},\ }\bibfield  {title} {\bibinfo {title} {{Testing Dynamical System
  Variables for Reconstruction}},\ }\href@noop {} {\bibfield  {journal}
  {\bibinfo  {journal} {Chaos}\ }\textbf {\bibinfo {volume} {28}},\ \bibinfo
  {pages} {103117} (\bibinfo {year} {2018})}\BibitemShut {NoStop}%
\bibitem [{\citenamefont {Portes}\ \emph {et~al.}(2019)\citenamefont {Portes},
  \citenamefont {Montanari}, \citenamefont {Correa}, \citenamefont {Small},\
  and\ \citenamefont {Aguirre}}]{Portes2019}%
  \BibitemOpen
  \bibfield  {author} {\bibinfo {author} {\bibfnamefont {L.~L.}\ \bibnamefont
  {Portes}}, \bibinfo {author} {\bibfnamefont {A.~N.}\ \bibnamefont
  {Montanari}}, \bibinfo {author} {\bibfnamefont {D.~C.}\ \bibnamefont
  {Correa}}, \bibinfo {author} {\bibfnamefont {M.}~\bibnamefont {Small}},\ and\
  \bibinfo {author} {\bibfnamefont {L.~A.}\ \bibnamefont {Aguirre}},\
  }\bibfield  {title} {\bibinfo {title} {{The reliability of recurrence network
  analysis is influenced by the observability properties of the recorded time
  series}},\ }\href@noop {} {\bibfield  {journal} {\bibinfo  {journal} {Chaos}\
  }\textbf {\bibinfo {volume} {29}},\ \bibinfo {pages} {083101} (\bibinfo
  {year} {2019})}\BibitemShut {NoStop}%
\bibitem [{\citenamefont {Haber}\ \emph {et~al.}(2018)\citenamefont {Haber},
  \citenamefont {Molnar},\ and\ \citenamefont {Motter}}]{Haber2017}%
  \BibitemOpen
  \bibfield  {author} {\bibinfo {author} {\bibfnamefont {A.}~\bibnamefont
  {Haber}}, \bibinfo {author} {\bibfnamefont {F.}~\bibnamefont {Molnar}},\ and\
  \bibinfo {author} {\bibfnamefont {A.~E.}\ \bibnamefont {Motter}},\ }\bibfield
   {title} {\bibinfo {title} {{State Observation and Sensor Selection for
  Nonlinear Networks}},\ }\href@noop {} {\bibfield  {journal} {\bibinfo
  {journal} {IEEE Transactions on Control of Network Systems}\ }\textbf
  {\bibinfo {volume} {5}},\ \bibinfo {pages} {694 } (\bibinfo {year}
  {2018})}\BibitemShut {NoStop}%
\bibitem [{\citenamefont {Guan}\ \emph {et~al.}(2018)\citenamefont {Guan},
  \citenamefont {Berry},\ and\ \citenamefont {Sauer}}]{Guan2018}%
  \BibitemOpen
  \bibfield  {author} {\bibinfo {author} {\bibfnamefont {J.}~\bibnamefont
  {Guan}}, \bibinfo {author} {\bibfnamefont {T.}~\bibnamefont {Berry}},\ and\
  \bibinfo {author} {\bibfnamefont {T.}~\bibnamefont {Sauer}},\ }\bibfield
  {title} {\bibinfo {title} {{Limits on reconstruction of dynamical
  networks}},\ }\href@noop {} {\bibfield  {journal} {\bibinfo  {journal}
  {Physical Review E}\ }\textbf {\bibinfo {volume} {98}},\ \bibinfo {pages}
  {022318} (\bibinfo {year} {2018})}\BibitemShut {NoStop}%
\bibitem [{\citenamefont {Montanari}\ and\ \citenamefont
  {Aguirre}(2019)}]{Montanari2019}%
  \BibitemOpen
  \bibfield  {author} {\bibinfo {author} {\bibfnamefont {A.~N.}\ \bibnamefont
  {Montanari}}\ and\ \bibinfo {author} {\bibfnamefont {L.~A.}\ \bibnamefont
  {Aguirre}},\ }\bibfield  {title} {\bibinfo {title} {{Particle filtering of
  dynamical networks: Highlighting observability issues}},\ }\href@noop {}
  {\bibfield  {journal} {\bibinfo  {journal} {Chaos}\ }\textbf {\bibinfo
  {volume} {29}},\ \bibinfo {pages} {033118} (\bibinfo {year}
  {2019})}\BibitemShut {NoStop}%
\bibitem [{\citenamefont {Su}\ \emph {et~al.}(2017)\citenamefont {Su},
  \citenamefont {Wang}, \citenamefont {Li}, \citenamefont {Deng}, \citenamefont
  {Yu},\ and\ \citenamefont {Liu}}]{Su2017}%
  \BibitemOpen
  \bibfield  {author} {\bibinfo {author} {\bibfnamefont {F.}~\bibnamefont
  {Su}}, \bibinfo {author} {\bibfnamefont {J.}~\bibnamefont {Wang}}, \bibinfo
  {author} {\bibfnamefont {H.}~\bibnamefont {Li}}, \bibinfo {author}
  {\bibfnamefont {B.}~\bibnamefont {Deng}}, \bibinfo {author} {\bibfnamefont
  {H.}~\bibnamefont {Yu}},\ and\ \bibinfo {author} {\bibfnamefont
  {C.}~\bibnamefont {Liu}},\ }\bibfield  {title} {\bibinfo {title} {{Analysis
  and application of neuronal network controllability and observability}},\
  }\href@noop {} {\bibfield  {journal} {\bibinfo  {journal} {Chaos}\ }\textbf
  {\bibinfo {volume} {27}},\ \bibinfo {pages} {023103} (\bibinfo {year}
  {2017})}\BibitemShut {NoStop}%
\bibitem [{\citenamefont {Aguirre}\ \emph {et~al.}(2017)\citenamefont
  {Aguirre}, \citenamefont {Portes},\ and\ \citenamefont
  {Letellier}}]{Aguirre2017}%
  \BibitemOpen
  \bibfield  {author} {\bibinfo {author} {\bibfnamefont {L.~A.}\ \bibnamefont
  {Aguirre}}, \bibinfo {author} {\bibfnamefont {L.~L.}\ \bibnamefont
  {Portes}},\ and\ \bibinfo {author} {\bibfnamefont {C.}~\bibnamefont
  {Letellier}},\ }\bibfield  {title} {\bibinfo {title} {{Observability and
  synchronization of neuron models}},\ }\href@noop {} {\bibfield  {journal}
  {\bibinfo  {journal} {Chaos}\ }\textbf {\bibinfo {volume} {27}},\ \bibinfo
  {pages} {103103} (\bibinfo {year} {2017})}\BibitemShut {NoStop}%
\bibitem [{\citenamefont {Liu}\ \emph {et~al.}(2013)\citenamefont {Liu},
  \citenamefont {Slotine},\ and\ \citenamefont {Barab{\'{a}}si}}]{Liu2013c}%
  \BibitemOpen
  \bibfield  {author} {\bibinfo {author} {\bibfnamefont {Y.-Y.}\ \bibnamefont
  {Liu}}, \bibinfo {author} {\bibfnamefont {J.-J.}\ \bibnamefont {Slotine}},\
  and\ \bibinfo {author} {\bibfnamefont {A.-L.}\ \bibnamefont
  {Barab{\'{a}}si}},\ }\bibfield  {title} {\bibinfo {title} {{Observability of
  complex systems}},\ }\href@noop {} {\bibfield  {journal} {\bibinfo  {journal}
  {PNAS}\ }\textbf {\bibinfo {volume} {110}},\ \bibinfo {pages} {2460}
  (\bibinfo {year} {2013})}\BibitemShut {NoStop}%
\bibitem [{\citenamefont {Aparicio}\ \emph {et~al.}(2021)\citenamefont
  {Aparicio}, \citenamefont {Velasco-Hern{\'{a}}ndez}, \citenamefont {Moog},
  \citenamefont {Liu},\ and\ \citenamefont {Angulo}}]{Aparicio2021}%
  \BibitemOpen
  \bibfield  {author} {\bibinfo {author} {\bibfnamefont {A.}~\bibnamefont
  {Aparicio}}, \bibinfo {author} {\bibfnamefont {J.~X.}\ \bibnamefont
  {Velasco-Hern{\'{a}}ndez}}, \bibinfo {author} {\bibfnamefont {C.~H.}\
  \bibnamefont {Moog}}, \bibinfo {author} {\bibfnamefont {Y.~Y.}\ \bibnamefont
  {Liu}},\ and\ \bibinfo {author} {\bibfnamefont {M.~T.}\ \bibnamefont
  {Angulo}},\ }\bibfield  {title} {\bibinfo {title} {{Structure-based
  identification of sensor species for anticipating critical transitions}},\
  }\href@noop {} {\bibfield  {journal} {\bibinfo  {journal} {PNAS}\ }\textbf
  {\bibinfo {volume} {118}},\ \bibinfo {pages} {e2104732118} (\bibinfo {year}
  {2021})}\BibitemShut {NoStop}%
\bibitem [{\citenamefont {Pasqualetti}\ \emph {et~al.}(2013)\citenamefont
  {Pasqualetti}, \citenamefont {Zampieri},\ and\ \citenamefont
  {Bullo}}]{Pasqualetti2013}%
  \BibitemOpen
  \bibfield  {author} {\bibinfo {author} {\bibfnamefont {F.}~\bibnamefont
  {Pasqualetti}}, \bibinfo {author} {\bibfnamefont {S.}~\bibnamefont
  {Zampieri}},\ and\ \bibinfo {author} {\bibfnamefont {F.}~\bibnamefont
  {Bullo}},\ }\bibfield  {title} {\bibinfo {title} {{Controllability,
  Limitations and Algorithms for Complex Networks}},\ }\href@noop {} {\bibfield
   {journal} {\bibinfo  {journal} {IEEE Transactions on Control of Network
  Systems}\ }\textbf {\bibinfo {volume} {1}},\ \bibinfo {pages} {40} (\bibinfo
  {year} {2013})}\BibitemShut {NoStop}%
\bibitem [{\citenamefont {Sun}\ and\ \citenamefont {Motter}(2013)}]{Sun2013}%
  \BibitemOpen
  \bibfield  {author} {\bibinfo {author} {\bibfnamefont {J.}~\bibnamefont
  {Sun}}\ and\ \bibinfo {author} {\bibfnamefont {A.~E.}\ \bibnamefont
  {Motter}},\ }\bibfield  {title} {\bibinfo {title} {{Controllability
  transition and nonlocality in network control}},\ }\href@noop {} {\bibfield
  {journal} {\bibinfo  {journal} {Physical Review Letters}\ }\textbf {\bibinfo
  {volume} {110}},\ \bibinfo {pages} {208701} (\bibinfo {year}
  {2013})}\BibitemShut {NoStop}%
\bibitem [{\citenamefont {Whalen}\ \emph {et~al.}(2015)\citenamefont {Whalen},
  \citenamefont {Brennan}, \citenamefont {Sauer},\ and\ \citenamefont
  {Schiff}}]{Whalen2015}%
  \BibitemOpen
  \bibfield  {author} {\bibinfo {author} {\bibfnamefont {A.~J.}\ \bibnamefont
  {Whalen}}, \bibinfo {author} {\bibfnamefont {S.~N.}\ \bibnamefont {Brennan}},
  \bibinfo {author} {\bibfnamefont {T.~D.}\ \bibnamefont {Sauer}},\ and\
  \bibinfo {author} {\bibfnamefont {S.~J.}\ \bibnamefont {Schiff}},\ }\bibfield
   {title} {\bibinfo {title} {{Observability and controllability of nonlinear
  networks: The role of symmetry}},\ }\href@noop {} {\bibfield  {journal}
  {\bibinfo  {journal} {Physical Review X}\ }\textbf {\bibinfo {volume} {5}},\
  \bibinfo {pages} {011005} (\bibinfo {year} {2015})}\BibitemShut {NoStop}%
\bibitem [{\citenamefont {Letellier}\ \emph {et~al.}(2018)\citenamefont
  {Letellier}, \citenamefont {Sendi{\~{n}}a-Nadal},\ and\ \citenamefont
  {Aguirre}}]{Letellier2018}%
  \BibitemOpen
  \bibfield  {author} {\bibinfo {author} {\bibfnamefont {C.}~\bibnamefont
  {Letellier}}, \bibinfo {author} {\bibfnamefont {I.}~\bibnamefont
  {Sendi{\~{n}}a-Nadal}},\ and\ \bibinfo {author} {\bibfnamefont {L.~A.}\
  \bibnamefont {Aguirre}},\ }\bibfield  {title} {\bibinfo {title} {{A nonlinear
  graph-based theory for dynamical network observability}},\ }\href@noop {}
  {\bibfield  {journal} {\bibinfo  {journal} {Physical Review E}\ }\textbf
  {\bibinfo {volume} {98}},\ \bibinfo {pages} {020303} (\bibinfo {year}
  {2018})}\BibitemShut {NoStop}%
\bibitem [{\citenamefont {Angulo}\ \emph {et~al.}(2020)\citenamefont {Angulo},
  \citenamefont {Aparicio},\ and\ \citenamefont {Moog}}]{Angulo2020}%
  \BibitemOpen
  \bibfield  {author} {\bibinfo {author} {\bibfnamefont {M.~T.}\ \bibnamefont
  {Angulo}}, \bibinfo {author} {\bibfnamefont {A.}~\bibnamefont {Aparicio}},\
  and\ \bibinfo {author} {\bibfnamefont {C.~H.}\ \bibnamefont {Moog}},\
  }\bibfield  {title} {\bibinfo {title} {{Structural Accessibility and
  Structural Observability of Nonlinear Networked Systems}},\ }\href@noop {}
  {\bibfield  {journal} {\bibinfo  {journal} {IEEE Transactions on Network
  Science and Engineering}\ }\textbf {\bibinfo {volume} {7}},\ \bibinfo {pages}
  {1656} (\bibinfo {year} {2020})}\BibitemShut {NoStop}%
\bibitem [{\citenamefont {Montanari}\ and\ \citenamefont
  {Aguirre}(2020)}]{Montanari2020}%
  \BibitemOpen
  \bibfield  {author} {\bibinfo {author} {\bibfnamefont {A.~N.}\ \bibnamefont
  {Montanari}}\ and\ \bibinfo {author} {\bibfnamefont {L.~A.}\ \bibnamefont
  {Aguirre}},\ }\bibfield  {title} {\bibinfo {title} {{Observability of Network
  Systems: A Critical Review of Recent Results}},\ }\href@noop {} {\bibfield
  {journal} {\bibinfo  {journal} {Journal of Control, Automation and Electrical
  Systems}\ }\textbf {\bibinfo {volume} {31}},\ \bibinfo {pages} {1348}
  (\bibinfo {year} {2020})}\BibitemShut {NoStop}%
\bibitem [{\citenamefont {Rosenblum}\ \emph {et~al.}(1997)\citenamefont
  {Rosenblum}, \citenamefont {Pikovsky},\ and\ \citenamefont
  {Kurths}}]{Rosenblum1997}%
  \BibitemOpen
  \bibfield  {author} {\bibinfo {author} {\bibfnamefont {M.~G.}\ \bibnamefont
  {Rosenblum}}, \bibinfo {author} {\bibfnamefont {A.~S.}\ \bibnamefont
  {Pikovsky}},\ and\ \bibinfo {author} {\bibfnamefont {J.}~\bibnamefont
  {Kurths}},\ }\bibfield  {title} {\bibinfo {title} {{From Phase to Lag
  synchronization in Coupled Chaotic Oscillators}},\ }\href@noop {} {\bibfield
  {journal} {\bibinfo  {journal} {Physical Review Letters}\ }\textbf {\bibinfo
  {volume} {78}},\ \bibinfo {pages} {4193} (\bibinfo {year}
  {1997})}\BibitemShut {NoStop}%
\bibitem [{\citenamefont {Freitas}\ \emph {et~al.}(2018)\citenamefont
  {Freitas}, \citenamefont {Torres},\ and\ \citenamefont
  {Aguirre}}]{Freitas2018}%
  \BibitemOpen
  \bibfield  {author} {\bibinfo {author} {\bibfnamefont {L.}~\bibnamefont
  {Freitas}}, \bibinfo {author} {\bibfnamefont {L.~A.}\ \bibnamefont
  {Torres}},\ and\ \bibinfo {author} {\bibfnamefont {L.~A.}\ \bibnamefont
  {Aguirre}},\ }\bibfield  {title} {\bibinfo {title} {{Phase definition to
  assess synchronization quality of nonlinear oscillators}},\ }\href@noop {}
  {\bibfield  {journal} {\bibinfo  {journal} {Physical Review E}\ }\textbf
  {\bibinfo {volume} {97}},\ \bibinfo {pages} {052202} (\bibinfo {year}
  {2018})}\BibitemShut {NoStop}%
\bibitem [{\citenamefont {Oh}\ \emph {et~al.}(2014)\citenamefont {Oh},
  \citenamefont {Reischmann},\ and\ \citenamefont {Rial}}]{Oh2014}%
  \BibitemOpen
  \bibfield  {author} {\bibinfo {author} {\bibfnamefont {J.}~\bibnamefont
  {Oh}}, \bibinfo {author} {\bibfnamefont {E.}~\bibnamefont {Reischmann}},\
  and\ \bibinfo {author} {\bibfnamefont {J.~A.}\ \bibnamefont {Rial}},\
  }\bibfield  {title} {\bibinfo {title} {{Polar synchronization and the
  synchronized climatic history of Greenland and Antarctica}},\ }\href@noop {}
  {\bibfield  {journal} {\bibinfo  {journal} {Quaternary Science Reviews}\
  }\textbf {\bibinfo {volume} {83}},\ \bibinfo {pages} {129} (\bibinfo {year}
  {2014})}\BibitemShut {NoStop}%
\bibitem [{\citenamefont {Smug}\ \emph {et~al.}(2018)\citenamefont {Smug},
  \citenamefont {Ashwin},\ and\ \citenamefont {Sornette}}]{Smug2018}%
  \BibitemOpen
  \bibfield  {author} {\bibinfo {author} {\bibfnamefont {D.}~\bibnamefont
  {Smug}}, \bibinfo {author} {\bibfnamefont {P.}~\bibnamefont {Ashwin}},\ and\
  \bibinfo {author} {\bibfnamefont {D.}~\bibnamefont {Sornette}},\ }\bibfield
  {title} {\bibinfo {title} {{Predicting financial market crashes using ghost
  singularities}},\ }\href@noop {} {\bibfield  {journal} {\bibinfo  {journal}
  {PLoS ONE}\ }\textbf {\bibinfo {volume} {13}},\ \bibinfo {pages} {e0195265}
  (\bibinfo {year} {2018})}\BibitemShut {NoStop}%
\bibitem [{\citenamefont {Youn}\ \emph {et~al.}(2020)\citenamefont {Youn},
  \citenamefont {Rhudy}, \citenamefont {Cho},\ and\ \citenamefont
  {Myung}}]{Youn2020}%
  \BibitemOpen
  \bibfield  {author} {\bibinfo {author} {\bibfnamefont {W.}~\bibnamefont
  {Youn}}, \bibinfo {author} {\bibfnamefont {M.~B.}\ \bibnamefont {Rhudy}},
  \bibinfo {author} {\bibfnamefont {A.}~\bibnamefont {Cho}},\ and\ \bibinfo
  {author} {\bibfnamefont {H.}~\bibnamefont {Myung}},\ }\bibfield  {title}
  {\bibinfo {title} {{Fuzzy Adaptive Attitude Estimation for a Fixed-Wing UAV
  With a Virtual SSA Sensor During a GPS Outage}},\ }\href@noop {} {\bibfield
  {journal} {\bibinfo  {journal} {IEEE Sensors Journal}\ }\textbf {\bibinfo
  {volume} {20}},\ \bibinfo {pages} {1456} (\bibinfo {year}
  {2020})}\BibitemShut {NoStop}%
\bibitem [{\citenamefont {Quail}\ \emph {et~al.}(2015)\citenamefont {Quail},
  \citenamefont {Shrier},\ and\ \citenamefont {Glass}}]{quail2015predicting}%
  \BibitemOpen
  \bibfield  {author} {\bibinfo {author} {\bibfnamefont {T.}~\bibnamefont
  {Quail}}, \bibinfo {author} {\bibfnamefont {A.}~\bibnamefont {Shrier}},\ and\
  \bibinfo {author} {\bibfnamefont {L.}~\bibnamefont {Glass}},\ }\bibfield
  {title} {\bibinfo {title} {Predicting the onset of period-doubling
  bifurcations in noisy cardiac systems},\ }\href@noop {} {\bibfield  {journal}
  {\bibinfo  {journal} {PNAS}\ }\textbf {\bibinfo {volume} {112}},\ \bibinfo
  {pages} {9358} (\bibinfo {year} {2015})}\BibitemShut {NoStop}%
\bibitem [{\citenamefont {Jirsa}\ \emph {et~al.}(2014)\citenamefont {Jirsa},
  \citenamefont {Stacey}, \citenamefont {Quilichini}, \citenamefont {Ivanov},\
  and\ \citenamefont {Bernard}}]{Jirsa2014}%
  \BibitemOpen
  \bibfield  {author} {\bibinfo {author} {\bibfnamefont {V.~K.}\ \bibnamefont
  {Jirsa}}, \bibinfo {author} {\bibfnamefont {W.~C.}\ \bibnamefont {Stacey}},
  \bibinfo {author} {\bibfnamefont {P.~P.}\ \bibnamefont {Quilichini}},
  \bibinfo {author} {\bibfnamefont {A.~I.}\ \bibnamefont {Ivanov}},\ and\
  \bibinfo {author} {\bibfnamefont {C.}~\bibnamefont {Bernard}},\ }\bibfield
  {title} {\bibinfo {title} {{On the nature of seizure dynamics}},\ }\href@noop
  {} {\bibfield  {journal} {\bibinfo  {journal} {Brain}\ }\textbf {\bibinfo
  {volume} {137}},\ \bibinfo {pages} {2210} (\bibinfo {year}
  {2014})}\BibitemShut {NoStop}%
\bibitem [{\citenamefont {Fernando}\ \emph {et~al.}(2010)\citenamefont
  {Fernando}, \citenamefont {Trinh},\ and\ \citenamefont
  {Jennings}}]{Fernando2010}%
  \BibitemOpen
  \bibfield  {author} {\bibinfo {author} {\bibfnamefont {T.~L.}\ \bibnamefont
  {Fernando}}, \bibinfo {author} {\bibfnamefont {H.~M.}\ \bibnamefont
  {Trinh}},\ and\ \bibinfo {author} {\bibfnamefont {L.}~\bibnamefont
  {Jennings}},\ }\bibfield  {title} {\bibinfo {title} {{Functional
  Observability and the Design of Minimum Order Linear Functional Observers}},\
  }\href@noop {} {\bibfield  {journal} {\bibinfo  {journal} {IEEE Transactions
  on Automatic Control}\ }\textbf {\bibinfo {volume} {55}},\ \bibinfo {pages}
  {1268} (\bibinfo {year} {2010})}\BibitemShut {NoStop}%
\bibitem [{\citenamefont {Jennings}\ \emph {et~al.}(2011)\citenamefont
  {Jennings}, \citenamefont {Fernando},\ and\ \citenamefont
  {Trinh}}]{Jennings2011}%
  \BibitemOpen
  \bibfield  {author} {\bibinfo {author} {\bibfnamefont {L.~S.}\ \bibnamefont
  {Jennings}}, \bibinfo {author} {\bibfnamefont {T.~L.}\ \bibnamefont
  {Fernando}},\ and\ \bibinfo {author} {\bibfnamefont {H.~M.}\ \bibnamefont
  {Trinh}},\ }\bibfield  {title} {\bibinfo {title} {{Existence conditions for
  functional observability from an eigenspace perspective}},\ }\href@noop {}
  {\bibfield  {journal} {\bibinfo  {journal} {IEEE Transactions on Automatic
  Control}\ }\textbf {\bibinfo {volume} {56}},\ \bibinfo {pages} {2957}
  (\bibinfo {year} {2011})}\BibitemShut {NoStop}%
\bibitem [{\citenamefont {Darouach}(2000)}]{Darouach2000}%
  \BibitemOpen
  \bibfield  {author} {\bibinfo {author} {\bibfnamefont {M.}~\bibnamefont
  {Darouach}},\ }\bibfield  {title} {\bibinfo {title} {{Existence and Design of
  Functional Observers for Linear Systems}},\ }\href@noop {} {\bibfield
  {journal} {\bibinfo  {journal} {IEEE Transactions on Automatic Control}\
  }\textbf {\bibinfo {volume} {45}},\ \bibinfo {pages} {940} (\bibinfo {year}
  {2000})}\BibitemShut {NoStop}%
\bibitem [{\citenamefont {Hieu}\ and\ \citenamefont
  {Tyrone}(2012)}]{Trinh2012}%
  \BibitemOpen
  \bibfield  {author} {\bibinfo {author} {\bibfnamefont {T.}~\bibnamefont
  {Hieu}}\ and\ \bibinfo {author} {\bibfnamefont {F.}~\bibnamefont {Tyrone}},\
  }\href@noop {} {\emph {\bibinfo {title} {{Functional Observers for Dynamical
  Systems}}}}\ (\bibinfo  {publisher} {Springer Berlin Heidelberg},\ \bibinfo
  {year} {2012})\BibitemShut {NoStop}%
\bibitem [{\citenamefont {Alhelou}\ \emph {et~al.}(2019)\citenamefont
  {Alhelou}, \citenamefont {Golshan},\ and\ \citenamefont
  {Hatziargyriou}}]{Alhelou2019}%
  \BibitemOpen
  \bibfield  {author} {\bibinfo {author} {\bibfnamefont {H.~H.}\ \bibnamefont
  {Alhelou}}, \bibinfo {author} {\bibfnamefont {M.~E.~H.}\ \bibnamefont
  {Golshan}},\ and\ \bibinfo {author} {\bibfnamefont {N.~D.}\ \bibnamefont
  {Hatziargyriou}},\ }\bibfield  {title} {\bibinfo {title} {A decentralized
  functional observer based optimal lfc considering unknown inputs,
  uncertainties, and cyber-attacks},\ }\href@noop {} {\bibfield  {journal}
  {\bibinfo  {journal} {IEEE Transactions on Power Systems}\ }\textbf {\bibinfo
  {volume} {34}},\ \bibinfo {pages} {4408} (\bibinfo {year}
  {2019})}\BibitemShut {NoStop}%
\bibitem [{\citenamefont {Emami}\ \emph {et~al.}(2015)\citenamefont {Emami},
  \citenamefont {Fernando}, \citenamefont {Nener}, \citenamefont {Trinh},\ and\
  \citenamefont {Zhang}}]{Emami2015}%
  \BibitemOpen
  \bibfield  {author} {\bibinfo {author} {\bibfnamefont {K.}~\bibnamefont
  {Emami}}, \bibinfo {author} {\bibfnamefont {T.}~\bibnamefont {Fernando}},
  \bibinfo {author} {\bibfnamefont {B.}~\bibnamefont {Nener}}, \bibinfo
  {author} {\bibfnamefont {H.}~\bibnamefont {Trinh}},\ and\ \bibinfo {author}
  {\bibfnamefont {Y.}~\bibnamefont {Zhang}},\ }\bibfield  {title} {\bibinfo
  {title} {{A functional observer based fault detection technique for dynamical
  systems}},\ }\href@noop {} {\bibfield  {journal} {\bibinfo  {journal}
  {Journal of the Franklin Institute}\ }\textbf {\bibinfo {volume} {352}},\
  \bibinfo {pages} {2113} (\bibinfo {year} {2015})}\BibitemShut {NoStop}%
\bibitem [{\citenamefont {Montanari}\ \emph {et~al.}(2021)\citenamefont
  {Montanari}, \citenamefont {Duan}, \citenamefont {Aguirre},\ and\
  \citenamefont {Motter}}]{Montanari2021}%
  \BibitemOpen
  \bibfield  {author} {\bibinfo {author} {\bibfnamefont {A.~N.}\ \bibnamefont
  {Montanari}}, \bibinfo {author} {\bibfnamefont {C.}~\bibnamefont {Duan}},
  \bibinfo {author} {\bibfnamefont {L.~A.}\ \bibnamefont {Aguirre}},\ and\
  \bibinfo {author} {\bibfnamefont {A.~E.}\ \bibnamefont {Motter}},\ }\bibfield
   {title} {\bibinfo {title} {{Functional observability and target state
  estimation in large-scale networks}},\ }\href@noop {} {\bibfield  {journal}
  {\bibinfo  {journal} {In submission}\ } (\bibinfo {year} {2021})}\BibitemShut
  {NoStop}%
\bibitem [{\citenamefont {Vidyasagar}(1978)}]{VidyasagarBook}%
  \BibitemOpen
  \bibfield  {author} {\bibinfo {author} {\bibfnamefont {M.}~\bibnamefont
  {Vidyasagar}},\ }\href@noop {} {\emph {\bibinfo {title} {{Nonlinear Systems
  Analysis}}}},\ \bibinfo {edition} {2nd}\ ed.\ (\bibinfo  {publisher}
  {Prentice Hall},\ \bibinfo {year} {1978})\BibitemShut {NoStop}%
\bibitem [{\citenamefont {Anguelova}(2004)}]{AnguelovaThesis}%
  \BibitemOpen
  \bibfield  {author} {\bibinfo {author} {\bibfnamefont {M.}~\bibnamefont
  {Anguelova}},\ }\emph {\bibinfo {title} {{Nonlinear Observability and
  Identifiability: General Theory and a Case Study of a Kinetic Model for S.
  cerevisiae}}},\ \href@noop {} {\bibinfo {type} {{PhD Thesis}}},\ \bibinfo
  {school} {Department of Mathematics, Chalmers University of Technology and
  G\"oteborg University} (\bibinfo {year} {2004})\BibitemShut {NoStop}%
\bibitem [{\citenamefont {Chen}(1999)}]{Chi-TsongChen1999}%
  \BibitemOpen
  \bibfield  {author} {\bibinfo {author} {\bibfnamefont {C.-T.}\ \bibnamefont
  {Chen}},\ }\href@noop {} {\emph {\bibinfo {title} {{Linear System Theory and
  Design}}}},\ \bibinfo {edition} {3rd}\ ed.\ (\bibinfo  {publisher} {Oxford
  University Press},\ \bibinfo {year} {1999})\BibitemShut {NoStop}%
\bibitem [{\citenamefont {Rudin}\ \emph {et~al.}(1992)\citenamefont {Rudin},
  \citenamefont {Osher},\ and\ \citenamefont {Fatemi}}]{Rudin1992}%
  \BibitemOpen
  \bibfield  {author} {\bibinfo {author} {\bibfnamefont {L.~I.}\ \bibnamefont
  {Rudin}}, \bibinfo {author} {\bibfnamefont {S.}~\bibnamefont {Osher}},\ and\
  \bibinfo {author} {\bibfnamefont {E.}~\bibnamefont {Fatemi}},\ }\bibfield
  {title} {\bibinfo {title} {{Nonlinear total variation based noise removal
  algorithms}},\ }\href@noop {} {\bibfield  {journal} {\bibinfo  {journal}
  {Physica D}\ }\textbf {\bibinfo {volume} {60}},\ \bibinfo {pages} {259}
  (\bibinfo {year} {1992})}\BibitemShut {NoStop}%
\bibitem [{\citenamefont {Chartrand}(2011)}]{Chartrand2011}%
  \BibitemOpen
  \bibfield  {author} {\bibinfo {author} {\bibfnamefont {R.}~\bibnamefont
  {Chartrand}},\ }\bibfield  {title} {\bibinfo {title} {{Numerical
  Differentiation of Noisy, Nonsmooth Data}},\ }\href@noop {} {\bibfield
  {journal} {\bibinfo  {journal} {ISRN Applied Mathematics}\ }\textbf {\bibinfo
  {volume} {2011}},\ \bibinfo {pages} {164564} (\bibinfo {year}
  {2011})}\BibitemShut {NoStop}%
\bibitem [{\citenamefont {Takens}(1981)}]{Takens1981}%
  \BibitemOpen
  \bibfield  {author} {\bibinfo {author} {\bibfnamefont {F.}~\bibnamefont
  {Takens}},\ }\bibfield  {title} {\bibinfo {title} {{Detecting strange
  attractors in turbulence}},\ }in\ \href@noop {} {\emph {\bibinfo {booktitle}
  {Dynamical Systems and Turbulence}}},\ \bibinfo {editor} {edited by\ \bibinfo
  {editor} {\bibfnamefont {D.}~\bibnamefont {Rand}}\ and\ \bibinfo {editor}
  {\bibfnamefont {L.}~\bibnamefont {Young}}}\ (\bibinfo  {publisher} {Springer,
  Berlin, Heidelberg},\ \bibinfo {year} {1981})\ pp.\ \bibinfo {pages}
  {366--381}\BibitemShut {NoStop}%
\bibitem [{\citenamefont {Stark}\ \emph {et~al.}(1997)\citenamefont {Stark},
  \citenamefont {Broomhead}, \citenamefont {Davies},\ and\ \citenamefont
  {Huke}}]{Stark1997}%
  \BibitemOpen
  \bibfield  {author} {\bibinfo {author} {\bibfnamefont {J.}~\bibnamefont
  {Stark}}, \bibinfo {author} {\bibfnamefont {D.}~\bibnamefont {Broomhead}},
  \bibinfo {author} {\bibfnamefont {M.}~\bibnamefont {Davies}},\ and\ \bibinfo
  {author} {\bibfnamefont {J.}~\bibnamefont {Huke}},\ }\bibfield  {title}
  {\bibinfo {title} {Takens embedding theorems for forced and stochastic
  systems},\ }\href@noop {} {\bibfield  {journal} {\bibinfo  {journal}
  {Nonlinear Analysis: Theory, Methods \& Applications}\ }\textbf {\bibinfo
  {volume} {30}},\ \bibinfo {pages} {5303} (\bibinfo {year}
  {1997})}\BibitemShut {NoStop}%
\bibitem [{\citenamefont {Rotella}\ and\ \citenamefont
  {Zambettakis}(2016)}]{Rotella2016a}%
  \BibitemOpen
  \bibfield  {author} {\bibinfo {author} {\bibfnamefont {F.}~\bibnamefont
  {Rotella}}\ and\ \bibinfo {author} {\bibfnamefont {I.}~\bibnamefont
  {Zambettakis}},\ }\bibfield  {title} {\bibinfo {title} {{A note on functional
  observability}},\ }\href@noop {} {\bibfield  {journal} {\bibinfo  {journal}
  {IEEE Transactions on Automatic Control}\ }\textbf {\bibinfo {volume} {61}},\
  \bibinfo {pages} {3197} (\bibinfo {year} {2016})}\BibitemShut {NoStop}%
\bibitem [{\citenamefont {Letellier}\ and\ \citenamefont
  {Aguirre}(2012)}]{Letellier2012}%
  \BibitemOpen
  \bibfield  {author} {\bibinfo {author} {\bibfnamefont {C.}~\bibnamefont
  {Letellier}}\ and\ \bibinfo {author} {\bibfnamefont {L.~A.}\ \bibnamefont
  {Aguirre}},\ }\bibfield  {title} {\bibinfo {title} {{Required criteria for
  recognizing new types of chaos: Application to the "cord" attractor}},\
  }\href@noop {} {\bibfield  {journal} {\bibinfo  {journal} {Physical Review
  E}\ }\textbf {\bibinfo {volume} {85}},\ \bibinfo {pages} {036204} (\bibinfo
  {year} {2012})}\BibitemShut {NoStop}%
\bibitem [{\citenamefont {Freitas}\ \emph {et~al.}(2020)\citenamefont
  {Freitas}, \citenamefont {Portes}, \citenamefont {Torres},\ and\
  \citenamefont {Aguirre}}]{Freitas2020}%
  \BibitemOpen
  \bibfield  {author} {\bibinfo {author} {\bibfnamefont {L.}~\bibnamefont
  {Freitas}}, \bibinfo {author} {\bibfnamefont {L.~L.}\ \bibnamefont {Portes}},
  \bibinfo {author} {\bibfnamefont {L.~A.}\ \bibnamefont {Torres}},\ and\
  \bibinfo {author} {\bibfnamefont {L.~A.}\ \bibnamefont {Aguirre}},\
  }\bibfield  {title} {\bibinfo {title} {{Phase coherence is not related to
  topology}},\ }\href@noop {} {\bibfield  {journal} {\bibinfo  {journal}
  {Physical Review E}\ }\textbf {\bibinfo {volume} {101}},\ \bibinfo {pages}
  {032207} (\bibinfo {year} {2020})}\BibitemShut {NoStop}%
\bibitem [{\citenamefont {Hindmarsh}\ and\ \citenamefont
  {Rose}(1984)}]{Hindmarsh1984}%
  \BibitemOpen
  \bibfield  {author} {\bibinfo {author} {\bibfnamefont {J.~L.}\ \bibnamefont
  {Hindmarsh}}\ and\ \bibinfo {author} {\bibfnamefont {R.~M.}\ \bibnamefont
  {Rose}},\ }\bibfield  {title} {\bibinfo {title} {{A model of neuronal
  bursting using three coupled first order differential equations.}},\
  }\href@noop {} {\bibfield  {journal} {\bibinfo  {journal} {Proceedings of the
  Royal Society of London. Series B}\ }\textbf {\bibinfo {volume} {221}},\
  \bibinfo {pages} {87} (\bibinfo {year} {1984})}\BibitemShut {NoStop}%
\bibitem [{\citenamefont {Storace}\ \emph {et~al.}(2008)\citenamefont
  {Storace}, \citenamefont {Linaro},\ and\ \citenamefont {{De
  Lange}}}]{Storace2008}%
  \BibitemOpen
  \bibfield  {author} {\bibinfo {author} {\bibfnamefont {M.}~\bibnamefont
  {Storace}}, \bibinfo {author} {\bibfnamefont {D.}~\bibnamefont {Linaro}},\
  and\ \bibinfo {author} {\bibfnamefont {E.}~\bibnamefont {{De Lange}}},\
  }\bibfield  {title} {\bibinfo {title} {{The Hindmarsh-Rose neuron model:
  Bifurcation analysis and piecewise-linear approximations}},\ }\href@noop {}
  {\bibfield  {journal} {\bibinfo  {journal} {Chaos}\ }\textbf {\bibinfo
  {volume} {18}},\ \bibinfo {pages} {033128} (\bibinfo {year}
  {2008})}\BibitemShut {NoStop}%
\bibitem [{\citenamefont {Gu}(2013)}]{Gu2013}%
  \BibitemOpen
  \bibfield  {author} {\bibinfo {author} {\bibfnamefont {H.}~\bibnamefont
  {Gu}},\ }\bibfield  {title} {\bibinfo {title} {{Biological experimental
  observations of an unnoticed chaos as simulated by the Hindmarsh-Rose
  model}},\ }\href@noop {} {\bibfield  {journal} {\bibinfo  {journal} {PLoS
  ONE}\ }\textbf {\bibinfo {volume} {8}},\ \bibinfo {pages} {e81759} (\bibinfo
  {year} {2013})}\BibitemShut {NoStop}%
\bibitem [{\citenamefont {Mormann}\ \emph {et~al.}(2007)\citenamefont
  {Mormann}, \citenamefont {Andrzejak}, \citenamefont {Elger},\ and\
  \citenamefont {Lehnertz}}]{mormann2007seizure}%
  \BibitemOpen
  \bibfield  {author} {\bibinfo {author} {\bibfnamefont {F.}~\bibnamefont
  {Mormann}}, \bibinfo {author} {\bibfnamefont {R.~G.}\ \bibnamefont
  {Andrzejak}}, \bibinfo {author} {\bibfnamefont {C.~E.}\ \bibnamefont
  {Elger}},\ and\ \bibinfo {author} {\bibfnamefont {K.}~\bibnamefont
  {Lehnertz}},\ }\bibfield  {title} {\bibinfo {title} {Seizure prediction: the
  long and winding road},\ }\href@noop {} {\bibfield  {journal} {\bibinfo
  {journal} {Brain}\ }\textbf {\bibinfo {volume} {130}},\ \bibinfo {pages}
  {314} (\bibinfo {year} {2007})}\BibitemShut {NoStop}%
\bibitem [{\citenamefont {Cook}\ \emph {et~al.}(2013)\citenamefont {Cook},
  \citenamefont {O'Brien}, \citenamefont {Berkovic}, \citenamefont {Murphy},
  \citenamefont {Morokoff}, \citenamefont {Fabinyi}, \citenamefont {D'Souza},
  \citenamefont {Yerra}, \citenamefont {Archer}, \citenamefont {Litewka} \emph
  {et~al.}}]{cook2013prediction}%
  \BibitemOpen
  \bibfield  {author} {\bibinfo {author} {\bibfnamefont {M.~J.}\ \bibnamefont
  {Cook}}, \bibinfo {author} {\bibfnamefont {T.~J.}\ \bibnamefont {O'Brien}},
  \bibinfo {author} {\bibfnamefont {S.~F.}\ \bibnamefont {Berkovic}}, \bibinfo
  {author} {\bibfnamefont {M.}~\bibnamefont {Murphy}}, \bibinfo {author}
  {\bibfnamefont {A.}~\bibnamefont {Morokoff}}, \bibinfo {author}
  {\bibfnamefont {G.}~\bibnamefont {Fabinyi}}, \bibinfo {author} {\bibfnamefont
  {W.}~\bibnamefont {D'Souza}}, \bibinfo {author} {\bibfnamefont
  {R.}~\bibnamefont {Yerra}}, \bibinfo {author} {\bibfnamefont
  {J.}~\bibnamefont {Archer}}, \bibinfo {author} {\bibfnamefont
  {L.}~\bibnamefont {Litewka}}, \emph {et~al.},\ }\bibfield  {title} {\bibinfo
  {title} {Prediction of seizure likelihood with a long-term, implanted seizure
  advisory system in patients with drug-resistant epilepsy: a first-in-man
  study},\ }\href@noop {} {\bibfield  {journal} {\bibinfo  {journal} {The
  Lancet Neurology}\ }\textbf {\bibinfo {volume} {12}},\ \bibinfo {pages} {563}
  (\bibinfo {year} {2013})}\BibitemShut {NoStop}%
\bibitem [{\citenamefont {Kuhlmann}\ \emph {et~al.}(2018)\citenamefont
  {Kuhlmann}, \citenamefont {Karoly}, \citenamefont {Freestone}, \citenamefont
  {Brinkmann}, \citenamefont {Temko}, \citenamefont {Barachant}, \citenamefont
  {Li}, \citenamefont {Titericz~Jr}, \citenamefont {Lang}, \citenamefont
  {Lavery} \emph {et~al.}}]{kuhlmann2018epilepsyecosystem}%
  \BibitemOpen
  \bibfield  {author} {\bibinfo {author} {\bibfnamefont {L.}~\bibnamefont
  {Kuhlmann}}, \bibinfo {author} {\bibfnamefont {P.}~\bibnamefont {Karoly}},
  \bibinfo {author} {\bibfnamefont {D.~R.}\ \bibnamefont {Freestone}}, \bibinfo
  {author} {\bibfnamefont {B.~H.}\ \bibnamefont {Brinkmann}}, \bibinfo {author}
  {\bibfnamefont {A.}~\bibnamefont {Temko}}, \bibinfo {author} {\bibfnamefont
  {A.}~\bibnamefont {Barachant}}, \bibinfo {author} {\bibfnamefont
  {F.}~\bibnamefont {Li}}, \bibinfo {author} {\bibfnamefont {G.}~\bibnamefont
  {Titericz~Jr}}, \bibinfo {author} {\bibfnamefont {B.~W.}\ \bibnamefont
  {Lang}}, \bibinfo {author} {\bibfnamefont {D.}~\bibnamefont {Lavery}}, \emph
  {et~al.},\ }\bibfield  {title} {\bibinfo {title} {{Epilepsyecosystem.org:
  crowd-sourcing reproducible seizure prediction with long-term human
  intracranial EEG}},\ }\href@noop {} {\bibfield  {journal} {\bibinfo
  {journal} {Brain}\ }\textbf {\bibinfo {volume} {141}},\ \bibinfo {pages}
  {2619} (\bibinfo {year} {2018})}\BibitemShut {NoStop}%
\bibitem [{\citenamefont {Yuan}\ \emph {et~al.}(2008)\citenamefont {Yuan},
  \citenamefont {Li},\ and\ \citenamefont {Mandic}}]{yuan2008comparison}%
  \BibitemOpen
  \bibfield  {author} {\bibinfo {author} {\bibfnamefont {Y.}~\bibnamefont
  {Yuan}}, \bibinfo {author} {\bibfnamefont {Y.}~\bibnamefont {Li}},\ and\
  \bibinfo {author} {\bibfnamefont {D.~P.}\ \bibnamefont {Mandic}},\ }\bibfield
   {title} {\bibinfo {title} {Comparison analysis of embedding dimension
  between normal and epileptic eeg time series},\ }\href@noop {} {\bibfield
  {journal} {\bibinfo  {journal} {The Journal of Physiological Sciences}\
  }\textbf {\bibinfo {volume} {58}},\ \bibinfo {pages} {239} (\bibinfo {year}
  {2008})}\BibitemShut {NoStop}%
\bibitem [{\citenamefont {El~Houssaini}\ \emph {et~al.}(2020)\citenamefont
  {El~Houssaini}, \citenamefont {Bernard},\ and\ \citenamefont
  {Jirsa}}]{el2020epileptor}%
  \BibitemOpen
  \bibfield  {author} {\bibinfo {author} {\bibfnamefont {K.}~\bibnamefont
  {El~Houssaini}}, \bibinfo {author} {\bibfnamefont {C.}~\bibnamefont
  {Bernard}},\ and\ \bibinfo {author} {\bibfnamefont {V.~K.}\ \bibnamefont
  {Jirsa}},\ }\bibfield  {title} {\bibinfo {title} {The epileptor model: a
  systematic mathematical analysis linked to the dynamics of seizures,
  refractory status epilepticus, and depolarization block},\ }\href@noop {}
  {\bibfield  {journal} {\bibinfo  {journal} {{eNeuro}}\ }\textbf {\bibinfo
  {volume} {7}} (\bibinfo {year} {2020})}\BibitemShut {NoStop}%
\bibitem [{\citenamefont {Scheffer}\ \emph {et~al.}(2009)\citenamefont
  {Scheffer}, \citenamefont {Bascompte}, \citenamefont {Brock}, \citenamefont
  {Brovkin}, \citenamefont {Carpenter}, \citenamefont {Dakos}, \citenamefont
  {Held}, \citenamefont {{Van Nes}}, \citenamefont {Rietkerk},\ and\
  \citenamefont {Sugihara}}]{Scheffer2009}%
  \BibitemOpen
  \bibfield  {author} {\bibinfo {author} {\bibfnamefont {M.}~\bibnamefont
  {Scheffer}}, \bibinfo {author} {\bibfnamefont {J.}~\bibnamefont {Bascompte}},
  \bibinfo {author} {\bibfnamefont {W.~A.}\ \bibnamefont {Brock}}, \bibinfo
  {author} {\bibfnamefont {V.}~\bibnamefont {Brovkin}}, \bibinfo {author}
  {\bibfnamefont {S.~R.}\ \bibnamefont {Carpenter}}, \bibinfo {author}
  {\bibfnamefont {V.}~\bibnamefont {Dakos}}, \bibinfo {author} {\bibfnamefont
  {H.}~\bibnamefont {Held}}, \bibinfo {author} {\bibfnamefont {E.~H.}\
  \bibnamefont {{Van Nes}}}, \bibinfo {author} {\bibfnamefont {M.}~\bibnamefont
  {Rietkerk}},\ and\ \bibinfo {author} {\bibfnamefont {G.}~\bibnamefont
  {Sugihara}},\ }\bibfield  {title} {\bibinfo {title} {{Early-warning signals
  for critical transitions}},\ }\href@noop {} {\bibfield  {journal} {\bibinfo
  {journal} {Nature}\ }\textbf {\bibinfo {volume} {461}},\ \bibinfo {pages}
  {53} (\bibinfo {year} {2009})}\BibitemShut {NoStop}%
\bibitem [{\citenamefont {Karoly}\ \emph {et~al.}(2018)\citenamefont {Karoly},
  \citenamefont {Cook}, \citenamefont {Kuhlmann}, \citenamefont {Freestone},
  \citenamefont {Grayden}, \citenamefont {Nurse}, \citenamefont {Lai},
  \citenamefont {Payne}, \citenamefont {D'Souza}, \citenamefont {Seneviratne},
  \citenamefont {Berkovic}, \citenamefont {O'Brien}, \citenamefont {Litt},
  \citenamefont {Himes}, \citenamefont {Leyde}, \citenamefont {Soudry},
  \citenamefont {Ahmadizadeh}, \citenamefont {Maturana},\ and\ \citenamefont
  {Dell}}]{KAROLY2018}%
  \BibitemOpen
  \bibfield  {author} {\bibinfo {author} {\bibfnamefont {P.}~\bibnamefont
  {Karoly}}, \bibinfo {author} {\bibfnamefont {M.}~\bibnamefont {Cook}},
  \bibinfo {author} {\bibfnamefont {L.}~\bibnamefont {Kuhlmann}}, \bibinfo
  {author} {\bibfnamefont {D.}~\bibnamefont {Freestone}}, \bibinfo {author}
  {\bibfnamefont {D.}~\bibnamefont {Grayden}}, \bibinfo {author} {\bibfnamefont
  {E.}~\bibnamefont {Nurse}}, \bibinfo {author} {\bibfnamefont
  {A.}~\bibnamefont {Lai}}, \bibinfo {author} {\bibfnamefont {D.}~\bibnamefont
  {Payne}}, \bibinfo {author} {\bibfnamefont {W.}~\bibnamefont {D'Souza}},
  \bibinfo {author} {\bibfnamefont {U.}~\bibnamefont {Seneviratne}}, \bibinfo
  {author} {\bibfnamefont {S.}~\bibnamefont {Berkovic}}, \bibinfo {author}
  {\bibfnamefont {T.}~\bibnamefont {O'Brien}}, \bibinfo {author} {\bibfnamefont
  {B.}~\bibnamefont {Litt}}, \bibinfo {author} {\bibfnamefont {D.}~\bibnamefont
  {Himes}}, \bibinfo {author} {\bibfnamefont {K.}~\bibnamefont {Leyde}},
  \bibinfo {author} {\bibfnamefont {D.}~\bibnamefont {Soudry}}, \bibinfo
  {author} {\bibfnamefont {S.}~\bibnamefont {Ahmadizadeh}}, \bibinfo {author}
  {\bibfnamefont {M.}~\bibnamefont {Maturana}},\ and\ \bibinfo {author}
  {\bibfnamefont {K.}~\bibnamefont {Dell}},\ }\bibfield  {title} {\bibinfo
  {title} {{Melbourne NeuroVista Seizure Prediction Trial}},\ }\bibfield
  {journal} {\bibinfo  {journal} {University of Melbourne. Dataset}\ }\href
  {https://doi.org/10.26188/5b6a999fa2316} {10.26188/5b6a999fa2316} (\bibinfo
  {year} {2018})\BibitemShut {NoStop}%
\bibitem [{\citenamefont {Kuehn}(2011)}]{kuehn2011mathematical}%
  \BibitemOpen
  \bibfield  {author} {\bibinfo {author} {\bibfnamefont {C.}~\bibnamefont
  {Kuehn}},\ }\bibfield  {title} {\bibinfo {title} {{A mathematical framework
  for critical transitions: Bifurcations, fast--slow systems and stochastic
  dynamics}},\ }\href@noop {} {\bibfield  {journal} {\bibinfo  {journal}
  {Physica D: Nonlinear Phenomena}\ }\textbf {\bibinfo {volume} {240}},\
  \bibinfo {pages} {1020} (\bibinfo {year} {2011})}\BibitemShut {NoStop}%
\bibitem [{\citenamefont {Maturana}\ \emph {et~al.}(2020)\citenamefont
  {Maturana}, \citenamefont {Meisel}, \citenamefont {Dell}, \citenamefont
  {Karoly}, \citenamefont {D'Souza}, \citenamefont {Grayden}, \citenamefont
  {Burkitt}, \citenamefont {Jiruska}, \citenamefont {Kudlacek}, \citenamefont
  {Hlinka}, \citenamefont {Cook}, \citenamefont {Kuhlmann},\ and\ \citenamefont
  {Freestone}}]{Maturana2020}%
  \BibitemOpen
  \bibfield  {author} {\bibinfo {author} {\bibfnamefont {M.~I.}\ \bibnamefont
  {Maturana}}, \bibinfo {author} {\bibfnamefont {C.}~\bibnamefont {Meisel}},
  \bibinfo {author} {\bibfnamefont {K.}~\bibnamefont {Dell}}, \bibinfo {author}
  {\bibfnamefont {P.~J.}\ \bibnamefont {Karoly}}, \bibinfo {author}
  {\bibfnamefont {W.}~\bibnamefont {D'Souza}}, \bibinfo {author} {\bibfnamefont
  {D.~B.}\ \bibnamefont {Grayden}}, \bibinfo {author} {\bibfnamefont {A.~N.}\
  \bibnamefont {Burkitt}}, \bibinfo {author} {\bibfnamefont {P.}~\bibnamefont
  {Jiruska}}, \bibinfo {author} {\bibfnamefont {J.}~\bibnamefont {Kudlacek}},
  \bibinfo {author} {\bibfnamefont {J.}~\bibnamefont {Hlinka}}, \bibinfo
  {author} {\bibfnamefont {M.~J.}\ \bibnamefont {Cook}}, \bibinfo {author}
  {\bibfnamefont {L.}~\bibnamefont {Kuhlmann}},\ and\ \bibinfo {author}
  {\bibfnamefont {D.~R.}\ \bibnamefont {Freestone}},\ }\bibfield  {title}
  {\bibinfo {title} {{Critical slowing down as a biomarker for seizure
  susceptibility}},\ }\href@noop {} {\bibfield  {journal} {\bibinfo  {journal}
  {Nature Communications}\ }\textbf {\bibinfo {volume} {11}},\ \bibinfo {pages}
  {2172} (\bibinfo {year} {2020})}\BibitemShut {NoStop}%
\bibitem [{\citenamefont {Berger}\ \emph {et~al.}(2017)\citenamefont {Berger},
  \citenamefont {Reis},\ and\ \citenamefont {Trenn}}]{Berger2017}%
  \BibitemOpen
  \bibfield  {author} {\bibinfo {author} {\bibfnamefont {T.}~\bibnamefont
  {Berger}}, \bibinfo {author} {\bibfnamefont {T.}~\bibnamefont {Reis}},\ and\
  \bibinfo {author} {\bibfnamefont {S.}~\bibnamefont {Trenn}},\ }\bibinfo
  {title} {Observability of linear differential-algebraic systems: A survey},\
  in\ \href@noop {} {\emph {\bibinfo {booktitle} {Surveys in
  Differential-Algebraic Equations IV}}}\ (\bibinfo  {publisher} {Springer
  International Publishing},\ \bibinfo {year} {2017})\ pp.\ \bibinfo {pages}
  {161--219}\BibitemShut {NoStop}%
\bibitem [{\citenamefont {Haragus}\ and\ \citenamefont
  {Iooss}(2010)}]{Haragus2010}%
  \BibitemOpen
  \bibfield  {author} {\bibinfo {author} {\bibfnamefont {M.}~\bibnamefont
  {Haragus}}\ and\ \bibinfo {author} {\bibfnamefont {G.}~\bibnamefont
  {Iooss}},\ }\href@noop {} {\emph {\bibinfo {title} {{Local Bifurcation,
  Center Manifolds and Normal Forms in Infinte-Dimensional Dynamical
  Systems}}}}\ (\bibinfo  {publisher} {Springer Science {\&} Business Media},\
  \bibinfo {year} {2010})\BibitemShut {NoStop}%
\bibitem [{\citenamefont {Tantet}\ \emph
  {et~al.}(2018{\natexlab{a}})\citenamefont {Tantet}, \citenamefont
  {Lucarini},\ and\ \citenamefont {Dijkstra}}]{tantet2018resonances}%
  \BibitemOpen
  \bibfield  {author} {\bibinfo {author} {\bibfnamefont {A.}~\bibnamefont
  {Tantet}}, \bibinfo {author} {\bibfnamefont {V.}~\bibnamefont {Lucarini}},\
  and\ \bibinfo {author} {\bibfnamefont {H.~A.}\ \bibnamefont {Dijkstra}},\
  }\bibfield  {title} {\bibinfo {title} {{Resonances in a chaotic attractor
  crisis of the Lorenz flow}},\ }\href@noop {} {\bibfield  {journal} {\bibinfo
  {journal} {Journal of Statistical Physics}\ }\textbf {\bibinfo {volume}
  {170}},\ \bibinfo {pages} {584} (\bibinfo {year}
  {2018}{\natexlab{a}})}\BibitemShut {NoStop}%
\bibitem [{\citenamefont {Tantet}\ \emph
  {et~al.}(2018{\natexlab{b}})\citenamefont {Tantet}, \citenamefont {Lucarini},
  \citenamefont {Lunkeit},\ and\ \citenamefont {Dijkstra}}]{tantet2018crisis}%
  \BibitemOpen
  \bibfield  {author} {\bibinfo {author} {\bibfnamefont {A.}~\bibnamefont
  {Tantet}}, \bibinfo {author} {\bibfnamefont {V.}~\bibnamefont {Lucarini}},
  \bibinfo {author} {\bibfnamefont {F.}~\bibnamefont {Lunkeit}},\ and\ \bibinfo
  {author} {\bibfnamefont {H.~A.}\ \bibnamefont {Dijkstra}},\ }\bibfield
  {title} {\bibinfo {title} {Crisis of the chaotic attractor of a climate
  model: a transfer operator approach},\ }\href@noop {} {\bibfield  {journal}
  {\bibinfo  {journal} {Nonlinearity}\ }\textbf {\bibinfo {volume} {31}},\
  \bibinfo {pages} {2221} (\bibinfo {year} {2018}{\natexlab{b}})}\BibitemShut
  {NoStop}%
\bibitem [{\citenamefont {Proverbio}\ \emph {et~al.}(2022)\citenamefont
  {Proverbio}, \citenamefont {Montanari}, \citenamefont {Skupin},\ and\
  \citenamefont {Gon{\c{c}}alves}}]{proverbio2022buffering}%
  \BibitemOpen
  \bibfield  {author} {\bibinfo {author} {\bibfnamefont {D.}~\bibnamefont
  {Proverbio}}, \bibinfo {author} {\bibfnamefont {A.~N.}\ \bibnamefont
  {Montanari}}, \bibinfo {author} {\bibfnamefont {A.}~\bibnamefont {Skupin}},\
  and\ \bibinfo {author} {\bibfnamefont {J.}~\bibnamefont {Gon{\c{c}}alves}},\
  }\bibfield  {title} {\bibinfo {title} {Buffering variability in cell
  regulation motifs close to criticality},\ }\href@noop {} {\bibfield
  {journal} {\bibinfo  {journal} {Physical Review E}\ }\textbf {\bibinfo
  {volume} {106}},\ \bibinfo {pages} {L032402} (\bibinfo {year}
  {2022})}\BibitemShut {NoStop}%
\bibitem [{\citenamefont {Sadeghzadeh}\ and\ \citenamefont
  {T\'oth}(2022)}]{Sadeghzadeh2022}%
  \BibitemOpen
  \bibfield  {author} {\bibinfo {author} {\bibfnamefont {A.}~\bibnamefont
  {Sadeghzadeh}}\ and\ \bibinfo {author} {\bibfnamefont {R.}~\bibnamefont
  {T\'oth}},\ }\bibfield  {title} {\bibinfo {title} {Improved embedding of
  nonlinear systems in linear parameter-varying models with polynomial
  dependence},\ }\href@noop {} {\bibfield  {journal} {\bibinfo  {journal} {IEEE
  Transactions on Control Systems Technology (Early Access)}\ ,\ \bibinfo
  {pages} {1}} (\bibinfo {year} {2022})}\BibitemShut {NoStop}%
\bibitem [{\citenamefont {Brunton}\ \emph {et~al.}(2016)\citenamefont
  {Brunton}, \citenamefont {Proctor},\ and\ \citenamefont
  {Kutz}}]{Brunton2016}%
  \BibitemOpen
  \bibfield  {author} {\bibinfo {author} {\bibfnamefont {S.~L.}\ \bibnamefont
  {Brunton}}, \bibinfo {author} {\bibfnamefont {J.~L.}\ \bibnamefont
  {Proctor}},\ and\ \bibinfo {author} {\bibfnamefont {J.~N.}\ \bibnamefont
  {Kutz}},\ }\bibfield  {title} {\bibinfo {title} {{Discovering governing
  equations from data by sparse identification of nonlinear dynamical
  systems}},\ }\href@noop {} {\bibfield  {journal} {\bibinfo  {journal} {PNAS}\
  }\textbf {\bibinfo {volume} {113}},\ \bibinfo {pages} {3932} (\bibinfo {year}
  {2016})}\BibitemShut {NoStop}%
\bibitem [{\citenamefont {Aguirre}\ and\ \citenamefont
  {Letellier}(2011)}]{Aguirre2011}%
  \BibitemOpen
  \bibfield  {author} {\bibinfo {author} {\bibfnamefont {L.~A.}\ \bibnamefont
  {Aguirre}}\ and\ \bibinfo {author} {\bibfnamefont {C.}~\bibnamefont
  {Letellier}},\ }\bibfield  {title} {\bibinfo {title} {{Investigating
  observability properties from data in nonlinear dynamics}},\ }\href@noop {}
  {\bibfield  {journal} {\bibinfo  {journal} {Physical Review E}\ }\textbf
  {\bibinfo {volume} {83}},\ \bibinfo {pages} {066209} (\bibinfo {year}
  {2011})}\BibitemShut {NoStop}%
\end{thebibliography}

\end{document}